\documentclass[11pt]{article}
\usepackage{amssymb,amsmath,amsfonts,amsthm,amstext,amscd,array,cite}
\usepackage{mathrsfs}
\usepackage{hyperref}
\usepackage{pdfsync}
\usepackage{bbm}
\usepackage{bm}
\usepackage[arrow,matrix,curve]{xy}
\usepackage{bbding}
\usepackage{wasysym}

\usepackage{epsf}
\usepackage{epsfig}
\usepackage{wrapfig}
\input{epsf}

\def\one{\mathbbm{1}}

\newcommand{\qen}{\hfill $\vartriangleleft$}

\newcommand{\midwedge}{\text{\Large$\wedge$}}

\def\ii{{\,{\rm i}\,}}
\def\dd{{\rm d}}
\def\DD{{\rm D}}
\def\LL{{\rm L}}

\def\im{{\sf im}}

\def\mfg{{\mathfrak g}}

\def\mfX{{\mathfrak X}}

\def\Lie{{\mathcal L}}

\newcommand{\CCA}{\mathscr{A}}

\newcommand{\CCE}{\mathscr{E}}

\newcommand{\CJ}{\mathcal{J}}

\newcommand{\Thetam}{{\mit\Theta}}

\newcommand{\Pim}{{\mit\Pi}}
\newcommand{\Upm}{{\mit\Upsilon}}
\newcommand{\Lambdam}{{\mit\Lambda}}

\newcommand{\eq}{\begin{equation}}
\newcommand{\eqend}{\end{equation}}
\newcommand{\eqa}{\begin{eqnarray}}
\newcommand{\nonueqa}{\begin{eqnarray*}}
\newcommand{\eqaend}{\end{eqnarray}}
\newcommand{\nonueqaend}{\end{eqnarray*}}

\newcommand{\bma}[1]{\begin{array}{#1}}
\newcommand{\ema}{\end{array}}
\newcommand{\bc}{\begin{center}}
\newcommand{\ec}{\end{center}}

\newcommand{\newsection}{\setcounter{equation}{0}\section}

\newcommand{\complex}{{\mathbb C}} 
\newcommand{\real}{{\mathbb R}} 

\newif\ifold             \oldtrue

\def\e{{\,\rm e}\,}

\def\be{\begin{equation}}
\def\ee{\end{equation}}
\def\bea{\begin{eqnarray}}
\def\eea{\end{eqnarray}}
\def\bd{\begin{displaymath}}
\def\ed{\end{displaymath}}

\newcommand{\beq}{\begin{eqnarray}}
\newcommand{\eeq}{\end{eqnarray}}

\makeatletter
\newdimen\normalarrayskip              
\newdimen\minarrayskip                 
\normalarrayskip\baselineskip
\minarrayskip\jot
\newif\ifold             \oldtrue            
\def\arraymode{\ifold\relax\else\displaystyle\fi} 
\def\@arrayskip{\ifold\baselineskip\z@\lineskip\z@
     \else
     \baselineskip\minarrayskip\lineskip2\minarrayskip\fi}
\def\@arrayclassz{\ifcase \@lastchclass \@acolampacol \or
\@ampacol \or \or \or \@addamp \or
   \@acolampacol \or \@firstampfalse \@acol \fi
\edef\@preamble{\@preamble
  \ifcase \@chnum
     \hfil$\relax\arraymode\@sharp$\hfil
     \or $\relax\arraymode\@sharp$\hfil
     \or \hfil$\relax\arraymode\@sharp$\fi}}
\def\@array[#1]#2{\setbox\@arstrutbox=\hbox{\vrule
     height\arraystretch \ht\strutbox
     depth\arraystretch \dp\strutbox
     width\z@}\@mkpream{#2}\edef\@preamble{\halign \noexpand\@halignto
\bgroup \tabskip\z@ \@arstrut \@preamble \tabskip\z@ \cr}%
\let\@startpbox\@@startpbox \let\@endpbox\@@endpbox
  \if #1t\vtop \else \if#1b\vbox \else \vcenter \fi\fi
  \bgroup \let\par\relax
  \let\@sharp##\let\protect\relax
  \@arrayskip\@preamble}
\makeatother

\def\FF{{\cal F}}

\def\be{\beta}

\newtheorem{lemma}[equation]{Lemma}
\newtheorem{proposition}[equation]{Proposition}
\newtheorem{corollary}[equation]{Corollary}

\theoremstyle{definition}
\newtheorem{definition}[equation]{Definition}
\newtheorem{example}[equation]{Example}
\newtheorem{remark}[equation]{Remark}

\def\ddo{\end{document}}
\usepackage{hyperref}
\hypersetup{colorlinks,%
citecolor=black,%
filecolor=black,%
linkcolor=black,%
urlcolor=black}

\newcommand{\nc}{\newcommand}
\nc{\lb}{\llbracket}
\nc{\rb}{\rrbracket}
\nc{\gl}{\llbracket}
\nc{\gr}{\rrbracket}

\begin{document}

\begin{titlepage}

\begin{flushright}
\small\sf 
EMPG--21--01
\end{flushright}
\normalsize

\begin{center}

\vspace{1cm}

\baselineskip=24pt

{\Large{\bf Symplectic embeddings, homotopy algebras \\ and almost Poisson gauge symmetry}}

\baselineskip=14pt

\vspace{1cm}

{\bf Vladislav G. Kupriyanov}${}^{1}$ \ and \ {\bf Richard
  J. Szabo}${}^{2}$
\\[5mm]
\noindent  ${}^1$ {\it Centro de Matem\'atica, Computa\c{c}\~{a}o e
Cogni\c{c}\~{a}o}\\{\it Universidade de Federal do ABC}\\
{\it Santo Andr\'e, SP, 
Brazil}\\ and {\it 
Tomsk State University, Tomsk, Russia}\\
Email: \ {\tt
    vladislav.kupriyanov@gmail.com}
\\[3mm]
\noindent  ${}^2$ {\it Department of Mathematics, Heriot-Watt University\\ Colin Maclaurin Building,
  Riccarton, Edinburgh EH14 4AS, U.K.}\\ and {\it Maxwell Institute for
Mathematical Sciences, Edinburgh, U.K.} \\ and {\it Higgs Centre
for Theoretical Physics, Edinburgh, U.K.}\\
Email: \ {\tt R.J.Szabo@hw.ac.uk}
\\[30mm]

\end{center}

\begin{abstract}
\baselineskip=12pt
\noindent
We formulate general definitions of semi-classical gauge transformations for noncommutative gauge theories in general backgrounds of string theory, and give novel explicit constructions using techniques based on symplectic embeddings of almost Poisson structures. In the absence of fluxes the gauge symmetries close a Poisson gauge algebra and their action is governed by a $P_\infty$-algebra which we construct explicitly from the symplectic embedding. In curved backgrounds they close a field dependent gauge algebra governed by an $L_\infty$-algebra which is not a $P_\infty$-algebra. Our technique produces new all orders constructions which are significantly simpler compared to previous approaches, and we illustrate its applicability in several examples of interest in noncommutative field theory and gravity. We further show that our symplectic embeddings naturally define a $P_\infty$-structure on the exterior algebra of differential forms on a generic almost Poisson manifold, which generalizes earlier constructions of differential graded Poisson algebras, and suggests a new approach to defining noncommutative gauge theories beyond the gauge sector and the semi-classical limit based on $A_\infty$-algebras.
\end{abstract}

\end{titlepage}
\setcounter{page}{2}

\newpage

{
\tableofcontents
}

\newpage


\newsection{Introduction}

The construction of noncommutative gauge theories on manifolds with non-trivial tensor fields is an important problem for understanding the low-energy physics of D-branes in general backgrounds of string theory. The problem is of course not new, and has been studied for around 20 years now. But despite its relatively long history of investigation, the construction is still not completely understood in full generality. In this paper we propose a new approach to this old problem, where the main mathematical tool employed is known as a `symplectic embedding'. 

In this section we start by providing some background and motivation from string theory, and then proceed to an informal description of our goals and main results, putting them into context with earlier physics literature on the subject. More precise statements and proofs will be given in subsequent sections, with a detailed technical analysis of the issues discussed here.

\paragraph{Symplectic embeddings.}

Consider a D-brane wrapping a submanifold $M$ in a flat background. Then the string equations imply the vanishing of the three-from $H$-flux for the $B$-field: $H=\dd B=0$; when the closed two-form $B$ is non-degenerate its inverse defines a Poisson bivector $\theta$ on $M$. Quantizing an open string with its ends on the D-brane shows that, in a suitable low-energy scaling limit, its worldvolume algebra of functions undergoes a deformation in the direction of $\theta$.
The problem of constructing a noncommutative gauge theory on $M$ then starts with a deformation quantization of the Poisson manifold $(M,\theta)$. Conversely, the semi-classical limit of an associative noncommutative algebra of functions on a manifold $M$ which is a flat deformation defines a Poisson bracket
\begin{align*}
\{x^i,x^j\}_\theta = \theta^{ij} \ ,
\end{align*}
on local coordinate functions $x^i$; in other words, the Poisson bracket is the first order semi-classical approximation to the star-product in deformation quantization. It is this semi-classical limit that we shall mostly work with in this paper, where all constructions are entirely geometric and are regarded as the classical infinitesimal data whose quantization yields the required ingredients of a noncommutative gauge theory. 

In the case of a flat D-brane, the Poisson bivector $\theta$ is constant, and the worldvolume deformation is provided by the usual Moyal-Weyl star-product~\cite{Schomerus:1999ug}, and the construction of noncommutative gauge theories is standard and well-known~\cite{Seiberg:1999vs}. For a curved D-brane in flat space, $\theta$ is not constant, and the worldvolume deformation is provided by the Kontsevich star-product~\cite{Cornalba:2001sm}. The first problem one then encounters is how to define derivatives in the field theory: the usual differential is no longer a derivation of the Poisson algebra, and this obstructs the naive definition of gauge transformations to closing a Lie algebra. The problem can be formulated and solved by embedding the Poisson bracket in a `noncommutative phase space'
\begin{align}\label{eq:NCphasespace}
\{x^i,x^j\} &= \theta^{ij} \ , \notag \\[4pt]
\{p_i,x^j\} &= \delta_i^j - \tfrac12\,\partial_i\theta^{jk}\,p_k + \cdots \ , \notag \\[4pt]
\{p_i,p_j\} &= 0 \ .
\end{align}
The auxiliary `momentum' coordinates $p_i$ are regarded as `derivatives' when acting on functions $f$ on $M$, since when $\theta$ is constant these brackets give $\{p_i,f\}=\partial_if$; the ellipsis denotes higher order monomials in the momenta $p_k$ which accompany higher orders in the bivector $\theta$ and its derivatives. In general, the extra derivative terms in these brackets ensure that they fulfill the Jacobi identity order by order in $p_i$ when $\theta$ is non-constant. Now the action of the `twisted' derivatives $\{p_i,\,\cdot\,\}$ obeys the Leibniz rule, as a consequence of the Jacobi identity for the bracket. Noncommutative gauge theories have been constructed in this manner in e.g.~\cite{Behr:2003qc}; see e.g.~\cite{Szabo:2006wx} for a review of this and other approaches, along with further references.

The Poisson brackets \eqref{eq:NCphasespace} also offer an alternative approach to quantization of the Poisson manifold $(M,\theta)$. We can map the coordinates $(x,p)$ to canonically conjugate variables $(X,P)$ by means of a generalized Bopp shift; the existence of such a diffeomorphism is guaranteed by Darboux's theorem, at least locally or in the case where $M$ is covered by a single Darboux chart. The canonical coordinates $(X,P)$ can be quantized geometrically in the usual way via a Schr\"odinger polarization, and mapped back to provide a polydifferential representation of the quantum version of the brackets \eqref{eq:NCphasespace} on the space of functions on $M$. By formally expanding functions of the polydifferential operators in Taylor series, this also constructs a star-product and provides a deformation quantization of the algebra of functions on $M$.

This construction of a noncommutative phase space is called a \emph{symplectic embedding} of the Poisson structure $\theta$; in Section~\ref{sec:Examples} we will see some explicit all orders examples which are given by closed analytic expressions. The brackets can be derived as the symplectic structure arising from open string quantization on the D-brane in a low-energy scaling limit~\cite{Chu:1999wz}. In particular, it enables one to define semi-classical gauge transformations that consistently close a gauge algebra defined by the Poisson brackets: we use `twisted covariant derivatives' of gauge parameters defined by the $p_i$ and suitably restrict them by imposing constraints on the phase space that eliminate the auxiliary momentum coordinates. This is explained in detail in Section~\ref{sec:Poissongauge}.

Mathematically, symplectic embeddings are a more general notion of `symplectic realizations' in Poisson geometry~\cite{Weinstein-local,Karasev}. We discuss the precise relation between the two notions in Section~\ref{sec:symplemb}. 
The global structure which integrates a Poisson manifold is called a symplectic groupoid~\cite{Weinstein-groupoid}, whose source maps always define symplectic realizations. They were introduced as part of the program of quantizing Poisson manifolds. Symplectic groupoids capture the semi-classical limit of deformation quantization, as inferred by our physical arguments above, and as shown precisely in~\cite{Cattaneo:2000iw} where the symplectic groupoid is constructed as the classical phase space of the open string sigma-model. 

\paragraph{Quantized differential forms.}

The problem of finding suitable derivative operators in a noncommutative field theory on a Poisson manifold $(M,\theta)$ can also be formulated dually as the problem of constructing a suitable noncommutative differential calculus on $M$. In the semi-classical limit, this amounts to finding an extension of the Poisson bracket to the exterior algebra of differential forms. The problem of constructing a differential graded Poisson algebra in this way has been addressed from several points of view, see e.g.~\cite{Chu:1997ik,Hawkins:2002rf,Beggs:2003ne,McCurdy:2009xz,Voronov,Lyakhovich}, while from the perspective of deformation quantization it is discussed in e.g.~\cite{Bieliavsky2006,Gutt2008,Dolgushev}. These constructions all depend on the choice of an auxiliary connection on $M$ and typically require the Poisson bivector $\theta$ to be invertible. 

Noncommutative gauge theories are treated using this formalism in~\cite{Ho:2001fi}. Their definition of gauge transformations is similar in structure to ours in Section~\ref{sec:Poissongauge}, but differs in important details. A crucial distinction is that our formulation is based on symplectic embeddings, which always exist locally without any further data and applies to arbitrary Poisson structures, while the formulation of~\cite{Ho:2001fi} requires the extra data of a symplectic connection on $M$.

\paragraph{Nonassociative deformation quantization.}

The problem is more involved in the case of a curved background, wherein the $H$-flux is non-zero: $H=\dd B\neq0$. In this case the Jacobi identity for the semi-classical bracket is violated, and the bivector $\theta$ defines an $H$-twisted Poisson structure~\cite{Park,Klimcik}. The worldvolume deformations of D-branes in this setting are still given by the Kontsevich expansion~\cite{Cornalba:2001sm}, which now leads to a nonassociative star-product. In this case, the symplectic embedding is described by a noncommutative (but associative) phase space of the form
\begin{align*}
\{x^i,x^j\} &= \theta^{ij} - \Pim^{ijk}\, p_k + \cdots \ , \\[4pt]
\{p_i,x^j\} &= \delta_i^j - \tfrac12\,\partial_i\theta^{jk}\,p_k + \cdots \ , \\[4pt]
\{p_i,p_j\} &= 0 \ ,
\end{align*}
where
\begin{align*}
\Pim^{ijk} = \tfrac13\,\big(\theta^{il}\partial_l\theta^{jk}+{\rm cyclic}\big)= \theta^{il}\,\theta^{jm}\,\theta^{kn}\,H_{lmn}
\end{align*}
is the Jacobiator $\{\{x^i,x^j\}_\theta,x^k\}_\theta+{\rm cyclic}$, the trivector encoding the violation of the Jacobi identity for the original twisted Poisson brackets $\{x^i,x^j\}_\theta=\theta^{ij}$; its inclusion is now necessary to ensure that these extended brackets fulfill the Jacobi identity. These brackets agree precisely with the symplectic structure found in~\cite{Ho:2000fv,Ho:2001qk} from quantizing an open string with its ends on a D-brane in a constant $H$-flux background. 

Regarding again the momenta as `derivatives', now the associative phase space structure is that of an algebra of pseudo-differential operators, whereby even the bracket of two functions on $M$ is generally a differential operator. This was explained by~\cite{KS18}, who show that the symplectic embedding in this case captures the semi-classical limit of the associative composition algebra of differential operators on $M$ corresponding to a twisted Poisson structure~\cite{MSS1}. In D-brane physics this poses a serious problem in the definition of gauge theories, as the gauge transformation of a gauge field will generically produce not another gauge field but a pseudo-differential operator;  a resolution was proposed by~\cite{Ho:2001qk} which involves promoting gauge parameters to functions of both coordinates $x^i$ and momenta $p_i$. In the present paper we propose a different way of resolving this issue, which is similar in spirit, but involves a suitable constraint on the phase space as in the case of a Poisson structure as well as field dependent gauge transformations. The precise definitions and constructions are explained in Section~\ref{sec:quasiP}.

This extended symplectic embedding formalism was initiated in~\cite{KS18} to describe the (associative) classical and quantum mechanics of an electric charge in smooth distributions of magnetic monopoles, whose phase space is described by a twisted Poisson structure. The main idea behind our construction of the ``nonassociative'' gauge algebra consists in starting with a given almost Poisson structure, and aiming to get the violation of the Jacobi identity under control. For this, we introduce auxiliary variables $p_i$ and construct its symplectic embedding. In the larger space the Jacobi identity is now satisfied, but one is then faced with the problem of interpreting the auxiliary (unphysical) degrees of freedom, which cannot be removed in the nonassociative case~\cite{KS18}. The most non-trivial point of our construction is getting rid of the auxiliary variables by introducing constraints in a consistent way, that is, in such a way that, in the ``nonassociative'' gauge algebra, the commutator of two gauge transformations is again a gauge transformation, but with a field dependent gauge parameter. 

The global objects corresponding to symplectic embeddings of twisted Poisson or more generally almost Poisson structures are not presently understood. From a string theory perspective, they can be described as follows. A fundamental closed string in a constant $H$-flux background can be lifted to M-theory as an open membrane ending on an M5-brane in a three-form $C$-field background with constant flux. This gives rise to a noncommutative algebra of functions on the loop space of the M5-brane worldvolume~\cite{Bergshoeff:2000jn}, whose deformation quantization has as semi-classical limit a symplectic groupoid on the loop space obtained through transgression of the three-form $H$~\cite{Saemann:2012ex,Saemann:2012ab}. In this sense symplectic embeddings play a similar role in the semi-classical limit of deformation quantization of almost Poisson structures. We give a more precise account in Section~\ref{sec:defquant}.

\paragraph{Strong homotopy algebras.}

Recent advances in nonassociative deformation quantization have produced some explicit constructions of star-products which quantize twisted Poisson and even almost Poisson structures. A star-product which quantizes the phase space of a closed string in a constant $R$-flux background, or equivalently an electric charge moving in a uniform magnetic monopole density, was constructed in~\cite{Mylonas2012}; the phase space defines a twisted Poisson structure and the star-product is, in a sense, a nonassociative analog of the Moyal-Weyl star-product. In~\cite{Kupriyanov:2016hsm} a star-product was proposed which quantizes the linear almost Poisson structure of the imaginary octonion algebra, which is an example of an almost Poisson structure that is not twisted Poisson; in~\cite{KS17} this was applied to the quantization of the phase space of a membrane in a non-geometric $R$-flux background of M-theory, or equivalently an M-wave in a non-geometric Kaluza-Klein monopole density, where it was shown that the contraction of the octonionic star-product reproduces the monopole star-product. See~\cite{Szabo:2019hhg} for a review of these developments and further references.

These examples have nice physical features. In particular, nonassociativity vanishes on-shell, as it should since string theory is based on a two-dimensional quantum field theory that involves strictly associative structures. Precisely, the integrated associator vanishes for these star-products, as the star-associator of three functions is a `total derivative'; this is also true for the nonassociative star-products of~\cite{Cornalba:2001sm}, provided one chooses a suitable density to integrate against (see e.g.~\cite{Szabo:2006wx}). However, this property fails in general for arbitrary almost Poisson structures, and the question arises as to what structure should be used to control the violation of associativity, or the Jacobi identity in the semi-classical limit, in a consistent way. A good candidate is Stasheff's $A_\infty$-algebras~\cite{Stasheff1963}, and also $L_\infty$-algebras~\cite{Lada:1992wc} where the violation of the Jacobi identity is proportional to a higher coherent homotopy or a `total derivative' of some higher bracket, together with their extension to $P_\infty$-algebras~\cite{Voronov2003}. This suggests that strong homotopy algebras might provide the right setting for the quantization of generic twisted Poisson and almost Poisson structures; their role in deformation quantization of twisted Poisson structures was already indicated in~\cite{Cornalba:2001sm,Mylonas2012}. In particular, the semi-classical limit of an $A_\infty$-algebra defines a $P_\infty$-algebra, and it is natural to look for $P_\infty$-algebras as the structure controlling the violation of the Jacobi identity. The deformation quantization of $P_\infty$-structures via the Kontsevich formality theorem was first considered in~\cite{Lyakhovich:2004xd}, and a bit later in~\cite{Cattaneo:2005zz}.

As we discuss in Section~\ref{sec:Pinfty}, this expectation is indeed correct from the following perspective. We will show that our symplectic embeddings, which always exist locally, naturally endow the exterior algebra of differential forms on an arbitrary almost Poisson manifold $(M,\theta)$ with the structure of a $P_\infty$-algebra which contains the almost Poisson bracket. The $P_\infty$-structure controls not only the potential violation of the Jacobi identity for the almost Poisson bracket, but also the violation of the Leibniz rule for the usual differential, in a way which is compatible with the derivation properties of the original bracket with respect to the classical pointwise multiplication of functions, extended to the exterior product of differential forms. We subsequently sketch an alternative new approach to the quantization of differential forms which does not require the choice of an auxiliary connection on $M$ and quantizes this $P_\infty$-algebra to an $A_\infty$-algebra. The details are beyond the scope of the present paper and are left for future investigation, but they illustrate another novel application of our symplectic embedding formalism, as well as its intimate relation with homotopy algebras.

The use of $L_\infty$-algebras as an alternative new construction of noncommutative gauge theories on arbitrary almost Poisson manifolds $(M,\theta)$ was pioneered by~\cite{BBKL} and called the `$L_\infty$-bootstrap'. In Section~\ref{sec:Linfinity} we provide the dictionary between our symplectic embedding formalism and the formulation in terms of $L_\infty$-algebras for noncommutative gauge transformations, finding agreement with the lower degree brackets that were constructed explicitly by~\cite{BBKL,Kupriyanov:2019ezf}. However, our approach is much better adapted to obtaining explicit closed all orders expressions for all brackets of the gauge $L_\infty$-algebra. We illustrate this explicitly on several concrete non-trivial examples of relevance to physics in Section~\ref{sec:Examples}, wherein we obtain the complete sets of $L_\infty$-brackets to all orders. In particular, for the twisted Poisson structure on the phase space of an electric charge in a constant magnetic monopole density, we obtain, for the first time, the complete gauge $L_\infty$-structure in closed form, see Section~\ref{sec:magneticPoisson}; this should aid in understanding the symmetries of a nonassociative theory of gravity underlying the low-energy sector of closed non-geometric strings (see e.g.~\cite{Szabo:2019hhg}). Technically, our approach based on symplectic embeddings is much simpler than the $L_\infty$-bootstrap approach of~\cite{BBKL}, particularly for explicit calculations.

In this language we also find a significant distinction between Poisson and almost Poisson gauge symmetries. For a Poisson manifold $(M,\theta)$ the gauge $L_\infty$-algebra is itself a $P_\infty$-algebra which is obtained by a truncation of the $P_\infty$-structure underlying the semi-classical limit of quantization of differential forms. This suggests a path towards the construction of noncommutative gauge transformations beyond the semi-classical limit, in terms of $A_\infty$-algebras. However, this is not the case for the homotopy algebras which govern the closure of ``nonassociative" gauge transformations: our construction of gauge transformations defines an $L_\infty$-algebra which is \emph{not} a $P_\infty$-algebra, unless the Jacobi identity is satisfied, that is, only ``associative" gauge algebras are controlled by $P_\infty$-algebras. The reasons for this are explained in detail in Section~\ref{sec:Pinfty}: the essential feature is that for an almost Poisson structure the $P_\infty$-structure on differential forms cannot be truncated to a subalgebra containing the gauge $L_\infty$-algebra.

An alternative approach to the construction of noncommutative gauge theories using homotopy algebras has been put forth recently by~\cite{Ciric:2020eab}, whereby the underlying $L_\infty$-structure is also quantized. In contrast to our approach, the gauge theory $L_\infty$-algebra involves only finitely many brackets as in the classical case, and compatibility with the Leibniz rule is manifest from the outset. It would be interesting to understand better how this more algebraic approach is related to the geometric approach of the present paper.

\paragraph{Structure of the paper.}

In this paper we work solely with the kinematical sector of a complete noncommutative gauge theory, and only with the first order semi-classical approximations to proper noncommutative gauge transformations. This will elucidate, in a model independent way, the classical geometric structures underlying the full gauge algebras which quantize them. In the absence of any dynamics, we will require that our gauge transformations close to a Lie algebra. One of the main technical achievements of this paper is the ability to do this in the nonassociative case: the naive passage from Poisson to general almost Poisson gauge transformations do not immediately close, even when one modifies the requirement of closure to include field dependent gauge transformations. We show explicitly how closure can be accomplished using symplectic embeddings by the addition of terms to the gauge variations which compensate the undesired contributions to the closure formulas; we demonstrate that these extra terms are explicitly calculable. When the dynamical sector of a particular gauge theory is included, it may be possible to avoid adding these additional terms in the gauge variations and work instead with a weaker notion of closure, whereby the almost Poisson gauge algebra also involves the field equations, that is, gauge transformations only close on-shell, see Remark~\ref{rem:AinftyBBKL}. Barring these model dependent extensions, it is remarkable that we are able to describe the closure of ``nonassociative'' gauge transformations in terms of strictly associative Lie algebras, albeit field dependent ones.

The majority of this paper can be read by taking $M$ to be a coordinate manifold $\real^d$ (or an open subset thereof), but in several places we have written tensor quantities in a more covariant form that we hope makes some of our calculations amenable to global generalizations in future work. Our main purpose here is to unravel the geometric and algebraic structures that govern the first steps to constructing noncommutative gauge theories on D-branes in general string backgrounds. This turns out to be a technically formidable task even with our restriction to local coordinate charts. Therefore, throughout we shall carefully formulate our geometric framework and then carry out all calculations via tensor calculus in local coordinates, leaving the corresponding description of the covariant tensor calculus within our framework for future study.

The outline of the remainder of this paper is as follows. In Section~\ref{sec:symplemb} we give a precise description of our symplectic embeddings which we motivated from physical arguments above, and compare it to the existing notions of symplectic realizations and symplectic groupoids from the mathematical literature on Poisson geometry. In Section~\ref{sec:Poissongauge} we give a precise general definition of what we mean by semi-classical gauge transformations on a Poisson manifold, and provide an explicit construction using symplectic embeddings of the Poisson structure. In Section~\ref{sec:quasiP} we extend these general definitions and constructions to the general case of almost Poisson manifolds. The closure condition on the semi-classical gauge algebra now requires an extension of the notion of gauge transformations using field dependent gauge parameters, which we formulate precisely in terms of $1$-jets of the cotangent bundle of the underlying manifold $M$, and an extension of their construction by symplectic embeddings using horizontal vector fields on the jet space. In Section~\ref{sec:Linfinity} we show how our symplectic embedding construction of gauge algebras can be cast into the standard framework of $L_\infty$-algebras~\cite{Fulp:2002kk,Hohm:2017pnh,Jurco:2018sby}, contrasting our approach with the approach based on the $L_\infty$-bootstrap. We further show how our symplectic embeddings define $P_\infty$-structures on differential forms and discuss their relation with the $L_\infty$-structures on our gauge algebras. Finally, in Section~\ref{sec:Examples} we work out several explicit examples as illustrations of our formalism for both Poisson and almost Poisson gauge transformations. Two appendices at the end of the paper collect some technical details which are used in the main text.

\paragraph{Acknowledgements.}

We are grateful to Alex Schenkel, Jim Stasheff, Dima Vassilevich and Alan Weinstein for helpful
discussions and correspondence. V.G.K. acknowledges the support from the S\~ao Paulo Research Foundation (FAPESP), grant 2021/09313-8. The work of {\sc R.J.S.} was supported in part by
the Consolidated Grant ST/P000363/1 ``Particle Theory at the Higgs Centre''
from the UK Science and Technology Facilities Council (STFC).

\newsection{Symplectic embeddings of almost Poisson structures}
\label{sec:symplemb}

In this section we develop our notion of symplectic embeddings, which encompasses earlier constructions from the physics literature in a rigorous and precise way. We provide a detailed comparison between our symplectic embeddings and the more common concept of symplectic realizations in Poisson geometry, and discuss their relation to deformation quantization. In Section~\ref{sec:Examples} we will give detailed applications to some non-trivial explicit examples.

\subsection{Formal deformations of cotangent bundles}

We begin with some definitions that will permeate this paper. Smooth
functions and tensors on a manifold $M$, as well as all vector spaces,
are always considered with the ground field $\real$ or $\complex$. 

\begin{definition}
\label{def:quasiPoisson}
Let $M$ be a manifold.
\begin{itemize}
\item[(a)]
An {\emph{almost Poisson structure}} on $M$ is a bivector field
$\theta\in\mathfrak{X}^2(M)$. The bivector field is equivalently
encoded by the skew-symmetric {\emph{almost Poisson bracket}} $\{\,\cdot\,,\,\cdot\,\}_\theta$ on
$C^\infty(M)$ given by 
$$
\{f,g\}_\theta=\theta(\dd f,\dd g) \ .
$$
The Schouten-Nijenjuis bracket of the bivector $\theta$ with itself,
the \emph{Jacobiator}, is the trivector field
$\Pim=[\theta,\theta]\in\mathfrak{X}^3(M)$, which satisfies
\begin{align*}
\tfrac12\,\Pim(\dd f,\dd g,\dd h)={\sf
  Cyc}_{f,g,h}\,\{f,\{g,h\}_\theta\}_\theta \ ,
\end{align*}
where ${\sf Cyc}_{f,g,h}$ indicates
the cyclic sum over the functions $f,g,h\in
C^\infty(M)$. The Jacobi identity $[[\theta,\theta],\theta]=0$ for the
Schouten bracket implies that the Jacobiator obeys the
integrability condition
\begin{align}\label{eq:Piint}
[\Pim,\theta] = 0 \ .
\end{align}

\item[(b)]
A bivector $\theta$ on $M$ gives rise to a linear map
\bea\label{eq:Thetasharp}
\theta^\sharp: \Omega^1(M)\longrightarrow\mfX(M) \ , \quad
\alpha\longmapsto \theta(\alpha,\,\cdot\,) \ .
\eea
If 
$$
\Pim=\tfrac12\,\midwedge^3\,\theta^\sharp H
$$ 
for a closed three-form $H\in\Omega^3(M)$,
then $\theta$ is an \emph{$H$-twisted Poisson structure} on~$M$ and
$\{\,\cdot\,,\,\cdot\,\}_\theta$ is the corresponding
\emph{$H$-twisted Poisson
bracket}.

\item[(c)]
The bracket $\{\,\cdot\,,\,\cdot\,\}_\theta$ obeys the Jacobi identity
if and only if $\Pim=0$. In this case $\theta$ is a \emph{Poisson structure}
and $\{\,\cdot\,,\,\cdot\,\}_\theta$ is the corresponding \emph{Poisson
bracket}.
\end{itemize}
\end{definition}

\begin{remark}
By definition, any almost Poisson structure
defines a derivation $\{f,\,\cdot\,\}_\theta$ of the commutative
algebra of functions $C^\infty(M)$ for any fixed $f\in C^\infty(M)$.
When $\theta$ is a Poisson bivector, this property makes
$(C^\infty(M),\{\,\cdot\,,\,\cdot\,\}_\theta)$ into a \emph{Poisson
  algebra}.
\qen\end{remark}

\begin{remark}
If an almost Poisson structure $\theta$ is non-degenerate, then it is
automatically a twisted Poisson structure: in this case the map
\eqref{eq:Thetasharp} is invertible, and the two-form
$\theta^{-1}\in\Omega^2(M)$, defined by
$\theta^{-1}(X,Y)=\theta(\theta^\sharp{}^{-1}X,\theta^\sharp{}^{-1}Y)$
for $X,Y\in\mfX(M)$, is an almost symplectic structure on $M$
and $H=\dd\theta^{-1}$ is the twisting three-form.
In particular, if $\theta$ is a non-degenerate Poisson bivector then $\theta^{-1}$
is a symplectic structure on $M$. 
\qen
\end{remark}

We shall often work on an open subset $M=U\subset\real^d$ with local
coordinates $(x^i)$. Then the bivector $\theta$ has the local coordinate expression
$\theta=\frac12\,\theta^{ij}(x)\,\partial_i\wedge\partial_j$, where
$\partial_i=\partial/\partial x^i$, and the almost Poisson
bracket can be expressed on $U$ as
$$
\{f,g\}_\theta = \theta^{ij}\,\partial_if\,\partial_jg \ .
$$
Throughout we use the Einstein convention for implicit summation
over repeated upper and lower indices. Similarly, the trivector
$\Pim=[\theta,\theta]$ has the local coordinate expression
$\Pim=\frac1{3!}\,\Pim^{ijk}(x)\,\partial_i\wedge\partial_j\wedge\partial_k$,
where
$$
\Pim^{ijk} = \tfrac13\,\big(
\theta^{il}\, \partial_l\theta^{jk}+\theta^{kl}\, \partial_l\theta^{ij}+\theta^{jl}\, \partial_l\theta^{ki} 
\big) \ .
$$

In this paper it will be convenient to rescale the bivector $\theta$
by a formal parameter $t$, and regard it as a 
deformation of the zero bivector $\theta_0=0$. If $V$ is a real or complex vector space, we
denote by $V[[t]]$ the space of formal power series in $t$ with
coefficients in $V$; it can be regarded as a module over
$\real[[t]]$ or $\complex[[t]]$. In this paper we will work in the
setting of formal Poisson
structures and formal symplectic structures, see e.g.~\cite{BOW20}.

\begin{definition}
\label{def:symplembed}
Let $\theta$ be an almost Poisson structure on a manifold $M$. Write
$\pi:T^*M\to M$ for the cotangent bundle of $M$. A {\emph{symplectic
    embedding}} of $(M,\theta)$ is a formal symplectic structure
$\omega\in\Omega_{\rm pol}^2(T^*M)[[t]]$ such
that
\begin{itemize}
  \item[(a)] $\omega$ is a deformation of the canonical symplectic
    structure $\omega_0$ on $T^*M$, that is, $\omega|_{t=0}=\omega_0$; 
  \item[(b)] The corresponding Poisson bracket
$\{\,\cdot\,,\,\cdot\,\}_{\omega^{-1}}$ on the ring $C_{\rm
  pol}^\infty(T^*M)[[t]]$ restricts to the zero section $M\subset T^*M$ as
\bea\label{eq:zetapicond}
\{\pi^*f,\pi^*g\}_{\omega^{-1}}\big|_M = t\,\{f,g\}_{\theta}
\eea
for all functions $f,g\in C^\infty(M)$; and
\item[(c)] The zero section is a Lagrangian section of
  $(T^*M,\omega)$. 
\end{itemize}
\end{definition}

Here and in the following we use the subscript ${}_{\rm pol}$ to
denote spaces of smooth tensor fields on the cotangent bundle $T^*M$ which are polynomial on the fibres. 
In general, the formal symplectic structure $\omega$
is a formal power series of closed two-forms on $T^*M$ (which are
polynomial on the fibres) given as
$\omega=\omega_0+\sum_{n=1}^\infty\, t^n\,\omega_n$. The inverse is a formal
Poisson structure $\omega^{-1}\in\mfX_{\rm pol}^2(T^*M)[[t]]$ which is a
formal power series of bivectors on $T^*M$ given as
$$
\omega^{-1} = \Thetam_0 + \sum_{n=1}^\infty \, t^n\,\Thetam_n \ ,
$$
where $\Thetam_0=\omega_0^{-1}$ is the Poisson bivector field on $T^*M$
corresponding to $\omega_0$. The Poisson integrability condition
$[\omega^{-1},\omega^{-1}]=0$ implies
\bea\label{eq:Omegamint}
\sum_{k=0}^n\, [\Thetam_k,\Thetam_{n-k}] = 0
\eea
for all $n\geq0$, and the condition \eqref{eq:zetapicond} implies 
$\omega^{-1}(\pi^*\alpha,\pi^*\beta)\big|_M =
t\,\theta(\alpha,\beta)$ which leads to 
$$
\Thetam_1(\pi^*\alpha,\pi^*\beta)\big|_M=\theta(\alpha,\beta)
\qquad \mbox{and} \qquad \Thetam_n(\pi^*\alpha,\pi^*\beta)\big|_M=0 \ ,
$$
for all $\alpha,\beta\in\Omega^1(M)$ and $n\geq2$. The Lagrangian
section condition further imposes $\omega(X,Y)=0$ for all $X,Y\in\mfX(M)$. In this sense a
symplectic embedding can be thought of as a formal deformation of an
almost Poisson structure $\theta$ on $M$, for which
$[\theta,\theta]\neq0$, to a nondegenerate Poisson structure $\omega^{-1}$ on $T^*M$,
for which $[\omega^{-1},\omega^{-1}]=0$; in particular, the canonical
symplectic structure $\omega_0$ yields a symplectic embedding for
the trivial bivector $\theta_0=0$.

\begin{remark}
Definition~\ref{def:symplembed} is an adapted version, to the
almost Poisson case, of a symplectic realization of a Poisson structure
$\theta$ on $M$, which more generally involves a surjective submersion
$\pi:S\to M$ of a symplectic manifold $S$ that is a Poisson morphism
admitting a Lagrangian section $M\to S$. When $S=T^*M$ with $\pi$ the
cotangent bundle projection and $M\to S$ the zero section, a symplectic
realization is a symplectic embedding, i.e.~the condition
\eqref{eq:zetapicond} holds, but the converse is not generally true. Similarly, one defines an almost symplectic realization of
a twisted Poisson structure. These can all be constructed globally in terms
of integrating symplectic
groupoids for $(M,\theta)$~\cite{Weinstein-groupoid,Karasev,Cattaneo}. In this paper we
deal only with local constructions, where the existence of symplectic
embeddings is always guaranteed as a formal deformation of the trivial
symplectic groupoid $(T^*M,\omega_0)$ near the identity elements. 
\qen
\end{remark}

We shall now establish the existence of local symplectic embeddings by
treating the three cases in Definition~\ref{def:quasiPoisson}
individually, in order of increasing complexity.
In the following we work on the open
neighbourhood $T^*U$ of the zero section of $M=U\subseteq\real^d$, and denote local
coordinates thereon by $(x^i,p_i)$, where $(p_i)$ are coordinates in the
normal directions to $U\subset T^*U$; in particular, $U$ is given by
the equations $p_i=0$ in $T^*U$. We also write
$\tilde\partial{}^i=\partial/\partial p_i$.

\subsection{Local symplectic embedding of Poisson manifolds}
\label{sec:Poissonembedding}

The simplest case is the original local
construction of a symplectic realization of a Poisson manifold, which is
due to Weinstein~\cite{Weinstein-local}. Let $\lambda_0$
be the Liouville one-form, i.e. the unique one-form on $T^*M$ with the property
that $s_\alpha^*\lambda_0=\alpha$ for all one-forms
$\alpha\in\Omega^1(M)$, regarded as smooth sections $s_\alpha:M\to T^*M$ of the
cotangent bundle of $M$ under the isomorphism $\Omega^1(M)\simeq\Gamma(T^*M)$. It is the tautological 
primitive for the canonical symplectic structure; in local
coordinates, $\lambda_0 =p_i\,\dd x^i$ and
$\omega_0=\dd\lambda_0=\dd p_i\wedge\dd x^i$ on $T^*U$. Contracting
the pushforward of the Poisson bivector field $\theta$ by the zero
section with the Liouville
one-form defines a vector field $X^\theta\in\mfX(T^*U)$:
$$
X^\theta = \theta(\lambda_0,\,\cdot\,) \ ,
$$
which in local coordinates reads
\bea\label{eq:Xtheta}
X^\theta = \theta^{ij}(x)\,p_i\,\partial_j \ .
\eea
The vector field $X^\theta$ defines a Poisson spray, see e.g.~\cite{Crainic}.
Let $\varphi^{\theta}_u:T^*U\to T^*U$ be the
flow of $t\,X^\theta$ for $u\in[0,1]$, which is the diffeomorphism defined by
$$
\frac{\dd\varphi_u^{\theta}}{\dd u} = t\,X^{\theta}\circ\varphi_u^{\theta}
\ .
$$

Then
the symplectic structure $\omega\in\Omega^2(T^*U)$ constructed
by~\cite[Theorem~9.1]{Weinstein-local} is given by the integrated
pullback of the canonical symplectic structure $\omega_0$ by this
flow:
\bea\label{eq:omegaintdef}
\omega := \int_0^1\, \big(\varphi_u^{\theta}\big)^* \omega_0 \ \dd u \ .
\eea
Since $\varphi^\theta_u\big|_{t=0}$ is the identity for all $u\in[0,1]$, this
symplectic structure is indeed a deformation of~$\omega_0$. 
Note that the zero section $U\subset T^*U$ is a Lagrangian
submanifold of $(T^*U,\omega)$, and that $\omega=\dd\lambda$ where $\lambda = \phi^i(x,p)\,\dd p_i$
with 
\bea\label{eq:phiixp}
\phi^i(x,p) = \int_0^1\, x^i\circ\varphi_u^{\theta} \ \dd u \ .
\eea
The Jacobian matrix ${\sf J}_\phi(x,p)=\big(\frac{\partial\phi^i}{\partial x^j}\big)$ is formally 
invertible (because $\varphi_u^\theta\big|_{t=0}$ is the identity for all $u\in[0,1]$),
and we denote its inverse by $\one+t\,\gamma(x,p)$. The
corresponding cosymplectic structure then assumes the local form
\begin{align}\label{PB1}
\omega^{-1} =
  \tfrac t2\,\theta^{ij}(x)\, \partial_i\wedge\partial_j +
  \tfrac12\,\big(\delta^i_j+t\,\gamma^i_j(x,p)\big)\,
  \big(\partial_i\wedge\tilde\partial{}^j+\tilde\partial{}^j\wedge\partial_i\big)
  \ .
\end{align}

For later use, and in particular for comparison with the generic cases of
almost Poisson structures, it is useful to cast the symplectic
embedding described by \eqref{eq:omegaintdef} into the setting 
of Definition~\ref{def:symplembed} by developing its asymptotic series
in $t$. For this, we expand 
the bivector $\gamma$ as a formal power series 
\bea\label{eq:gammaijxp}
\gamma_i^j(x,p) = 
\sum_{n=1}^\infty\,t^{n-1}\,\gamma_i^{j|i_1\cdots i_n}(x)\, p_{i_1}\cdots
p_{i_n} 
\eea
using \eqref{eq:phiixp}, where the local functions
$\gamma_i^{j|i_1\cdots i_n}(x)$ are proportional to the components of
the bivector $\theta$ and their derivatives. Alternatively, they can
be found by solving the Poisson integrability condition
$[\omega^{-1},\omega^{-1}]=0$, which using $[\theta,\theta]=0$ yields
local first order differential equations
\bea\label{eq2}
\tilde\partial{}^l\gamma_i^k - \tilde\partial{}^k\gamma_i^l +
t\,\big(\gamma_j^l\,\tilde\partial{}^j\gamma_i^k -
\gamma_j^k\,\tilde\partial{}^j\gamma_i^l\big) = \partial_i\theta^{lk} + 
t\,\big(\gamma_i^j\,\partial_j\theta^{lk} +
\theta^{kj}\,\partial_j\gamma_i^l - \theta^{lj}\,\partial_j\gamma_i^k \big)
\eea
for the bivector $\gamma$ in terms of the given bivector
$\theta$. Substituting the formal power series \eqref{eq:gammaijxp}
in \eqref{eq2} then yields an infinite system of recursive
differential equations given by
\begin{align} \label{gtrec}
\gamma_i^{k|l} - \gamma_i^{l|k} &= \partial_i\theta^{lk} \ , 
\end{align}
and
\begin{align} 
& (n+1)\,\big(\gamma_i^{k|li_2\cdots i_{n+1}} -\gamma_i^{l|ki_2\cdots i_{n+1}}\big)
\notag \\ & \hspace{3cm} +\sum_{m=1}^n\,(n-m+1)\,\big(\gamma_j^{l|i_1\cdots i_{m}}\,\gamma_i^{k|ji_{m+1}\cdots i_{n+1}} -
\gamma_j^{k|i_1\cdots i_{m}}\,\gamma_i^{l|ji_{m+1}\cdots i_{n+1}}
                 \big) \notag \\[4pt]
& \hspace{5cm}
= \gamma_i^{j|i_2\cdots i_{n+1}}\,\partial_j\theta^{lk} + \theta^{kj}\,\partial_j\gamma_i^{l|i_2\cdots i_{n+1}} - \theta^{lj}\,\partial_j\gamma_i^{k|i_2\cdots i_{n+1}} 
\ , \label{h5}
\end{align}
for $n\geq1$.
A formal power series solution of (\ref{eq2}) was constructed in
this way by~\cite{Kup14}, with the first two leading orders given by
\begin{align}
\gamma_i^{j|k} &= -\tfrac12\,\partial_{i}\theta^{jk} \nonumber
  \\[4pt]
\gamma_i^{j|kl} &=
                  -\tfrac1{12}\,\big(2\,\theta^{lm}\,\partial_i\partial_m\theta^{kj}
                  + \partial_i\theta^{km}\,\partial_m\theta^{jl} \big) \ .
\label{gt}\end{align}

\begin{remark}
The solution \eqref{gt} is not unique. There is no general notion of
equivalence of symplectic embeddings for a general Poisson
manifold. If $(M,\theta)$ is integrable, there is a notion of Morita
self-equivalence, see e.g.~\cite{BOW20}. We will return to the
question of uniqueness, as well as the meaning of the symplectic
structure on $T^*M$ away from the zero section, later on where we will
find that they have natural interpretations in terms of gauge
algebras (see Remark~\ref{rem:uniqueness} and Proposition~\ref{APhi},
respectively).
\qen\end{remark}

\begin{example}\label{ex:constemb}
Let $M=\real^d$ with a constant Poisson structure $\theta$. This is
the only case in which the solution
$\gamma_i^j=0$ to \eqref{eq2} is possible; this is also the solution that
follows from \eqref{eq:phiixp} for constant $\theta$. With this choice, the
symplectic embedding of $(\real^d,\theta)$ is given by the strict
deformation $(T^*\real^d,\omega_0+t\,\theta^*)$ of the cotangent
symplectic groupoid for $\real^d$, where $\theta^*$
is the vertical two-form on $T^*\real^d$ induced by the linear dual of
the bivector $\theta$ on the vector space $\real^d$, that is, 
$\theta^*(\tilde\partial^i,\tilde\partial^j)=\theta^{ij}$ and
$\theta^*(\partial_i,\partial_j)=0= \theta^*(\partial_i,\tilde\partial^j)$. The
integrating symplectic groupoid
is the direct product of the pair groupoid $\real^r\times\real^r$ 
and the cotangent groupoid $T^*\real^{d-r}$ for $\real^{d-r}$ where $r$ is the
rank of $\theta$, see e.g.~\cite{Cattaneo:2000iw}. However, there are also non-zero
solutions of \eqref{eq2} in this case.
\qen\end{example}

\begin{remark}
\label{rem:symplgroupoid}
An explicit global extension of the local symplectic embedding
\eqref{eq:omegaintdef} to an open neighbourhood $N\subset T^*M$ of the
zero section is given in~\cite{Crainic,Broka}. The construction depends on
the choice of an affine connection $\nabla$ on $M$. It amounts to
replacing the vector field \eqref{eq:Xtheta} with
$$
X^{\theta,\nabla} = \theta^{ij}(x)\,p_i\,\partial_j +
p_k\,p_l\,\theta^{ki}(x)\,\Gamma_{ij}^l(x)\,\tilde\partial{}^j \ ,
$$
where 
$\nabla_{\partial_i}\partial_j =
\Gamma_{ij}^l(x)\, \partial_l$, and replacing $\varphi_u^{\theta}$ with
the corresponding flow $\varphi_u^{\theta,\nabla}$ of
$t\,X^{\theta,\nabla}$ in \eqref{eq:omegaintdef}. The connection
$\nabla$ induces a Lie algebroid connection on the cotangent Lie algebroid
$(T^*M,\theta^\sharp,[\,\cdot\,,\,\cdot\,]_\theta)$ associated to
an integrable Poisson manifold $(M,\theta)$, where the bracket of this Lie algebroid
extends the natural Lie bracket on exact one-forms, given by $[\dd f,\dd
g]_\theta=\dd\{f,g\}_\theta$, in a unique way to the Koszul bracket
$$
[\alpha,\beta]_\theta :=
\LL_{\theta^\sharp\alpha}\beta-\LL_{\theta^\sharp\beta}\alpha
- \dd \theta(\alpha,\beta)
$$ 
where $\LL$ denotes the Lie derivative.
This construction makes $N$ into a
local symplectic groupoid $N\rightrightarrows M$ (cf. \cite{Karasev}), with source map
$\pi$ and target map $\pi\circ\varphi_1^{\theta,\nabla}$, which
integrates the cotangent Lie algebroid. 
\qen\end{remark}

\subsection{Local symplectic embedding of twisted Poisson manifolds}
\label{sec:twistedPoisson}

The definition \eqref{eq:omegaintdef} makes sense for any bivector
$\theta$ and defines a symplectic structure, but it yields a
symplectic embedding only when $\theta$ is a Poisson
structure. However, when $\theta$ is an $H$-twisted Poisson
structure, it is possible to
modify $\omega$ accordingly~\cite{Cattaneo} and turn it into an almost symplectic
embedding of $(M,\theta)$: the two-form
\bea\label{eq:omegaH}
\omega_H = \int_0^1\,
\big(\varphi_u^{\theta}\big)^*\big(\omega_0 +
t\,\pi^*H(X^\theta,\,\cdot\,,\,\cdot\,) \big) \ \dd u
\eea
is an almost symplectic structure with $\dd\omega_H=t\,\pi^*H$ which
satisfies \eqref{eq:zetapicond}.

\begin{remark}
Similarly to the untwisted case (see Remark~\ref{rem:symplgroupoid}), a twisted Poisson structure makes the
cotangent bundle into a Lie algebroid
$(T^*M,\theta^\sharp,[\,\cdot\,,\,\cdot\,]_{\theta,H})$ with Lie
bracket
$$
[\alpha,\beta]_{\theta,H} := [\alpha,\beta]_\theta
+H(\theta^\sharp\alpha,\theta^\sharp\beta,\,\cdot\,) \ .
$$
In particular, for $f,g\in C^\infty(M)$ this gives
$$
[\dd f,\dd g]_{\theta,H} = \dd\{f,g\}_\theta + H(X_f,X_g,\,\cdot\,)
$$
with $X_f=\theta^\sharp\dd f$. If $\theta$ is integrable, the
corresponding integrating Lie groupoid is called a twisted symplectic
groupoid in~\cite{Cattaneo}; it is provided by extending this local almost
symplectic embedding construction by replacing $X^\theta$ with
$X^{\theta,\nabla}$ and $\varphi_u^{\theta}$ with
$\varphi_u^{\theta,\nabla}$ in \eqref{eq:omegaH}~\cite{Crainic}, as
explained in Remark~\ref{rem:symplgroupoid}.
\qen\end{remark} 

In this local picture, it is possible to locally `untwist' the almost
symplectic embedding to a bonafide symplectic embedding by
interpreting the closed three-form $H\in\Omega^3(M)$ as the curvature
of a gerbe connection and following the approach
of~\cite{Aschieri2002,Mylonas2012}. For this, we choose a suitable
covering of the 
manifold $M$ by good open subsets $U_a$. We can write $H$ in terms of
local two-forms $B_a\in\Omega^2(U_a)$ as $H=\dd B_a$. On intersections
$U_{ab}:=U_a\cap U_b$, the two-form $F_{ab}:=B_b-B_a$ is closed and hence
exact, so it can be expressed in terms of one-form fields
$A_{ab}\in\Omega^1(U_{ab})$ as
$F_{ab}=\dd A_{ab}$; triple intersections involve local gauge
transformations by functions $f_{abc}\in C^\infty(U_a\cap
U_b\cap U_c)$ satisfying a suitable integrability condition. Then the
two-form $\omega_a\in\Omega^2(U_a)$ defined by
$$
\omega_a = \omega_H-t\,\pi^*B_a
$$
is a symplectic structure on $T^*U_a$, i.e. $\dd\omega_a=0$. The local
symplectic two-forms $\omega_a$ are related by the flows
$\varphi_u^{ab}$ at $u=1$ generated by the vector fields
$t\,\theta(A_{ab},\,\cdot\,)\in\mfX(U_{ab})$, such that
$$
\big(\varphi_1^{ab}\big)^*\{F,G\}_{\omega^{-1}_b} =
\big\{\big(\varphi_1^{ab}\big)^*F,\big(\varphi_1^{ab}\big)^*G\big\}_{\omega^{-1}_a}
$$ 
for $F,G\in C^\infty(T^*U_{ab})$. This defines a twisted sheaf (or stack) of Poisson algebras on $T^*M$~\cite{Severa2003}.

\subsection{Local symplectic embedding of almost Poisson manifolds}
\label{sec:quasiPoisson}

The existence of local symplectic embeddings for arbitrary
almost Poisson structures is established
in~\cite{Kupriyanov:2018yaj}. The construction of
Section~\ref{sec:Poissonembedding} is
spoilt for a non-zero Jacobiator $\Pim\neq0$, as then
$(C^\infty(M),t\,\{\,\cdot\,,\,\cdot\,\}_{\theta})$ cannot be embedded as a
Poisson subalgebra of the cosymplectic algebra
$(C_{\rm pol}^\infty(T^*M)[[t]],\{\,\cdot\,,\,\cdot\,\}_{\omega^{-1}})$. Nevertheless, 
we can gleam off the
general form of the symplectic embedding from the two special cases
\eqref{PB1} and \eqref{eq:omegaH}. 

For this, we take the
zero section $U\subset T^*U$ to be a Lagrangian submanifold of
$(T^*U,\omega)$ and account for the additional terms involving
$[\theta,\theta]\neq0$. 
This means that the formal Poisson bivector $\omega^{-1}$ can be
written as
\begin{align}\label{PBq}
\omega^{-1} =
  \tfrac t2\,\underline{\theta}^{ij}(x,p)\, \partial_i\wedge\partial_j +
  \tfrac12\,\big(\delta^i_j+t\,\gamma^i_j(x,p)\big)\,
  \big(\partial_i\wedge\tilde\partial{}^j+\tilde\partial{}^j\wedge\partial_i\big)
  \ ,
\end{align}
where the bivector $\gamma$ has a formal power series
expansion as in \eqref{eq:gammaijxp}, with leading order term given in
\eqref{gt}, while 
\bea \label{l12}
\underline{\theta}^{ij}(x,p) = \theta^{ij}(x) - t\,\Pim^{ijk}(x)\, p_k +
\sum_{n=2}^\infty\, t^n\, \theta^{ij|i_1\cdots i_n}(x)\, p_{i_1}\cdots
p_{i_n}
\eea
and the local functions $\theta^{ij|i_1\cdots i_n}(x)$ for $n\geq2$ are
proportional to the components of the trivector $\Pim$ and their
derivatives. The various functions satisfy local first order differential equations determined from the Poisson integrability condition \eqref{eq:Omegamint} with
\begin{align}
\Thetam_0 &= \tfrac12\,\big(\partial_i\wedge\tilde\partial^i + \tilde\partial^i\wedge\partial_i\big) \ , \notag \\[4pt]
\Thetam_1 &= \tfrac12\,\theta^{ij}(x)\,\partial_i\wedge\partial_j + \tfrac12\,\gamma_j^{i|k}(x)\,p_k\,\big(\partial_i\wedge\tilde\partial^j + \tilde\partial^j\wedge\partial_i\big) \ , \notag \\[4pt]
\Thetam_2 &= -\tfrac12\,\Pim^{ijk}(x)\,p_k\,\partial_i\wedge\partial_j + \tfrac12\,\gamma_j^{i|kl}(x)\,p_k\,p_l\,\big(\partial_i\wedge\tilde\partial^j + \tilde\partial^j\wedge\partial_i\big) \ , \label{eq:Thetanalmost}\\[4pt]
\Thetam_{n} &= \tfrac12\,\theta^{ij|i_1\cdots i_{n-1}}(x)\,p_{i_1}\cdots p_{i_{n-1}}\,\partial_i\wedge\partial_j + \tfrac12\,\gamma_j^{i|i_1\cdots i_{n}}(x)\,p_{i_1}\cdots p_{i_{n}}\,\big(\partial_i\wedge\tilde\partial^j + \tilde\partial^j\wedge\partial_i\big) \ , \notag
\end{align}
for $n\geq3$.
This cosymplectic structure indeed satisfies the two requisite
requirements: (a)~it coincides
with the original almost Poisson structure $\theta$ along the zero section of $T^*U$, as in
\eqref{eq:zetapicond}; and (b) in the case of a Poisson
bivector $\theta$, the symplectic embedding \eqref{PBq} restores the
symplectic embedding \eqref{PB1} of a Poisson structure. 
For example, an explicit form for the order $t^2$ term
according to~\cite{Kupriyanov:2018yaj} reads
\begin{align}
  \theta^{ij|kl}&=\tfrac{3}{16}\,\big(\Pim^{jlm}\,\partial_m\theta^{ki}+\Pim^{jkm}\,\partial_m\theta^{li}-
  \Pim^{ilm}\,\partial_m\theta^{kj}-\Pim^{ikm}\,\partial_m\theta^{lj}
                   \big) \nonumber \\
  & \quad
    -\tfrac{1}{8}\,\big(\theta^{km}\,\partial_m\Pim^{ijl}+\theta^{lm}\,\partial_m\Pim^{ijk}
    \big) \ .
\label{Theta4}\end{align}

\begin{remark}\label{rem:twistedB}
The formalism here covers as well the special case of twisted Poisson
structures from Section~\ref{sec:twistedPoisson}, and it explicitly realises the local `untwisting' of the almost symplectic realization \eqref{eq:omegaH} to a symplectic embedding $\omega$. For example, consider the case of a topologically trivial closed three-form $H\in\Omega^3(M)$, that is, $H=\dd B$ for a globally defined two-form $B\in\Omega^2(M)$, with associated linear map $B^\flat:\mfX(M)\to\Omega^1(M)$. In this case we can choose the trivial one-form $A=0$ (up to gauge equivalence), for which the associated flows are identity maps of $M$ and the corresponding map
\begin{align*}
\big(\omega^{-1}\big)^\sharp = \big(\omega_{\dd B}^{-1}\big)^\sharp\,\big(\one-t\,(\pi^*B)^\flat\,(\omega_{\dd B}^{-1})^\sharp\big)^{-1} : \Omega_{\rm pol}^1(T^*M)[[t]]\longrightarrow\mfX_{\rm pol}(T^*M)[[t]]
\end{align*}
defines a \emph{$B$-transformation} of the almost cosymplectic structure $\omega_{\dd B}^{-1}$. 
\qen
\end{remark}

\subsection{Semi-classical limit of deformation quantization}
\label{sec:defquant}

In the context of deformation quantization, a Poisson bracket may be
regarded as the semi-classical limit of the commutator bracket of an associative noncommutative
star-product which quantizes a Poisson manifold $(M,\theta)$; the
existence of such star-products is provided by the famous Kontsevich
formality theorem~\cite{Kontsevich}. More precisely, a star-product on
$M$ is a product on $C^\infty(M)[[\hbar]]$ (regarded as a
$\complex[[\hbar]]$-module for a formal deformation parameter $\hbar$) of the form
\begin{align*}
f\star g =
f\,g+\sum_{n=1}^\infty \,\hbar^n\,{\rm B}_n(f,g)
\end{align*}
for smooth functions $f,g\in C^\infty(M)$, where each ${\rm B}_n:C^\infty(M)\times
C^\infty(M)\to C^\infty(M)$ for $n\geq1$ is a bidifferential operator, and the
Poisson structure is recovered through $\{f,g\}_\theta =
\frac1\hbar\,[f,g]_\star\big|_{\hbar=0}$ where \smash{$[f,g]_\star=f\star
  g-g\star f$}. 
  
The relation between symplectic embeddings and the semi-classical
limit of deformation
quantization of a Poisson structure $\theta$ is well-known (at least implicitly) in both the
physics and mathematics literature; after all, symplectic realizations
were originally introduced with the quantization problem for
Poisson manifolds in mind. On $M=U\subseteq \real^d$, the Fourier
integral representation of the Kontsevich star-product on $(U,t\,\theta)$ can be
brought to the form
\begin{align}\label{eq:Spxdef}
f\star g(x) = \int_{(\real^d)^*\times(\real^d)^*} \, \hat f(p)\,\hat
  g(p') \, a_\hbar(p,p',x) \, \e^{\frac\ii\hbar \, \Sigma(p,p',x)} \,
  \frac{\dd p \ \dd p'}{(2\pi\,\hbar)^d} \ ,
\end{align}
where $\hat f$ and $\hat g$ are the asymptotic Fourier transforms of
$f,g\in C^\infty(U)$, and the function $a_\hbar(p,p',x)$ is regular at
$\hbar=0$. In the semi-classical limit $\hbar\to0$, the leading
contribution is given by the oscillatory phase $\Sigma(p,p',x)$, which is a
formal power series
\begin{align*}
\Sigma(p,p',x) = \langle p+p',x\rangle + \sum_{n=1}^\infty \, t^n \,
  \Sigma_n(p,p',x) \ ,
\end{align*}
where $\langle \,\cdot\,,\,\cdot\,\rangle$ is the dual pairing between
$\real^d$ and $(\real^d)^*$, and each $\Sigma_n(p,p',x)$ for $n\geq1$ is a homogeneous polynomial in
$p,p'\in(\real^d)^*$ of degree $n+1$, with $\Sigma_n(p,0,x)=\Sigma_n(0,p,x)$,
whose homogeneous part $\tilde\Sigma_n(p,p',x)$ in $p$ satisfies $\tilde
\Sigma_n(p,p,x)=0$. It formally generates a local symplectic groupoid structure on
$(T^*U,\omega_0)$~\cite{CattaneoDherinFelder}, and in this manner the formal symplectic
groupoid can be regarded as a semi-classical version of the full
noncommutative algebra of functions $(C^\infty(U)[[\hbar]],\star)$.

This determines a formal Poisson submersion $\pi_\theta:(T^*U,\omega_0^{-1})\to (U,t\,\theta)$ with
Lagrangian zero section which in components is defined by
\begin{align}\label{eq:Boppshift}
\pi_\theta(x,p)^i = \tilde\partial'{}^i\Sigma(p,p',x)\big|_{p'=0} = x^i + \sum_{n=1}^\infty \, t^n \, \Sigma^{i|i_1\cdots i_n}(x) \, p_{i_1}\cdots p_{i_n} \ ,
\end{align}
where $\tilde\partial'{}^i=\partial/\partial p_i'$ and $\Sigma^{i|i_1\cdots i_n}(x)\,p_{i_1}\cdots p_{i_n} = \tilde\partial'{}^i\Sigma_n(p,p',x)\big|_{p'=0}$. This map is sometimes called a `generalized Bopp shift' in the physics literature, see e.g.~\cite{KV,Kup14}. In~\cite{Kupriyanov:2018yaj} it was shown how to construct the (not unique) local functions $\Sigma^{i|i_1\cdots i_n}(x)$ by solving the (formal) Poisson map equation
\begin{align*}
\{\pi_\theta^*f,\pi_\theta^*g\}_{\omega_0^{-1}} = t\,\pi_\theta^*\{f,g\}_\theta
\end{align*} 
order by order in $t$, with the first three leading orders which are compatible with quantization given by
\begin{align*}
\Sigma^{i|j} &= -\tfrac12\,\theta^{ij} \ , \\[4pt]
\Sigma^{i|jk} &= \tfrac1{24}\,\theta^{kl}\,\partial_l\theta^{ij} + \tfrac1{24} \, \theta^{jl}\,\partial_l\theta^{ik} \ , \\[4pt]
\Sigma^{i|jmn} &= -\tfrac1{12} \, \big(2\,\Sigma^{l|mn}\,\partial_l\theta^{ij} + 2\,\Sigma^{l|jn}\,\partial_l\theta^{im} + 2\,\Sigma^{ljm}\,\partial_l\theta^{in} \\ & \quad \hspace{1cm} + \tfrac1{12}\,\theta^{lm}\,\theta^{kn}\,\partial_l\partial_k\theta^{ij} + \tfrac1{12}\,\theta^{lj}\,\theta^{kn}\,\partial_l\partial_k\theta^{im} + \tfrac1{12}\,\theta^{lj}\,\theta^{km}\,\partial_l\partial_n\theta^{in}\big) \ .
\end{align*}
Writing $\tilde\pi(x,p)=p$ for the projection to the normal directions to $U\subset T^*U$, the inverse of the map $\pi_\theta\times\tilde\pi:T^*U\to T^*U$ then replaces \eqref{eq:Boppshift} with the canonical cotangent bundle projection $\pi(x,p)=x$ and yields the local symplectic embedding constructed in Section~\ref{sec:Poissonembedding}; this is the sense in which the symplectic embedding `integrates' the Poisson manifold $(M,t\,\theta)$. The existence of invertible generalized Bopp shifts is always guaranteed, at least locally, by virtue of Darboux's theorem.

\begin{example}
If $\theta=\theta_0=0$, then $\Sigma_n=0$ for all $n\geq1$ and $\pi_0(x,p)=x$. More generally, if $M=\real^d$ with a constant Poisson structure $\theta$, then 
\begin{align*}
\Sigma(p,p',x) = \langle p+p',x\rangle + \tfrac t2\,\theta^{ij}\,p_i\,p_j' \ ,
\end{align*}
and the function $a_\hbar$ in \eqref{eq:Spxdef} is identically equal to $1$. In this case the generalized Bopp shift
\begin{align*}
\pi_\theta(x,p)^i = x^i -\tfrac t2\,\theta^{ij}\,p_j
\end{align*}
reproduces the symplectic embedding of Example~\ref{ex:constemb}.
\qen
\end{example}

In the general case, we may again formulate deformation
quantization in the direction of an almost Poisson bracket 
through a suitable nonassociative star-product. For twisted Poisson
manifolds these star-products can be constructed through Kontsevich's
formalism (see e.g.~\cite{Cornalba:2001sm,Severa2003,Aschieri2002,Mylonas2012}). For generic
almost Poisson manifolds the existence and uniqueness of nonassociative Weyl
star-products is established by~\cite{Kupriyanov:2015dda}.

The relation between symplectic embeddings and the semi-classical
limit of deformation quantization of a generic almost Poisson structure
$\theta$ is much more involved. 
For twisted Poisson manifolds, the
twisted symplectic groupoids of~\cite{Cattaneo} capture the
semi-classical limit of the nonassociative algebra of functions, but
these are not directly induced by our symplectic embedding formalism,
which deals with strictly associative structures. Instead,
in~\cite[Section~9]{KS18} it was shown that symplectic embeddings capture the
semi-classical limit of the associative composition algebra of
differential operators on $M$ induced by the star-product~\cite{MSS1}, which coincides
with the noncommutative algebra of functions precisely when the Jacobiator
$\Pim$ vanishes. 
In this setting, the semi-classical limit of a nonassociative star-product quantizing a generic almost Poisson structure $\theta$ was formulated as a generalized Bopp shift by~\cite{Kupriyanov:2018yaj} as the submersion $\pi_\theta:T^*U\to U$ satisfying
\begin{align*}
\{\pi_\theta^*f,\pi_\theta^*g\}_{\omega_0^{-1}} = t\,(\pi_\theta\times\tilde\pi)^*\{\pi^*f,\pi^*g\}_{\underline{\theta}} \ ,
\end{align*}
where $\underline{\theta}=\frac12\,\underline{\theta}^{ij}(x,p)\,\partial_i\wedge\partial_j$ is the bivector introduced in Section~\ref{sec:quasiPoisson} with the formal power series expansion \eqref{l12}. The map $\pi_\theta$ has precisely the same expansion \eqref{eq:Boppshift} as in the case of a Poisson structure, and inverting $\pi_\theta\times\tilde\pi$ then gives a local symplectic embedding in the weaker sense of Definition~\ref{def:symplembed}, as constructed explicitly in Section~\ref{sec:quasiPoisson}. 

\newsection{Poisson gauge transformations}
\label{sec:Poissongauge}

In this section we focus on the case of Poisson structures,
i.e. $\Pim=0$. Recalling our discussion of deformation quantization from Section~\ref{sec:defquant}, a noncommutative gauge theory is a formal
deformation of an ordinary gauge theory, obtained by replacing
pointwise products of gauge fields with star-products.
A \emph{Poisson gauge theory} is a limit of a noncommutative gauge
theory whose gauge algebra is the semi-classical limit of the
noncommutative gauge algebra with closure defined by the commutators
$[f,g]_\star$, in the sense defined below; we return to this point of view in Section~\ref{sec:Pinfty} where we will make this notion somewhat more precise in the context of homotopy Poisson algebras. In
this paper we deal only with the kinematical data of a Poisson gauge
theory. We also discuss only the case of gauge theories with structure
group $U(1)$ for simplicity.

A conventional local $U(1)$ gauge transformation is specified by a pair
$$
(f,A) \ \in \ C^\infty(U)\times\Omega^1(U)
$$
on an open subset $U\subseteq M$ consisting of a gauge parameter $f$ and a gauge field $A$. The
commutative ring of functions $C^\infty(U)$ is an enveloping
algebra for an abelian Lie algebra which acts on $\Omega^1(U)$ as $(f,A)\mapsto
A+\delta_f^0A$ via the gauge variation
$$
\delta_f^0A=\dd f \ .
$$
The closure condition for the gauge variations defines an abelian Lie algebra:
$$
\big[\one+\delta_f^0,\one+\delta_g^0\big]A:=\big((\one+\delta_f^0)\circ(\one+\delta_g^0) -
(\one+\delta_g^0)\circ(\one+\delta_f^0)\big)A = 0
$$
for $f,g\in C^\infty(U)$. 

In a Poisson gauge theory, this construction is
deformed in the following way. The idea is then to mimick the
representation of the Poisson algebra as infinitesimal
diffeomorphisms of $M$: The map $f\mapsto X_f:=\theta^\sharp\dd f$ is a
Lie algebra homomorphism from
$(C^\infty(M),\{\,\cdot\,,\,\cdot\}_\theta)$ to the Lie algebra of
vector fields $(\mfX(M),[\,\cdot\,,\,\cdot\,])$:
\bea\label{eq:XfLiealg}
[X_f,X_g] = X_{\{f,g\}_\theta} \ .
\eea
For later reference, we note
that this construction makes the trivial line bundle $M\times\real$
into a Lie algebroid over $M$. 

\begin{definition}
\label{def:Poissongauge}
Let $(M,\theta)$ be a Poisson manifold. A (\emph{local})
\emph{Poisson gauge transformation} on an open subset $U\subseteq M$ 
is an action of the Lie algebra
$(C^\infty(U),t\,\{\,\cdot\,,\,\cdot\}_\theta)$ on the affine space
$\Omega^1(U)[[t]]$ of the form
$$
(f,A)\longmapsto A+\delta^\theta_fA 
$$
for $(f,A)\in C^\infty(U)\times\Omega^1(U)$, which is a deformation of
an abelian gauge transformation:
\bea\label{initial}
\delta_f^\theta A\big|_{t=0} = \delta_f^0A = \dd f \ ,
\eea
and satisfies the derivation property
\begin{align*}
\delta_{f\,g}^\theta A = g\,\delta_f^\theta A + f\,\delta_g^\theta A 
\end{align*}
over the algebra $C^\infty(U)$.
The closure condition for the Lie
algebra of gauge variations is the \emph{Poisson gauge algebra} 
\bea\label{closure}
\big[\one+\delta^\theta_f,\one+\delta^\theta_g\big]A=\delta^\theta_{t\,\{f,g\}_\theta} A \ .
\eea
\end{definition}

\begin{remark}\label{rem:Abracket}
In noncommutative gauge theory, the space 
$\Omega^1(U)[[\hbar]]$ is naturally acted upon by the vector space $C^\infty(U)[[\hbar]]$ under the
left and right actions by multiplication with the star-product
extended over one-forms by replacing the bidifferential operators
${\rm B}_n$ with Lie bidifferential operators. In the semi-classical
limit, this naturally defines a skew-symmetric $\Omega^1(U)$-valued bracket
between gauge parameters and gauge fields that we denote by
$\{f,A\}_\theta=-\{A,f\}_\theta$. Defining $\DD
A\in\Omega^1(U)\otimes \Omega^1(U)$ in local coordinates by $\DD
A:=\LL_{\partial_i}A\otimes \dd x^i$, this bracket reads as
$$
\{f,A\}_\theta = (\theta\otimes \one)(\dd f,\DD A) \ ,
$$
which can be expressed as
$$
\{f,A\}_\theta = \theta^{ij}\, \partial_if \, \LL_{\partial_j}A
$$
in local coordinates. In the following we will use the abbreviation
$\theta^\otimes:=\theta\otimes \one$. This bracket obeys the
usual derivation property over the algebra $C^\infty(U)$ in its first
entry, as well as with respect to the $C^\infty(U)$-bimodule
$\Omega^1(U)$ in its second entry since
\begin{align*}
  \DD(g\,A)=\dd
  g\otimes A+g\,\DD A
\end{align*}
  implies
\begin{align*}
\{f,g\,A\}_\theta=\{f,g\}_\theta\,A + g\,\{f,A\}_\theta \ ,
\end{align*}
for all $g\in C^\infty(U)$.

In later sections we will also use an extension of this bracket to differential forms of arbitrary degree. For $\alpha,\beta\in\Omega^1(U)$, we define their symmetric bracket $\{\alpha,\beta\}_\theta\in\Omega^2(U)$ by
\begin{align*}
\{\alpha,\beta\}_\theta := \theta^\otimes(\DD\alpha,\DD\beta) \ .
\end{align*}
This is then extended to a graded skew-symmetric bracket $\{\,\cdot\,,\,\cdot\,\}_\theta$ on the entire exterior algebra $\Omega^\bullet(U)$ as a graded biderivation of degree~$0$.
\qen\end{remark}

\begin{example}
\label{ex:Rd}
The simplest example (apart from the zero bivector) of a Poisson gauge
transformation comes from taking $M=\real^d$ and a constant
skew-symmetric $d{\times}d$ matrix $\theta=(\theta^{ij})$, regarded as
a bivector on $\real^d$. Then the gauge variations
$$
\delta_f^\theta A = \dd f + t\,\{A,f\}_\theta
$$
fulfill the requirements of Definition~\ref{def:Poissongauge}. In the
following we extend this result to generic Poisson manifolds
$(M,\theta)$, which will generally require the use of formal power series.
\qen\end{example}

By the discussion of Section~\ref{sec:defquant}, it is natural to expect that symplectic embeddings should
play a role
in determining the semi-classical limit of a noncommutative gauge theory, that is,
in Poisson gauge theories. The purpose of this section is to
demonstrate that this is indeed the case, focusing in detail
on the realization of gauge symmetries.
Our main observation here is that a Poisson gauge algebra is
canonically defined by the restriction of a symplectic embedding
$(T^*M,\omega)$ of
$(M,\theta)$ to a constraint submanifold defined by the gauge
fields. For this, we regard a 
gauge field $A\in\Omega^1(U)$ as a section $s_A:U\to T^*U$, whose image is a submanifold
$\im(s_A)\subset T^*U$; in local coordinates where $A=A_i(x)\,\dd
x^i$, $s_A(x)=(x,A(x))\in T^*U$ for $x\in U$. 
Define the local one-form $\Phi_A\in\Omega^1(T^*U)$ by
$$
\Phi_A=\lambda_0 - \pi^*A
$$
where we recall that $\lambda_0$ is the Liouville one-form on $T^*U$. The
one-form $\Phi_A$ vanishes precisely on the submanifold $\im(s_A)\subset T^*U$.
In local coordinates where 
$\Phi_A=(\Phi_A)_i\,\dd x^i$ with
\begin{align*}
  (\Phi_A)_i=p_i-A_i(x) \ ,
\end{align*}
this is the
submanifold of $T^*U$ defined by the constraints $(\Phi_A)_i=0$. Since
$s_A^*\lambda_0=A$, this is equivalently presented as $s_A^*\Phi_A=0$ in
$\Omega^1(U)$.

Note that $s_A$ is a Lagrangian section of $(T^*U,\omega_0)$
if and only if $A$ is a flat connection, i.e.~$\dd A=0$. Thus in general, $\im(s_A)$ is
\emph{not} a Lagrangian submanifold of $(T^*U,\omega)$. In physics
parlance, the local equations $(\Phi_A)_i=0$ do not define first class
constraints, as expected since otherwise they would eliminate all local degrees of
freedom. In the present context, the consistent elimination of the
`auxiliary' variables $p_i$ which are adjoined in symplectic 
embeddings of Poisson manifolds~\cite{KS18}
acquires a natural meaning through
\begin{proposition}
\label{APhi}
Let $(T^*M,\omega)$ be a local symplectic embedding of a Poisson
manifold $(M,\theta)$, and let $U\subseteq M$ be an open subset. For
$(f,A)\in C^\infty(U)\times\Omega^1(U)$, the gauge variation 
\bea\label{gtA0}
\delta^\theta_fA := s_A^*\{\pi^*f,\Phi_A\}_{\omega^{-1}}
\eea
is a Poisson gauge transformation.
\end{proposition}
\begin{proof}
Using \eqref{PB1} the gauge variations \eqref{gtA0} read
\begin{equation}\label{gtA}
\delta^\theta_f A =\dd f + t\,
s_A^*\gamma(\dd f,\,\cdot\,)+t\,\{A,f\}_{\theta} \ ,
\end{equation}
where the one-form $s_A^*\gamma(\dd f,\,\cdot\,)\in \Omega^1(U)[[t]]$
is given by 
$$
s_A^*\gamma(\dd f,\,\cdot\,) =
\gamma_i^j\big(x,A(x)\big)\,\partial_j f(x)\, \dd x^i
$$ 
in local coordinates
on $U$.
We need to check the closure condition \eqref{closure} for
this definition of gauge transformations. 

We begin with some preliminary definitions. For a one-form which is a functional
$\FF(A)\in\Omega^1(U)$ of the gauge field $A$, we define the
gauge variation
\begin{equation}\label{eq:FFAgt}
\delta^\theta_f \FF(A):=\FF\big(A+\delta^\theta_f A\big)-\FF(A) \ .
\end{equation}
In particular,
\begin{equation*}
\delta^\theta_g \{A, f\}_{\theta}=\big\{\delta^\theta_g A,f\big\}_{\theta} \ .
\end{equation*}
If $\underline{\FF}\in\Omega^1_{\rm pol}(T^*U)$, then we define
\begin{equation*}
\delta^\theta_f \big(s_A^*\underline\FF\big):=s_A^*\big(\tilde\partial{}^j\underline{\FF}\big) \, \delta^\theta_f A_j \ ,
\end{equation*}
where $\delta^\theta_fA:=\delta^\theta_fA_i\,\dd x^i$. 

We are now ready to calculate the composition of two gauge
variations. We obtain
\begin{align*}
\delta^\theta_f \big(\delta^\theta_g
  A\big)&=s_A^*\big(\tilde\partial{}^j\{\pi^*g,\Phi_A\}_{\omega^{-1}}\big)\,\delta^\theta_fA_j
      + t\,\big\{\delta^\theta_fA,g\big\}_\theta \\[4pt]
&=
  s_A^*\big(\tilde\partial{}^j\{\pi^*g,\Phi_A\}_{\omega^{-1}}\big)\,
 s_A^*\{\pi^*f,(\Phi_A)_j\}_{\omega^{-1}}
  +t\,\{s_A^*\{\pi^*f,\Phi_A\}_{\omega^{-1}},g\}_\theta \ .
\end{align*}
Thus for the left-hand side of the closure condition
\eqref{closure} one finds
\begin{align}
& \delta^\theta_f \big(\delta^\theta_g A\big)-\delta^\theta_g \big(\delta^\theta_f A\big) \notag \\[4pt]
& \hspace{1cm}= s_A^*\big(\tilde\partial{}^j\{\pi^*g,\Phi_A\}_{\omega^{-1}}\big)\,
 s_A^*\{\pi^*f,(\Phi_A)_j\}_{\omega^{-1}} - s_A^*\big(\tilde\partial{}^j\{\pi^*f,\Phi_A\}_{\omega^{-1}}\big)\,
 s_A^*\{\pi^*g,(\Phi_A)_j\}_{\omega^{-1}}
                                               \notag \\
& \hspace{2cm} + s_A^*\{\pi^*s_A^*\{\pi^*f,\Phi_A\}_{\omega^{-1}},\pi^*g\}_{\omega^{-1}}
  - s_A^*\{\pi^*s_A^*\{\pi^*g,\Phi_A\}_{\omega^{-1}},\pi^*f\}_{\omega^{-1}} \ .
\label{t1}\end{align}
We now apply the relation (\ref{r3}) from Appendix~\ref{app:useful} to the
right-hand side of \eqref{t1} to get
\begin{align}
\delta^\theta_f \big(\delta^\theta_g
  A\big)-\delta^\theta_g \big(\delta^\theta_f A\big) &=
  s_A^*\{\{\pi^*f,\Phi_A\}_{\omega^{-1}},\pi^*g\}_{\omega^{-1}}
- s_A^*\{\{\pi^*g,\Phi_A\}_{\omega^{-1}},\pi^*f\}_{\omega^{-1}} \ .
\label{t3}\end{align}
Using the Jacobi identity in the right-hand side of
\eqref{t3}, we finally end up with
\begin{align}\label{t4}
\delta^\theta_f \big(\delta^\theta_g A\big)-\delta^\theta_g \big(\delta^\theta_f
  A\big)=
  s_A^*\{\{\pi^*f,\pi^*g\}_{\omega^{-1}},\Phi_A\}_{\omega^{-1}} \ .
\end{align}
Since 
$$
\{\pi^*f,\pi^*g\}_{\omega^{-1}}=t\,\pi^*\{f,g\}_{\theta} \ ,
$$
we conclude that the gauge transformations (\ref{gtA0}) close the Lie algebra
\begin{equation*}
\big[\one+\delta^\theta_f,\one+\delta^\theta_g\big] A=\delta^\theta_{t\,\{f,g\}_\theta} A \ ,
\end{equation*}
as required.
\end{proof}

\begin{remark}
The field dependent term $\dd f + t\,s_A^*\gamma(\dd f,\,\cdot\,)$ in \eqref{gtA} can
be thought of as defining a `twisted exterior derivative' of the gauge
parameter $f$ required to close the Poisson gauge algebra when the
bivector $\theta$ is no longer constant, cf.\
Examples~\ref{ex:constemb} and~\ref{ex:Rd}. Indeed, the gauge closure
condition \eqref{closure} for \eqref{gtA} implies the local symplectic
embedding equations \eqref{eq2}~\cite{Kupriyanov:2019cug}.
The precise algebraic meaning of this term will be explained in
Section~\ref{sec:Linfinity}, and we shall see some explicit examples
in Section~\ref{sec:Examples}. 
\qen\end{remark}

\newsection{Almost Poisson gauge transformations}
\label{sec:quasiP}

In this section we consider the case of a generic almost Poisson
bivector $\theta$ on $M$ with non-vanishing Jacobiator,
i.e. $\Pim\neq0$. In this case, we may again formulate the notion of an
\emph{almost Poisson gauge theory} as an appropriate semi-classical limit of a
noncommutative gauge theory, which is constructed by nonassociative deformation
quantization. The inherent associativity of the symplectic
embeddings discussed in Section~\ref{sec:defquant} will be needed to ensure that the gauge transformations in an
almost Poisson gauge theory close a (strict) Lie algebra, despite their
origin in a nonassociative algebra.

The aim of this section is to generalize Proposition~\ref{APhi} to the
case of generic almost Poisson manifolds. The
construction of gauge transformations and gauge algebras in these
instances is considerably more involved than in the case of Poisson
structures from Section~\ref{sec:Poissongauge}. We proceed in two
steps: We first set up the problem of defining a suitable notion of
almost Poisson gauge transformations and their closure condition, and
subsequently prove that solutions to this problem exist and are explicitly
computable. 

\subsection{Formulation of the gauge algebra}
\label{sec:formulationquasi}

It is clear that for a non-trivial Jacobiator $\Pim\neq0$, a
direct generalization of Definition~\ref{def:Poissongauge} is not
possible, because $(C^\infty(M),\{\,\cdot\,,\,\cdot\,\}_\theta)$ is no
longer a Lie algebra. This is already manifested in a violation of the
Lie algebra homomorphism property \eqref{eq:XfLiealg}, which for an
$H$-twisted Poisson structure is modified to
$$
[X_f,X_g] = X_{\{f,g\}_\theta} +
\theta^\sharp H(X_f,X_g,\,\cdot\,) \ .
$$
In particular, the gauge closure condition must
change substantially. Drawing on results from the
$L_\infty$-algebra approach to noncommutative gauge
theories~\cite{Kupriyanov:2019ezf}, it is clear what needs to be done:
the space of gauge parameters should be enlarged to include `field
dependent' gauge parameters, such that the commutator of two gauge
transformations is again a gauge transformation but with a field
dependent gauge parameter. We shall discuss the precise relation between our
approach here based on symplectic embeddings and the approach based on
$L_\infty$-algebras in Section~\ref{sec:Linfinity} below.

Let us first explain precisely what we mean by a `field dependent'
gauge parameter in our setting. We will work with $1$-jets
of sections of vector bundles, see e.g.~\cite{Sardan}.

\begin{definition}
Let $M$ be a manifold and $U\subseteq M$ an open subset. Let
$J^1T^*U$ be the first order jet space of the cotangent bundle over $U$. A \emph{field
dependent gauge parameter} is a smooth function in $C^\infty(J^1T^*U)$. 
\end{definition}

Concretely, in local coordinates a field dependent gauge parameter is
specified by a local function $f[x,A]=f(x,A(x),\partial A(x))$
depending on a gauge field $A\in\Omega^1(U)$, viewed as a section
$s_A:U\to T^*U$, and its first order derivatives. Via pullback along
the vector jet bundle $J^1T^*U\to U$, this contains the usual space
$C^\infty(U)$ of
(field independent) gauge parameters from
Section~\ref{sec:Poissongauge}. In what follows we will also use
pullbacks along the
affine jet bundle $J^1T^*U\to T^*U$. In particular, this enables us to
view the space of gauge fields $\Omega^1(U)$ as a subspace of the
sections $\Gamma(J^1T^*U)$ of the vector jet bundle via prolongation;
in other words, gauge fields are precisely the integrable sections of
$J^1T^*U\to U$.

With this notion in place, we are now in a position to generalize
Definition~\ref{def:Poissongauge}. The idea is to enlarge the Lie
algebra action $C^\infty(U)\times\Omega^1(U)\to\Omega^1(U)$ from
Section~\ref{sec:Poissongauge} to a
suitable affine action
$C^\infty(J^1T^*U)\times\Gamma(J^1T^*U)\to\Gamma(J^1T^*U)$. 

\begin{definition}
\label{def:quasiPoissongauge}
Let $(M,\theta)$ be an almost Poisson manifold. A (\emph{local})
\emph{almost Poisson gauge transformation} on an open subset
$U\subseteq M$ is an affine action of $C_{\rm
  pol}^\infty(J^1T^*U)[[t]]$ on the space
$\Gamma_{\rm pol}(J^1T^*U)[[t]]$ of the form
$$
(f,A)\longmapsto A+\delta^\theta_fA
$$
for $(f,A)\in C^\infty(U)\times\Omega^1(U)$, such that:
\begin{itemize}
\item[(a)] The assignment
$f\mapsto\delta^\theta_f$ is $\complex$-linear and defines a
deformation of an abelian gauge transformation:
$$
\delta_f^\theta A\big|_{t=0} = \delta_f^0A=\dd f \ ,
$$
which satisfies the derivation property
\begin{align*}
\delta_{f\,g}^\theta A = g\,\delta_f^\theta A + f\,\delta_g^\theta A 
\end{align*}
over the algebra $C^\infty(U)$,
 and closes the
\emph{almost Poisson gauge algebra}
\bea\label{eq:qPoissonalg}
\big[\one+\delta^\theta_f,\one+\delta^\theta_g\big]A = \delta^\theta_{t\,[\![f,g]\!]_\theta(A)}A \ ,
\eea
where
$[\![f,g]\!]_\theta \in C_{\rm pol}^\infty(J^1T^*U)[[t]]$;
\item[(b)] The
assignment $(f,g)\mapsto[\![f,g]\!]_\theta$ is
$\complex$-bilinear, skew-symmetric and defines a deformation of the
almost Poisson bracket:
$$
[\![f,g]\!]_\theta\big|_{t=0}=\{f,g\}_\theta \ ,
$$ 
which satisfies the relation
\begin{align}\label{j3}
{\sf Cyc}_{f,g,h}\big(t\,\big[\!\big[h,[\![f,g]\!]_\theta(A)\big]\!\big]_\theta(A) +
  \delta_h^\theta[\![f,g]\!]_\theta(A) \big) = 0
\end{align}
for consistency with the Jacobi identity for the commutator
algebra \eqref{eq:qPoissonalg}, where we use the definition
\eqref{eq:FFAgt} for gauge variations; and
\item[(c)]
For any fixed $f\in C^\infty(U)$, the map $[\![f,\,\cdot\,]\!]_\theta$
is a derivation of the commutative algebra $C_{\rm pol}^\infty(J^1T^*U)[[t]]$.
\end{itemize}
\end{definition}

When $\theta$ is a Poisson structure, i.e. the Jacobi identity for the bracket
$\{\,\cdot\,,\,\cdot\,\}_\theta$ holds,
Definition~\ref{def:Poissongauge} is a special case of
Definition~\ref{def:quasiPoissongauge}, but the latter allows more freedom to
accomodate situations where the
Jacobi identity is violated. Nevertheless, a direct generalization of
Proposition~\ref{APhi} is still not possible for generic
almost Poisson bivectors $\theta$; in the language of~\cite{KS18}, the
auxiliary variables $p_i$ can no longer be consistently eliminated by
the constraints $\Phi_A=0$ on a symplectic embedding of an
almost Poisson manifold. One can still repeat most of the proof of
Proposition~\ref{APhi} in the present case to arrive again at the
expression \eqref{t4}. However, due to \eqref{PBq}, the key difference from the case of a
Poisson bivector is that the cosymplectic bracket
$$
\{\pi^*f,\pi^*g\}_{\omega^{-1}} =
t\, \underline{\theta}(\dd\pi^*f,\dd\pi^*g) \neq t\,\pi^*\{f,g\}_{\theta}
$$
is not a gauge parameter since it depends on the coordinates in the normal directions to the zero section $U\subset T^*U$ through its local
dependence on
$\underline{\theta}^{ij}(x,p)$. Similarly, the restriction of
$$
\{\pi^*A,\pi^*f\}_{\omega^{-1}} =
t\,\underline{\theta}^\otimes(\DD\pi^*A, \dd\pi^*f) \neq
t\,\pi^*\{A,f\}_{\theta} 
$$
to $\im(s_A)$ involves a field dependent gauge transformation through
its local dependence on $\underline{\theta}^{ij}(x,A(x))$.

Using the relation
(\ref{r4}) from Appendix~\ref{app:useful} in the right-hand side of
\eqref{t4}, we can represent the commutator of two gauge
transformations defined by \eqref{gtA0} in the form
\begin{align}
\delta^\theta_f\big(\delta^\theta_gA\big) -
\delta^\theta_g\big(\delta^\theta_fA\big) &=
s_A^*\{\pi^*s_A^*\{\pi^*f,\pi^*g\}_{\omega^{-1}} , \Phi_A\}_{\omega^{-1}}
                                                      \notag \\ &
                                                                  \quad +
  s_A^*\big(\tilde\partial{}^j\{\pi^*f,\pi^*g\}_{\omega^{-1}}\big) \,
  s_A^*\{(\Phi_A)_j,\Phi_A\}_{\omega^{-1}} \ .
\label{t7}\end{align}
While the expression $s_A^*\{\pi^*f,\pi^*g\}_{\omega^{-1}}$ defines a
field dependent gauge parameter, the second line of \eqref{t7} prevents the
gauge transformations defined by \eqref{gtA0} from closing a Lie
algebra. To overcome this problem we need to `correct' the gauge
variation \eqref{gtA0} in order to absorb this term, which is the
content of
\begin{proposition}
\label{p1}
Let $(T^*M,\omega)$ be a local symplectic embedding of an almost Poisson
manifold $(M,\theta)$, and let $U\subseteq M$ be an open subset. For
$f,g\in C^\infty(U)$ and $A\in\Omega^1(U)$, define
$[\![f,g]\!]_\theta\in C_{\rm pol}^\infty(J^1T^*U)[[t]]$
by 
\bea\label{c3}
t\,[\![f,g]\!]_\theta(A) = s_A^*\{\pi^*f,\pi^*g\}_{\omega^{-1}} -
s_A^*\{\Phi_A(L_f),\Phi_A(L_g)\}_{\omega^{-1}} \ ,
\eea
where $
L_f\in \mfX_{\rm pol}^{\tt h}(J^1T^*U)[[t]]\subset \mfX_{\rm pol}^{\tt
  h}(T^*U)[[t]]$ are horizontal vector
fields $L_f=L_f^i\,\partial_i$ on the jet space of $T^*U$
satisfying the recursion relations
\begin{align}
L^i_{t\,[\![f,g]\!]_\theta(A)} &= \delta_f^\theta L^i_g - \delta_g^\theta L^i_f +
                   s_A^*\big(\tilde\partial{}^i\{\pi^*f,\pi^*g\}_{\omega^{-1}}\big) +
       L_f^j\,L_g^k\,s_A^*\big(\tilde\partial{}^i\{(\Phi_A)_j,(\Phi_A)_k\}_{\omega^{-1}}\big)
                          \notag
  \\ & \quad + s_A^*\{\pi^*f,L_g^i\}_{\omega^{-1}} -
       s_A^*\{\pi^*g,L_f^i\}_{\omega^{-1}} \notag \\
& \quad +
       s_A^*\{L_f^i,\Phi_A(L_g)\}_{\omega^{-1}} -
       s_A^*\{L_g^i,\Phi_A(L_f)\}_{\omega^{-1}} \notag \\ & \quad +
       L_g^j\,s_A^*\big(\tilde\partial{}^i\{\pi^*f,(\Phi_A)_j\}_{\omega^{-1}}\big)
       -
       L_f^j\,s_A^*\big(\tilde\partial{}^i\{\pi^*g,(\Phi_A)_j\}_{\omega^{-1}}\big)
                                                          \ ,
\label{c4}\end{align} 
with
\bea\label{l4}
\delta^\theta_fL^i_g(A):=L^i_g\big(A+\delta_f^\theta A\big) - L^i_g(A)
\eea
and
\bea\label{c1}
\delta_f^\theta A := s_A^*\{\pi^*f +
\Phi_A(L_f),\Phi_A\}_{\omega^{-1}} \ .
\eea
Then the gauge variations \eqref{c1} close the almost Poisson gauge
algebra
\bea\label{c2}
\big[\one+\delta_f^\theta,\one+\delta_g^\theta\big]A =
\delta^\theta_{t\,[\![f,g]\!]_\theta(A)}A 
\eea
and $(f,A)\mapsto A+\delta_f^\theta A$ is an almost Poisson gauge
transformation.
\end{proposition}
\begin{proof}
The proof of \eqref{c2} utilizes the
derivation property of the almost Poisson bracket together with the fact that,
since $\Phi_A=0$ on the constraint locus $\im(s_A)\subset T^*U$, the gauge
variations \eqref{c1} can be expressed as
\bea\label{c20}
\delta_f^\theta A = s_A^*\{\pi^*f,\Phi_A\}_{\omega^{-1}} + L_f^i\,
s_A^*\{(\Phi_A)_i,\Phi_A\}_{\omega^{-1}} \ ,
\eea
while the field dependent gauge parameter $[\![f,g]\!]_\theta\in
C_{\rm pol}^\infty(J^1T^*U)[[t]]$ from \eqref{c3} can be determined
as
\bea\label{l3}
t\,[\![f,g]\!]_\theta(A) = s_A^*\{\pi^*f,\pi^*g\}_{\omega^{-1}} -
L_f^i\,L_g^j\,s_A^*\{(\Phi_A)_i,(\Phi_A)_j\}_{\omega^{-1}} \ .
\eea
The details are defered to Appendix~\ref{app:p1proof}.

To show that the Jacobi identity \eqref{j3} is satisfied by
\eqref{c3}, we note that the left-hand side of \eqref{j3} in this case reads
\begin{align}
&
                {\sf Cyc}_{f,g,h}\Big( s_A^*\{\pi^*h,\pi^*s_A^*\{\pi^*f,\pi^*g\}_{\omega^{-1}}\}_{\omega^{-1}}-s_A^*\{\pi^*h,\pi^*s_A^*\{\Phi_A(L_f),\Phi_A(L_g)\}_{\omega^{-1}}\}_{\omega^{-1}}\notag\\
& \hspace{14mm}-s_A^*\{\Phi_A(L_h),\Phi_A(L_{t\,[\![f,g]\!]_\theta(A)})\}_{\omega^{-1}} +s_A^*\big(\tilde\partial^i\{\pi^*f,\pi^*g\}_{\omega^{-1}}\big)\,s_A^*\{\pi^*h+\Phi_A(L_h),(\Phi_A)_i\}_{\omega^{-1}}\notag\\
& \hspace{23mm} -\big(\delta_h^\theta
                                                                                                                                                                                                                            L^i_f\big)\,s_A^*\{(\Phi_A)_i,\Phi_A(L_g)\}_{\omega^{-1}} -\big(\delta^\theta_h L^i_g\big)\,s_A^*\{\Phi_A(L_f),(\Phi_A)_i\}_{\omega^{-1}}\notag\\
& \hspace{32mm}+s_A^*\{\pi^*s_A^*\{\pi^*h+\Phi_A(L_h),\Phi_A(L_f)\}_{\omega^{-1}},\Phi_A(L_g)\}_{\omega^{-1}}\notag\\
& \hspace{41mm}
                                                                                                                            +s_A^*\{\Phi_A(L_f),\pi^*s_A^*\{\pi^*h+\Phi_A(L_h),\Phi_A(L_g)\}_{\omega^{-1}}\}_{\omega^{-1}} \Big)
                                                                                                                            \
                                                                                                                                  . \label{j4}
\end{align}
Now we use the expression for $L_{t\,[\![f,g]\!]_\theta(A)}$ from
(\ref{c4}) and take into account the cyclic sum over
$f$, $g$ and $h$. Many terms cancel, while the remaining terms combine
using the formulas (\ref{r2}), (\ref{r5}), and (\ref{r6}) from
Appendix~\ref{app:useful}. After these simplifications the expression (\ref{j4}) becomes
\begin{align*}
&
                {\sf Cyc}_{f,g,h}\big( s_A^*\{\pi^*h,\{\pi^*f,\pi^*g\}_{\omega^{-1}}\}_{\omega^{-1}}+s_A^*\{\{\Phi_A(L_f),\Phi_A(L_g)\}_{\omega^{-1}},\pi^*h\}_{\omega^{-1}}\\
& \hspace{2cm} +s_A^*\{\{\Phi_A(L_g),\pi^*h\}_{\omega^{-1}},\Phi_A(L_f)\}_{\omega^{-1}}+s_A^*\{\{\pi^*h,\Phi_A(L_f)\}_{\omega^{-1}},\Phi_A(L_g)\}_{\omega^{-1}}\\
& \hspace{4cm}
                                                                                                                                                                                                                                                                                       +2\,s_A^*\{\{\Phi_A(L_f),\Phi_A(L_g)\}_{\omega^{-1}},\Phi_A(L_h)\}_{\omega^{-1}}\big)
                                                                                                                                                                    \
                                                                                                                                                                    .
\end{align*}
This expression now vanishes identically as a consequence of the Jacobi identity for the cosymplectic bracket. 
Thus the gauge transformations (\ref{c1}) form a Lie algebra.

Finally, the linearity and derivation requirements of
Definition~\ref{def:quasiPoissongauge} will follow from the
corresponding properties of the cosymplectic bracket and by constructing the
vector fields $L_f$ as $C^\infty(J^1T^*U)$-linear duals of the
one-forms $\dd f$, so that
\begin{align*}
L_{f\,g} = g\,L_f + f\,L_g
\end{align*}
and linearity in $f$ is immediate.
\end{proof}

\begin{remark}
By virtue of their role in cancelling the unwanted terms in the
commutator \eqref{t7}, we shall refer to the horizontal vector fields
$L_f$ as \emph{Lagrangian multipliers}.
\qen\end{remark}

\begin{remark}\label{rem:quasiPoissongauge}
One can compare with the Poisson gauge transformations \eqref{gtA},
and in particular see the effects of the modifications by the Lagrangian
multiplier vector fields, by using \eqref{PBq} together with the
notation of Remark~\ref{rem:Abracket} to write the almost Poisson gauge
transformations \eqref{c20} as
\begin{align*}
\delta_f^\theta A &= \dd f+ t\,s_A^*\gamma(\dd
                    f,\,\cdot\,)+t\,s_A^*\{\pi^*A,\pi^*f\}_{\underline{\theta}}
                    +
    t\,s_{A}^*\underline{\theta}^\otimes\big(\pi^*\DD
    A(L_f^\otimes),\pi^*\DD A\big) + \dd A(L_f,\,\cdot\,) \\[4pt]
  & \quad \, +t\,s_A^*\gamma^\otimes(\DD A,L_f) - t\,
    s_A^*\gamma\big(\DD A(L_f^\otimes),\,\cdot\,\big) \ ,
\end{align*}
for $f\in C^\infty(U)$ and $A\in\Omega^1(U)$, where $L_f^\otimes:=L_f\otimes \,\cdot\,$.
Similarly, the modification of the
almost Poisson bracket $\{f,g\}_\theta$ to the field dependent 
bracket \eqref{l3} can be written as
\begin{align*}
[\![f,g]\!]_\theta(A) &= s_A^*\{\pi^*f,\pi^*g\}_{\underline{\theta}} -
                        s_{A}^*\underline{\theta}\big(\pi^*\DD
                        A(L_f^\otimes),\pi^*\DD
                        A(L_g^\otimes)\big) - t^{-1}\, \dd
                        A(L_f,L_g) \\[4pt]
  & \quad \, +s_A^*\gamma\big(\DD
    A(L_f^\otimes),L_g\big) - s_A^*\gamma\big(\DD
    A(L_g^\otimes),L_f\big) \ ,
\end{align*}
 for $f,g\in C^\infty(U)$ and $A\in\Omega^1(U)$.
  \qen
\end{remark}

\subsection{Existence of Lagrangian multipliers}
\label{sec:Lagrangian}

Proposition~\ref{p1}, which establishes sufficiency conditions under
which local symplectic embeddings lead to almost Poisson gauge
transformations, is only meaningful of course if we can establish the
existence of the vector fields $L_f$, i.e. the existence of solutions
to the defining recursion equations determined by \eqref{c3}--\eqref{c1}. We
shall now demonstrate that they can be constructed recursively for
$f\in C^\infty(U)$ to any order in $t$.

We can immediately 
compute the leading order contribution to $L_f$ by using \eqref{l12}
to write the leading contribution to the second line of \eqref{t7} as
\bea\nonumber
s_A^*\big(\tilde\partial{}^j\{\pi^*f,\pi^*g\}_{\omega^{-1}}\big) \,
  s_A^*\{(\Phi_A)_j,\Phi_A\}_{\omega^{-1}} = t^2\,\dd A\big(\Pim(\dd
  f,\dd g,\,\cdot\,),\,\cdot\,\big) + O(t^3) \ .
  \eea
  At this order, this term can be cancelled by taking
\begin{equation}\label{lambda2}
L_f=-\tfrac{t^2}2 \, \Pim(\dd f,A,\,\cdot\,) + O(t^3)
\end{equation}
in \eqref{c1}, which from \eqref{c3} can be straightforwardly seen to yield
\begin{equation}\nonumber
[\![f,g]\!]_\theta(A) = \{f,g\}_\theta - t\, \Pim(\dd f,\dd g,A) + O(t^2) \ .
\end{equation}

To find the higher order contributions, we will work locally. The Lagrangian multipliers $L_f=L_f^i\,\partial_i\in
\mfX_{\rm pol}^{\tt h}(J^1T^*U)[[t]]$ should be linear in the gauge
parameter $f\in C^\infty(U)$, and will generally depend on the gauge fields $A\in\Omega^1(U)$ and
also on their derivatives $\partial A$, so we write their local forms
as
\begin{equation*}
L_f^i=t^2\,L^{ij}\big(x,A(x), \partial A(x)\big)\,\partial_j f(x) \ ,
\end{equation*}
for $x\in U$ and functions $L^{ij}\in
C_{\rm pol}^\infty(J^1T^*U)[[t]]$. In the following we write
$\partial_A^l=\partial/\partial A_l$ and $\partial_A^{k;l} =\partial
/\partial(\partial_k A_l)$ for derivatives in local coordinates on the
jet space. We can then formulate 

\begin{proposition}\label{p2}
  The solution to the recursion relations \eqref{c4} is given by
  \begin{align}
L^{ij}(A,\partial A)=\Lambdam^{ij}(A)+\sum_{n=1}^{\infty}\,
    (-t^2)^n\, & \Lambdam^{ik_1}(A)\, \partial_{k_1} A_{m_1}\,
    \Lambdam^{m_1k_2}(A) \notag \\ & \times \ \cdots \Lambdam^{m_{n-1}k_n}(A)\,
    \partial_{k_n}A_{m_n}\, \Lambdam^{m_nj}(A) \ , \label{m5}
\end{align}
where the function $\Lambdam^{ij}\in
s_A^*C_{\rm pol}^\infty(T^*U)[[t]] \subset C_{\rm pol}^\infty(J^1T^*U)[[t]]$ does not
depend on the derivatives $\partial A$ and satisfies local first
order differential equations on the jet space given by
\begin{align}
& t\,\big(\partial_A^i\Lambdam^{kj}-\partial_A^j\Lambdam^{ki}\big) +
   t^3\,\big(\Lambdam^{lj}\,\partial_l\Lambdam^{ki}-\Lambdam^{li}\,\partial_l\Lambdam^{kj}\big)
  \notag \\[4pt]
  & \hspace{1cm} = -\partial_A^k\underline{\theta}^{ij}(A) +
    t^2\,\big(\Lambdam^{kl}\,\partial_l\underline{\theta}^{ij}(A) +
    \Lambdam^{li}\,\partial_A^k\gamma_l^j (A)-\Lambdam^{lj}\,\partial_A^k\gamma_l^i (A)
  \label{m6} \\ & \hspace{5cm}
       +\underline{\theta}^{jl}(A)\,\partial_l\Lambdam^{ki} -
       \underline{\theta}^{il}(A)\,\partial_l\Lambdam^{kj} +
       \gamma_l^j (A)\,\partial_A^l\Lambdam^{ki} -
       \gamma_l^i (A)\,\partial_A^l\Lambdam^{kj}\big) \notag \\
  & \hspace{8cm}
    +t^4\,\big(\gamma_m^l (A)\,\Lambdam^{mi}\,\partial_l\Lambdam^{kj}
    - \gamma_m^l (A)\,\Lambdam^{mj}\,\partial_l\Lambdam^{ki}\big) \
    , \notag
\end{align}
with $\underline{\theta}^{ij}(A):=s_{A}^*\underline{\theta}^{ij}$ and
$\gamma_i^j(A):=s_{A}^*\gamma^j_i$. 
\end{proposition}

\begin{proof}
First we introduce
\begin{equation}\label{m7}
L^{ij}(A,\partial A)=\Lambdam^{ij}(A)+\bar\Lambdam^{ij}(A,\partial A)
\ ,
\end{equation}
where $\Lambdam^{ij}(A)$ is a function of the gauge fields $A$ only,
whereas $\bar\Lambdam^{ij}(A,\partial A)$ depends on both gauge
fields $A$ and their first order derivatives $\partial A$. To obtain
an equation for the function $\Lambdam^{ij}(A)$, we substitute
(\ref{m7}) in (\ref{c4}) and collect the terms which contain only
first order derivatives of gauge parameters $\partial_i f\,\partial_j
g$ but do not contain derivatives of the gauge fields $\partial
A$.

The left-hand side of \eqref{c4} reads
\begin{align}\label{m8}
L^i_{t\,[\![f,g]\!]_\theta(A)} = t^2\,L^{im}\,\partial_m\big(s_A^*\{\pi^*f,\pi^*g\}_{\omega^{-1}} -
L_f^j\,L_g^k\,s_A^*\{(\Phi_A)_j,(\Phi_A)_k\}_{\omega^{-1}} \big) \ ,
\end{align}
and it contributes
$t^2\,\Lambdam^{lk}\,\partial_k\underline{\theta}^{ij}(A)$
to \eqref{m6}. The contributions on the right-hand side of \eqref{c4} from
\begin{align}
\delta^\theta_fL_g^i&=t^2\,\partial_A^l L^{ij}
                      \big(s_A^*\{\pi^*f,(\Phi_A)_l\}_{\omega^{-1}}+t^2\,L^{km}\,s_A^*\{(\Phi_A)_k,(\Phi_A)_l\}_{\omega^{-1}}\,
                      \partial_kf\big)\,\partial_j g \notag \\
& \quad \, + t^2\,\partial_A^{p;l} L^{ij}\,\partial_p                                                                                                                                  \big(s_A^*\{\pi^*f,(\Phi_A)_l\}_{\omega^{-1}}+t^2\,L^{km}\,s_A^*\{(\Phi_A)_k,(\Phi_A)_l\}_{\omega^{-1}}\, \partial_kf\big)\,\partial_j g \ ,\label{m9}
\end{align}
and from
\begin{align*}
s_A^*\big(\tilde\partial{}^i\{\pi^*f,\pi^*g\}_{\omega^{-1}}\big) =
  t\,\partial_A^i\underline{\theta}^{jk}(A)\,\partial_jf\,\partial_kg
  \ ,
\end{align*}
together yield
\begin{align*}
t\,\big(\delta^i_l+t\,\gamma^i_l(A)\big)\partial_A^l\Lambdam^{kj} -
  t\,\big(\delta^j_l+t\,\gamma^j_l(A)\big)\partial_A^l\Lambdam^{ki} +
  \partial_A^k\underline{\theta}^{ij}(A) \ . 
\end{align*}
The terms coming from
\begin{align*}
L_g^j\,s_A^*\big(\tilde\partial{}^i\{\pi^*f,(\Phi_A)_j\}_{\omega^{-1}}\big)
= t^3\,L^{lj}\,\partial^i_A\gamma^k_l(A)\,\partial_kf\,\partial_jg-t^3\,L^{lj} \,\partial_A^i\underline{\theta}^{km}(A) \,\partial_k f\,\partial_jg\,\partial_mA_l
\end{align*}
are given by
\begin{align*}
t^2\Lambdam^{lj}\,\partial_A^k\gamma_l^i(A) -
  t^2\Lambdam^{li}\,\partial_A^k\gamma_l^j(A) \ .
\end{align*}
There is no contribution from 
$L_f^j\,L_g^k\,s_A^*\big(\tilde\partial{}^i\{(\Phi_A)_j,(\Phi_A)_k\}_{\omega^{-1}}\big)$
on the right-hand side of \eqref{c4} since all of its terms contain
derivatives of the gauge fields $\partial A$. Finally, the
contributions from the second and third lines of \eqref{c4}:
\begin{align}
s_A^*\{\pi^*f,L_g^i\}_{\omega^{-1}} &=
                                      t^3\,\underline{\theta}^{kl}(A)\,\partial_kf\,\partial_l(L^{ij}\,\partial_jg)
                                      \ , \label{m13} \\[4pt]
  s_A^*\{L_f^i,\Phi_A(L_g)\}_{\omega^{-1}} &=
                                             t^4\,L^{kp}\,\big((\delta_k^l+t\,\gamma_k^l(A))\,\partial_l(L^{ij}\,\partial_jf)
                                             -
                                             t\,\underline{\theta}^{ml}(A)\,\partial_mA_k\,\partial_l(L^{ij}\,\partial_jf)\big)\,
                                             \partial_pg \ , \notag
\end{align}
read as
\begin{align*}
t^2\,\underline{\theta}^{il}(A)\,\partial_l\Lambdam^{kj}-t^2\,\underline{\theta}^{jl}(A)\,\partial_l\Lambda^{ki}-t^3\,\big(\delta^l_m+t\,\gamma^l_m(A)\big)\,\Lambdam^{mi}\,\partial_l\Lambdam^{kj}+t^3\,\big(\delta^l_m+t\,\gamma^l_m(A)\big)\,\Lambdam^{mj}\,\partial_l\Lambdam^{ki}
  \ .
\end{align*}
Altogether the contributions from such terms results in the
differential equation
\eqref{m6}. 
                                             
The contributions to (\ref{c4}) containing second order derivatives of
gauge parameters $\partial_i\partial_kf\,\partial_j g$ and $\partial_i
f\,\partial_j\partial_kg$ should vanish separately. Let us analyse the
contribution with $\partial_i\partial_kf\,\partial_j g$. The
contribution from the left-hand side of \eqref{c4} is given by
\eqref{m8} and reads
\begin{equation*}
t^3\,\big( L^{mk}\,\underline{\theta}^{ij}(A)-t^3\,L^{mk}\,L^{pi}\,L^{lj}\,s_A^*\{(\Phi_A)_p,(\Phi_A)_l\}_{\omega^{-1}}\big)\,\partial_i\partial_kf\,\partial_jg
\ .
\end{equation*}
On the right-hand side of \eqref{c4}, the contribution from
$\delta^\theta_f(L_g^i)-\delta^\theta_g(L_f^i)$ comes from the second
line of (\ref{m9}) and reads
\begin{equation*}
t^2\,\partial_A^{k;l}
L^{mj}\,\big(\delta^i_l+t\,\gamma^i_l(A)-t\,\underline{\theta}^{in}(A)\,\partial_nA_l+t^2\,L^{pi}\,s_A^*\{(\Phi_A)_p,(\Phi_A)_l\}_{\omega^{-1}}\big)\,
\partial_i\partial_kf\,\partial_jg \ .
\end{equation*}
The only remaining terms on the right-hand side of \eqref{c4} which
contain second order derivatives of $f$ and $g$ come from the second
and third lines, and from (\ref{m13}) one finds
\begin{equation*}
t^3\,\big(L^{pk}\,\underline{\theta}^{ij}(A)
+t\,L^{pi}\,L^{nj}\,(\delta^k_n+t\,\gamma^k_n(A))+t^2\,L^{pi}\,L^{nj}\,\underline{\theta}^{mk}(A)\,\partial_mA_n\big)\,\partial_i\partial_kf\,\partial_jg
\ .
\end{equation*}
Thus demanding that \eqref{c4} hold for such contributions leads to the equations
\begin{align*}
&\big(\partial_A^{k;l}L^{mj}+t^2\,L^{mk}\,L^{lj}\big)\,
                 \big(\delta^i_l+t\,\gamma^i_l(A)-t\,\underline{\theta}^{in}(A)\,\partial_nA_l+t^2\,L^{pi}\,s_A^*\{(\Phi_A)_p,(\Phi_A)_l\}_{\omega^{-1}}\big)
  \\
& \qquad +
       \big(\partial_A^{i;l}L^{mj}+t^2\,L^{mi}\,L^{lj}\big)\,\big(\delta^k_l+t\,\gamma^k_l(A)-t\,\underline{\theta}^{kn}(A)\,\partial_nA_l+t^2\,L^{pk}\,s_A^*\{(\Phi_A)_p,(\Phi_A)_l\}_{\omega^{-1}}\big)=0
       \ . 
\end{align*}
This is satisfied if
\begin{equation}\label{eqlambda}
\partial_A^{k;l}L^{ij}+t^2\,L^{ik}\,L^{lj}=0 \ ,
\end{equation}
or taking into account (\ref{m7}) we may rewrite it as
\begin{equation*}
\partial_A^{k;l}
\bar\Lambdam^{ij}+t^2\,\big(\Lambdam^{ik}+\bar\Lambdam^{ik}\big)\,
\big(\Lambdam^{lj}+\bar\Lambdam^{lj}\big)=0 \ .
\end{equation*}
This equation is solved by
\begin{equation*}
\bar\Lambdam^{ij}(A,\partial A)=\sum_{n=1}^{\infty}\,
(-t^2)^n\,\Lambdam^{ik_1}(A)\,\Lambdam^{m_1k_2}(A) \cdots
\Lambdam^{m_{n-1}k_n}(A)\,\Lambdam^{m_nj}(A)\, \partial_{k_1}
A_{m_1}\cdots \partial_{k_n}A_{m_n} \ ,
\end{equation*}
which yields the expansion (\ref{m5}).

To complete the proof we need to verify that the remaining terms in
(\ref{c4}) cancel due to (\ref{eqlambda}) and (\ref{m6}). This is
tedious but completely straightforward: for example, the term 
\begin{equation*}
-t^6\,L^{mk}\,\partial_k\big(L^{li}\,L^{pj}\big)\,s_A^*\{(\Phi_A)_l,(\Phi_A)_p\}_{\omega^{-1}}\,\partial_if\,\partial_jg
\end{equation*}
 coming from (\ref{m8}) precisely cancels the two terms
 \begin{equation*}
t^4\,\partial_A^{k;l}
L^{mj}\,\partial_kL^{pi}\,s_A^*\{(\Phi_A)_p,(\Phi_A)_l\}_{\omega^{-1}}\,
(\partial_if\,\partial_jg-
\partial_ig\,\partial_jf )
\end{equation*}
coming from (\ref{m9}). 
\end{proof}

It is now straightforward to construct the functions $\Lambdam^{ij}$
recursively from \eqref{m6}, and we have
\begin{corollary}\label{cor:LambdamA}
The functions $\Lambdam^{ij}\in
s_A^*C_{\rm pol}^\infty(T^*U)[[t]]$ can be calculated order by order
in $t$ as the formal power series
\begin{align*}
\Lambdam^{ij}(A) = -\frac12\,\Pim^{ijk}\, A_k-\sum_{n=1}^\infty\, \frac{t^{n}}{n+2} \,
  \Upm^{ijl|k_1\cdots k_{n}} \, A_l\,A_{k_1}\cdots A_{k_{n}} \ ,
\end{align*}
where the functions $\Upm^{ijl|k_1\cdots k_{n}}\in C^\infty(U)$ are
skew-symmetric in their indices $jl$, and are
given by explicit expressions involving the components of the
almost Poisson bivector $\theta$, its Jacobiator~$\Pim$, and their
derivatives.
\end{corollary}
\begin{proof}
We recall the expansions
\eqref{eq:gammaijxp} (satisfying the local symplectic embedding
equations \eqref{gtrec} and \eqref{h5}) and \eqref{l12}, and write
$\Lambdam^{ij}(A)$ as a formal power series of the form
\begin{align*}
  \Lambdam^{ij}=
  \sum_{n=1}^\infty\, t^{n-1} \, \Lambdam_n^{ij} \qquad \mbox{with}
  \quad \Lambdam^{ij}_n=\Lambdam^{ij|k_1\cdots k_n}\,A_{k_1}\cdots
  A_{k_n} \ ,
\end{align*}
where $\Lambdam^{ij|k_1\cdots k_n}\in C^\infty(U)$ for each
$n\geq1$. At first non-trivial order the differential equation \eqref{m6} reads
\begin{equation*}
\partial^i_A\Lambdam^{kj}_1-\partial^j_A\Lambdam^{ki}_1=\Pim^{ijk} \ ,
\end{equation*}
which has solution $
  \Lambdam^{ij}_1=-\tfrac12\, \Pim^{ijk}\, A_k
$ as expected from \eqref{lambda2}.
At higher orders $t^{n-1}$ for $n\geq2$, the differential equation (\ref{m6}) can be
schematically represented as 
\begin{equation}\label{m17}
\partial^i_A\Lambdam^{kj}_n-\partial^j_A\Lambdam^{ki}_n=\Upm_n^{kij}
\qquad \mbox{with} \quad \Upm_n^{kij}=\Upm^{kij|k_1\cdots
  k_{n-1}}\,A_{k_1}\cdots A_{k_{n-1}} \ ,
\end{equation}
where on the right-hand side $\Upm_n^{kij}$ is constructed from
the previously determined lower order solutions $\Lambdam^{kj}_m$ with
$m<n$, and $\Upm^{kij|k_1\cdots
  k_{n-1}}\in C^\infty(U)$. The integrability condition for the equation (\ref{m17}) reads 
\begin{equation}\label{m18}
\partial^k_A\Upm_n^{lij}+\partial^i_A\Upm_n^{ljk}+\partial^j_A\Upm_n^{lki}=0
\ .
\end{equation}

At second order,
after careful calculation and simplification one obtains
\begin{align}
\Upm^{ilj|k}&=
                -2\,\theta^{lj|ki}+\tfrac12\,\Pim^{ikm}\,\partial_m\theta^{lj}+\tfrac14\,
                \Pim^{ilm}\,\partial_m\theta^{jk}-\tfrac14\,
                \Pim^{ijm}\,\partial_m\theta^{lk} \notag \\
& \quad +\tfrac14\, \Pim^{jkm}\,\partial_m\theta^{li}-\tfrac14\,
                                                              \Pim^{lkm}\,\partial_m\theta^{ji}+\tfrac12\,
                                                              \theta^{lm}\,\partial_m\Pim^{ijk}-\tfrac12\,\theta^{jm}\,\partial_m\Pim^{ilk}
                                                              \ , \label{l7}
\end{align}
where the function $\theta^{lj|ki}$ determines the order $t^2$
contribution to the expansion (\ref{l12}). The integrability condition
(\ref{m18}) reads 
\begin{equation}\label{l8}
\Upm^{ilj|k}+\Upm^{ikl|j}+\Upm^{ijk|l}=0 \ .
\end{equation}
Substituting the explicit form of $\Upm^{ilj|k}$ from (\ref{l7}) into
(\ref{l8}) we arrive at
\begin{equation}\label{l9}
 2\,\big( \theta^{lj|ki}+\theta^{kl|ji}+\theta^{jk|li}\big)-F^{ljki}=0 \ ,
\end{equation}
with
\begin{align*}
  F^{ljki}&=\Pim^{kmi}\,\partial_m\theta^{lj}+\Pim^{jmi}\,\partial_m\theta^{kl}+\Pim^{lmi}\,\partial_m\theta^{jk}+\theta^{km}\,\partial_m\Pim^{lji}\\
& \quad
                                                                                                                                                        +\theta^{jm}\,\partial_m\Pim^{kli}+\theta^{lm}\,\partial_m\Pim^{jki}+\tfrac{1}{2}\,\Pim^{ljm}\,\partial_m\theta^{ki}+\tfrac{1}{2}\,\Pim^{klm}\,\partial_m\theta^{ji}+\tfrac{1}{2}\,\Pim^{jkm}\,\partial_m\theta^{li}
                                                                                                                                                        \
                                                                                                                                                        .
\end{align*}
The expression (\ref{l9}) is exactly the equation from which the
function $\theta^{lj|ki}$ was determined in
\cite{Kupriyanov:2018yaj}, as given explicitly
in \eqref{Theta4}, whose integrability condition is precisely the
integrability identity \eqref{eq:Piint} for the Jacobiator $\Pim$. So the
integrability condition (\ref{l8}) is indeed satisfied and  
\begin{align*}
 \Lambdam_2^{il}=-\tfrac13\, \Upm^{ilj|k}\,A_j\,A_k \ .
\end{align*}

This recursive construction can be extended in exactly the same way to
higher order calculations. The integrability condition (\ref{m18}) is
always satisfied as a consequence of the definition of the functions
$\Upm_n^{kij}$ from (\ref{m6}) and the construction of the symplectic
embedding for the almost Poisson structure $\theta$
from~\cite{Kupriyanov:2018yaj}.
\end{proof}

\newsection{Homotopy algebra actions}
\label{sec:Linfinity}

The reader versed in the BV formalism will have already
noticed a strong resemblance between our approach to almost Poisson
gauge theories based on symplectic embeddings and the approach to
generalized gauge symmetries based on $L_\infty$-algebras. In this
framework, the classical BRST construction amounts to quotienting a
Poisson algebra of smooth functions by the ideal of functions which
vanish on a constraint locus through a process of homological
reduction to achieve the gauge closure condition. In this language,
the concept of field dependent gauge transformation that we introduced
in Section~\ref{sec:formulationquasi} is very natural.

In this section we will show that our almost Poisson gauge symmetries,
which do not arise from any Lie algebra action, in fact arise from actions
of $L_\infty$-algebras. For this, we provide a dictionary between the symplectic
embeddings proposed in the present paper and the corresponding
$L_\infty$-algebras of generalized gauge
transformations~\cite{Fulp:2002kk}. This can be achieved by
identifying the gauge
variations defined by (\ref{c1}), and the closure bracket which is
determined in (\ref{c3}), with the corresponding objects defined in
terms of $L_\infty$-algebras. As an interesting application of these constructions, we obtain a new perspective on the deformation quantization of the exterior algebra of differential forms on an arbitrary almost Poisson manifold, whose semi-classical limit is described by a $P_\infty$-algebra.

\subsection{$L_\infty$-algebras and generalized gauge symmetries}

Let us start by fixing definitions and conventions. A (\emph{flat})
\emph{$L_\infty$-algebra} (also known as a \emph{strong homotopy Lie
  algebra}) is a graded vector space $V$ over a field of
characteristic zero together with a sequence of linear maps
$\ell_n:V^{\otimes n}\to V$ of degree $2-n$, for $n=1,2,\dots$, which
satisfy two properties. Firstly, $\ell_n$
are graded skew-symmetric,
\begin{align*}
\ell_n(\dots, v_i, v_{i+1},\dots )= -(-1)^{|v_i|\,|v_{i+1}|} \,
  \ell_n(\dots, v_{i+1}, v_i,\dots) 
\end{align*}
for all $v_1,\dots,v_n\in V$, where $|v|$ denotes the degree of a
homogeneous element $v\in V$. Secondly, the linear map defined by
\begin{align}\label{eq:altmap}
\Lie_n(v_1\otimes\cdots\otimes v_n) := \sum_{k=1}^n\,
  \frac{(-1)^{k\,(n-k)}}{k!\,(n-k)!} \,
  \ell_{n-k+1}\big(\ell_k(v_1,\dots, v_k),
  v_{k+1},\dots, v_n\big) 
\end{align}
vanishes on the image of the alternatization map
\begin{align*}
{\sf Alt}_n(v_1\otimes\cdots\otimes v_n) = \sum_{\sigma\in S_n} \, {\rm
  sgn}(\sigma) \, v_{\sigma(1)}\otimes \cdots \otimes v_{\sigma(n)} \ ,
\end{align*}
where the sum runs over permutations $\sigma$ of degree $n$ with sign
${\rm sgn}(\sigma)$. 
We denote the alternatization of the map
\eqref{eq:altmap} by
$\CJ_n=\Lie_n\circ{\sf Alt}_n:V^{\otimes n}\to V$. Then the relations $\CJ_n=0$ are called \emph{homotopy Jacobi
  identities}; when evaluated explicitly on elements they involve 
cumbersome sign factors which are determined by the usual Koszul sign
rules for skew-symmetric maps of graded vector spaces.

In particular, the first relation $\CJ_1=0$ implies that $\ell_1$ is a
differential making $(V,\ell_1)$ into a cochain complex, while $\CJ_2=0$ implies that $\ell_1$ is a (graded)
derivation of the bracket $\ell_2$, i.e. $\ell_2$ is a cochain
map. The third relation $\CJ_3=0$ then implies that the bracket $\ell_2$ obeys the Jacobi identity
up to exact terms, i.e. $\ell_2$ induces a (graded) Lie bracket on the
cohomology of $\ell_1$. Differential graded Lie algebras can be regarded as
$L_\infty$-algebras where $\ell_n=0$ for all $n\geq3$.

Generalized gauge symmetries of classical field theories on a manifold
$M$ which are irreducible can be encoded
in $2$-term $L_\infty$-algebras~\cite{Fulp:2002kk,Hohm:2017pnh}. For the purposes of the present paper,
the underlying graded vector space encoding the kinematical data on an
open subset $U\subseteq M$ has non-zero homogeneous subspaces
only in degrees~$0$ and~$1$, and is given by
\begin{align} \label{eq:2term}
V = C^\infty(U) \oplus \Omega^1(U) \ ,
\end{align}
with the grading provided by differential form degree. The gauge
variations are then encoded by an $L_\infty$-structure
$(\ell_n)_{n=1}^\infty$ on $V$ with
$\ell_{n+1}(f, A^{\otimes n})\in\Omega^1(U)$, for $f\in
C^\infty(U)$ and $A\in\Omega^1(U)$, as the
formal power series
\begin{align}
\delta_f A &= \sum_{n=0}^\infty\, \frac{t^n}{n!} \,
  (-1)^{\frac{n\,(n-1)}2} \, \ell_{n+1}(f, A^{\otimes n})
  \nonumber \\[4pt] &=
  \ell_1(f) + t\,\ell_2(f, A)-\tfrac{t^2}2\,\ell_3(f,
  A, A) + O(t^3) \label{eq:Linftygt}
\end{align}
in $\Gamma_{\rm pol}(J^1T^*U)[[t]]$. For ``standard'' gauge symmetries
which only involve a finite number of brackets, one can set $t=1$ and
avoid the use of formal power series as well as jet coordinates
altogether, but the constructions of this paper necessitate their
usage in general.

For the graded vector space \eqref{eq:2term}, since $\ell_n$ is of
degree $2-n$, it vanishes unless its entries contain at least one
and at most two gauge parameters. Hence the only non-trivial
homotopy Jacobi identities for the $L_\infty$-structure on $V$ involve
two and three gauge parameters, that is, ${\cal J}_{n+2}(f, g, 
A^{\otimes n} )=0$ and $\CJ_{n+3}(f,g,h,A^{\otimes n})=0$.
It was shown in~\cite{Berends:1984rq,Fulp:2002kk,Hohm:2017pnh} that
the homotopy Jacobi identities 
involving two gauge parameters imply the
closure of the symmetry variations 
\begin{equation*}
                      [\one+\delta_{f},\one+\delta_g] A
                      =\delta_{t\,[\![f,g]\!]( A)} A \ ,
\end{equation*}
with the formal power series
\begin{align}
[\![f,g]\!]( A)  &=-\sum_{n= 0}^\infty \, \frac{t^n}{ n!} \, 
            (-1)^{\frac{n\,(n-1)}{ 2}}
            \, \ell_{n+2}(f, g, A^{\otimes n} ) \nonumber \\[4pt] &=
            -\ell_2(f, g) - t\, \ell_3(f, g, A) +
            \tfrac{t^2}2\,\ell_4(f, g, A, A) +
            O(t^3) \label{eq:Linftyclosure}
          \end{align}
          in $C^\infty_{\rm pol}(J^1T^*U)[[t]]$. The homotopy Jacobi identities
          involving three gauge parameters then guarantee that the strict Jacobi identities
          \begin{align*}
{\sf Cyc}_{f,g,h}\,\big[\one+\delta_h,[\one+\delta_f,\one+\delta_g]\big]A=0
          \end{align*}
          hold for any triple of gauge variations.

The following explicit formulas are useful for concrete
checks of the homotopy Jacobi identities for gauge transformations.
\begin{lemma}\label{lem:homotopygauge}
For $f,g,h\in C^\infty(U)$ and $A\in\Omega^1(U)$, the
$L_\infty$-relations ${\cal J}_{n+2}(f, g, 
A^{\otimes n} )=0$ and $\CJ_{n+3}(f,g,h,A^{\otimes n})=0$ for $n\geq1$ 
are given explicitly by
\begin{align}
\CJ_{n+2}\big(f,g,A^{\otimes n}\big) &=
                                       \ell_1\big(\ell_{n+2}(f,g,A^{\otimes
                                       n})\big) -
                                       (-1)^{n}\,\ell_{n+2}\big(\ell_1(f),g,A^{\otimes
                                       n}\big)-(-1)^{n}\,\ell_{n+2}\big(f,\ell_1(g),A^{\otimes
                                       n}\big) \notag\\
& \quad \, +\sum_{i=1}^{n}\,(-1)^{(i+1)\,(n-i+1)} \,
                                                          \bigg[\binom{n}{i-1}\,\ell_{n-i+2}\big(\ell_{i+1}(f,g,A^{\otimes i-1}),A^{\otimes n-i+1}\big) \notag\\
& \quad \, \hspace{4cm} +(-1)^{i} \,
                                                                                                                                                                   \binom{n}{i}\,\ell_{n-i+2}\big(\ell_{i+1}(f,A^{\otimes i}),g,A^{\otimes n-i}\big) \label{Jnfg}\\
& \quad \, \hspace{4.5cm} -(-1)^{i}\,\binom{n}{i}\,\ell_{n-i+2}\big(\ell_{i+1}(g,A^{\otimes
   i}),f,A^{\otimes n-i}\big)\bigg]
   \ , \notag
\end{align}
and
\begin{align}
  & \CJ_{n+3}\big(f,g,h,A^{\otimes n}\big) \notag \\[4pt]
  & \hspace{0.5cm} = {\sf
                                         Cyc}_{f,g,h}\bigg((-1)^{n}\,\ell_{n+3}\big(\ell_1(f),g,h,A^{\otimes
                                         n}\big)
    +(-1)^n\,\ell_{2}\big(\ell_{n+2}(f,g,A^{\otimes n}),h\big) \notag\\
& \hspace{2.5cm} +\sum_{i=1}^{n}\,(-1)^{(i+1)\,(n-i)}\, \bigg[\binom{n}{i}\,\ell_{n-i+3}\big(\ell_{i+1}(f,A^{\otimes i}),g,h,A^{\otimes n-i}\big) \label{Jnfgh}\\
& \quad \, \hspace{5.5cm} -
                                                                                                                                                                                                   (-1)^{i}\,\binom{n}{i-1}\,\ell_{n-i+3}\big(\ell_{i+1}(f,g,A^{\otimes
                                                                                                                                                                                                   i-1}),h,A^{\otimes
                                                                                                                                                                                                   n-i+1}\big)
                                                                                                                                                             \bigg]
                                                                                                                                                             \bigg)
                                                                                                                                                             \
                                                                                                                                                             . \notag
\end{align}
\end{lemma}
\begin{proof}
The alternatization of the sum $\Lie_{n+2}(f,g,A^{\otimes n})$ in
(\ref{eq:altmap}) involves, for fixed $n$ and $k$, a sum over $\binom{n+2}{k}=(n+2)!/k!\,(n+2-k)!$
inequivalent splittings of permutations $\sigma\in S_{n+2}$~\cite{Hohm:2017pnh}. If $k=1$ the contribution to
(\ref{eq:altmap}) is $\ell_{n+2}(\ell_1(f),g,A^{\otimes
  n})+\ell_{n+2}(f,\ell_1(g),A^{\otimes n})$, and when $k=n+2$ the
contribution is $\ell_1(\ell_{n+2}(f,g,A^{\otimes n}))$; note that
terms such as $\ell_{n+2}(f,g,\ell_1(A),A^{\otimes n-1})$ do not
appear as $\ell_1(A)=0$ for degree reasons. 
If $k\neq 1,n+2$ then, taking into account that there are $n$ identical
entries $A$, the sum over inequivalent splittings of permutations
$\sigma\in S_{n+2}$  will contain $\binom{n}{k-2}$ contributions of the type
$\ell_{n-k+3}(\ell_k(f,g,A^{\otimes k-2}),A^{\otimes n-k+2})$,
$\binom{n}{k-1}$ contributions of each type
$\ell_{n-k+3}(\ell_k(f,A^{\otimes k-1}),g,A^{\otimes n-k+1})$ and $\ell_{n-k+3}(\ell_k(g,A^{\otimes k-1}),f,A^{\otimes n-k+1})$, 
and also $\binom{n}{k}$ elements $\ell_{n-k+3}(\ell_k(A^{\otimes
  k}),f,g,A^{\otimes n-k})$; these latter contributions however vanish
as $\ell_k(A^{\otimes k})=0$ for degree reasons. Altogether, the total
number of such contributions is
\begin{equation*}
\binom{n}{k-2}+\binom{n}{k-1}+\binom{n}{k-1}+\binom{n}{k}=\binom{n+2}{k}
\ ,
\end{equation*}
which is exactly the number of all contributions to the sum over
$\sigma\in S_{n+2}$. This results in
\eqref{Jnfg}, where the sign factor $(-1)^i$ in the last two lines appears from moving $f$ and $g$ from the second argument to the $(i{+}1)$-th argument of the  $(n-i+2)$-bracket. Following the same line of argument we find
\eqref{Jnfgh}. 
\end{proof}

\subsection{$L_\infty$-structures on Poisson gauge algebroids}
\label{sec:PoissonLinfty}

It is instructive to start with the simpler case of Poisson gauge transformations
from Section~\ref{sec:Poissongauge}, where we know
$[\![f,g]\!]( A)  =\{f,g\}_\theta$ to all orders and a Poisson gauge
symmetry corresponds to a Lie algebra structure on the graded vector
space $V$. Geometrically, this structure is
encoded in the action algebroid
\begin{align*}
C^\infty(U)\ltimes \Omega^1(U) \longrightarrow \Omega^1(U)
\end{align*}
corresponding to the Lie algebra 
$(C^\infty(U),\{\,\cdot\,,\,\cdot\,\}_\theta)$ and the Lie module
$\Omega^1(U)$ over $C^\infty(U)$; for the moment we drop the formal
power series extension to streamline the presentation. This is the Lie
algebroid with trivial vector bundle $C^\infty(U)\times\Omega^1(U)$
over $\Omega^1(U)$ (regarded as a Fr\'echet space with the weak
Whitney $C^\infty$-topology), whose anchor map
$\rho:C^\infty(U)\times\Omega^1(U)\to T\Omega^1(U)$ sends a pair
$(f,A)$ to the gauge variation $\delta^\theta_fA$ of
Definition~\ref{def:Poissongauge}, and whose Lie bracket is induced by the
Poisson bracket $\{\,\cdot\,,\,\cdot\,\}_\theta$ and the gauge
variations. Generally, any Lie algebroid gives rise to a cochain
complex (its Chevalley-Eilenberg algebra), and in the case of an action algebroid this yields the BRST
complex of a classical field theory, which is dual to a gauge
$L_\infty$-algebra (see e.g.~\cite{Jurco:2018sby}).

In fact, the proper reinstatement of the formal power series
extension, as required by the symplectic embedding approach, into this Lie
algebroid perspective necessitates the use of $L_\infty$-algebras. For
this, let us note the following useful geometric way of thinking 
about the `Taylor expansion' of the twisting one-forms
$s_A^*\gamma(\dd f,\,\cdot\,)$ in \eqref{gtA} which are induced by
\eqref{eq:gammaijxp}. Without loss of generality we can assume that
the functions $\gamma_i^{j|i_1\cdots
  i_n} \in C^\infty(U)$ are symmetric in their last $n$ indices, and introduce the
$(n+1,1)$-tensor fields
$\gamma^{(n)}\in\Omega^1\big(U,\mfX(U)\otimes\mathfrak{X}(U)^{\odot n}\big)$
which in local coordinates read as
\begin{align*}
  \gamma^{(n)} = n!\,\gamma_i^{j|i_1\cdots
  i_n}(x) \,
\partial_j\otimes(\partial_{i_1}\odot\cdots \odot\partial_{i_n})\otimes \dd
  x^i \ .
\end{align*}
Then
\begin{align}\label{eq:sAgammaformal}
s_A^*\gamma(\dd f,\,\cdot\,) = \sum_{n=1}^\infty\, \frac{t^{n-1}}{n!} \,
  \gamma^{(n)}\big(\dd f,A^{\otimes n}\big)
\end{align}
as a formal power series in $\Omega^1(U)[[t]]$.
\begin{proposition}\label{prop:PoissonLinfty}
Let $(T^*M,\omega)$ be a local symplectic embedding of a Poisson manifold $(M,\theta)$, and let $U\subseteq M$ be an
open subset. Then there is an $L_\infty$-structure
$(\ell_n^{\,\theta})_{n=1}^\infty$ on
$C^\infty(U)\oplus\Omega^1(U)$, unique up to
$L_\infty$-quasi-isomorphism, which induces the Poisson gauge
algebroid constructed by Proposition~\ref{APhi} with non-vanishing
maps on coincident degree~$1$ entries given by
\begin{align*}
  \ell^{\,\theta}_1(f) & = \dd f \ , \\[4pt]
  \ell^{\,\theta}_2(f, g)& = -\{f,g\}_\theta \ , \\[4pt]
  \ell^{\,\theta}_2(f, A) & =  \{A,f\}_\theta +\gamma^{(1)}(\dd f,A) \ , \\[4pt]
  \ell_{n+1}^{\,\theta}\big(f, A^{\otimes n}\big) & =
                                          (-1)^{\frac{n\,(n-1)}2}
                                          \, 
                                          \gamma^{(n)}\big(\dd
                                          f,A^{\otimes n}\big) \ ,
\end{align*}
and skew-symmetry imposed by definition, for $n\geq2$, $f,g\in
C^\infty(U)$, and $A\in\Omega^1(U)$. 
\end{proposition}
\begin{proof}
The brackets $\ell_n^{\,\theta}$ follow from comparing \eqref{closure}
with \eqref{eq:Linftyclosure} and \eqref{gtA} with \eqref{eq:Linftygt}
order by order in $t$ using \eqref{eq:sAgammaformal}. The fact that our symplectic embeddings define
Poisson gauge algebras in the sense of
Definition~\ref{def:Poissongauge}, as we proved in
Proposition~\ref{APhi}, then implies that the statement is a special instance
of~\cite[Theorem~2]{Fulp:2002kk} for the case of gauge transformations
which are generated by Lie algebra actions
(see~\cite[Section~6]{Fulp:2002kk}).

The uniqueness statement follows
from noting that different completions of the brackets to
non-coincident gauge field entries, using the homotopy Jacobi
identities, are related by invertible field redefinitions
$\chi:\Gamma_{\rm pol}(J^1T^*U)[[t]]\to\Gamma_{\rm pol}(J^1T^*U)[[t]]$, called Seiberg-Witten
maps, which leave invariant the
gauge parameters $f$ and define deformations of the gauge fields:
$\chi(A)\big|_{t=0} = A$ for $A\in\Omega^1(U)$. We define the new gauge transformations
$\chi(A)\mapsto \chi(A)+\hat\delta_f^{\theta}\chi(A)$ by setting
\begin{align*}
\hat\delta_f^\theta := \chi\circ\delta_f^\theta\circ\chi^{-1} \ .
\end{align*}
Then the field redefinition maps gauge orbits onto
gauge orbits:
\begin{align*}
\chi\big(A+\delta_f^\theta A\big) = \chi(A) +
  \hat\delta_f^\theta\chi(A) \ ,
\end{align*}
and it preserves the Poisson gauge algebra:
\begin{align*}
\big[\one+\hat\delta_f^\theta,\one+\hat\delta_g^\theta\big] \chi(A) =
  \chi\circ\big[\one+\delta_f^\theta,\one+\delta_g^\theta\big]A
  =
  \chi\circ\delta_{t\,\{f,g\}_\theta}^\theta A =
  \hat\delta_{t\,\{f,g\}_\theta}^\theta \chi(A) \ . 
\end{align*}
In \cite{BBKT} it was shown that the Seiberg-Witten maps $\chi$ correspond
to $L_\infty$-quasi-isomorphisms which describe the arbitrariness in
the definition of the related 
$L_\infty$-algebras in the $L_\infty$-bootstrap approach~\cite{BBKL}.
\end{proof}

\begin{example}\label{ex:Linftyconstant}
Let $M=\real^d$ with a constant Poisson structure $\theta$. By
Example~\ref{ex:constemb}, in this
case we can take $\gamma^{(n)}=0$ for all $n\geq1$. Then the only
non-vanishing maps are
\begin{align*}
\ell_1^{\,\theta}(f) = \dd f \ , \quad \ell_2^{\,\theta}(f,g) =
  -\{f,g\}_\theta \qquad \mbox{and} \qquad \ell_2^{\,\theta}(f,A) =
  \{A,f\}_\theta \ .
\end{align*}
Thus in this case the symplectic embedding of $(\real^d,\theta)$ in
$(T^*\real^d,\omega_0+t\,\theta^*)$ makes
$C^\infty(\real^d)\oplus\Omega^1(\real^d)$ into a differential graded
Lie algebra, with Lie bracket given by the Poisson
bracket $\{\,\cdot\,,\,\cdot\,\}_\theta$. This describes the algebroid
of Poisson gauge transformations from Example~\ref{ex:Rd}. 
\qen\end{example}

\begin{remark}\label{rem:ell2theta}
The brackets $\ell_2^{\,\theta}(f,A)$ in
Proposition~\ref{prop:PoissonLinfty} encode the failure of the
exterior derivative $\dd$ in being a derivation of the Poisson algebra in
general. The homotopy Jacobi identity $\CJ^{\,\theta}_2(f,g)=0$ reads
\begin{align*}
\dd\{f,g\}_\theta = \{\dd f,g\}_\theta + \{f,\dd g\}_\theta
  +\gamma^{(1)}(\dd g,\dd f) - \gamma^{(1)}(\dd f,\dd g) \ ,
\end{align*}
which using \eqref{gtrec} can be written locally in components as
\begin{align*}
\partial_i\{f,g\}_\theta = \{\partial_if,g\}_\theta +
  \{f,\partial_ig\}_\theta
  + \partial_i\theta^{jk}\,\partial_jf\,\partial_kg \ .
\end{align*}
This is just the familiar violation of the Leibniz rule for non-constant
Poisson bivectors $\theta$.

Using $\ell^{\,\theta}_{n+2}(f,g,A^{\otimes n})=0$ for $n\geq1$, we can
demonstrate explicitly that the higher Jacobi identities are
equivalent to the local equations \eqref{eq2} for a symplectic
embedding of a Poisson manifold, as an illustration of the tight
relationship between our symplectic embeddings and the corresponding
$L_\infty$-structures. For this, we use Lemma~\ref{lem:homotopygauge}
and write \eqref{Jnfg} as
\begin{align}
\CJ^{\,\theta}_{n+2}\big(f,g,A^{\otimes n}\big) &=
                                       \ell^{\,\theta}_{n+1}\big(\ell^{\,\theta}_2(f,g),A^{\otimes
                                       n}\big) -
                                       (-1)^{n}\,\ell^{\,\theta}_{n+2}\big(\ell^{\,\theta}_1(f),g,A^{\otimes
                                       n}\big) -
                                       (-1)^{n}\,\ell^{\,\theta}_{n+2}\big(f,\ell^{\,\theta}_1(g),A^{\otimes
                                       n}\big) \notag \\
& \quad \, - n\,\Big(\ell^{\,\theta}_{n+1}\big(\ell^{\,\theta}_2(f,A),g,A^{\otimes n-1}\big)-\ell^{\,\theta}_{n+1}\big(\ell^{\,\theta}_2(g,A),f,A^{\otimes n-1}\big) \notag\\
& \quad \, \hspace{1cm} + \ell^{\,\theta}_{2}\big(\ell^{\,\theta}_{n+1}(f,A^{\otimes n}),g\big)-\ell^{\,\theta}_{2}\big(\ell^{\,\theta}_{n+1}(g,A^{\otimes n}),f\big)\Big) \label{Jnfg1}\\
& \quad \, - \sum_{i=3}^{n}\,(-1)^{(i+1)\,(n-i)} \, \binom{n}{i-1} \,
                                                                                                                                                                                      \Big(\ell^{\,\theta}_{n-i+3}\big(\ell^{\,\theta}_i(f,A^{\otimes
                                                                                                                                                                                      i-1}),g,A^{\otimes
                                                                                                                                                                                      n-i+1}\big)
                                                                                                                                                                                      \notag\\
& \quad \, \hspace{6cm}
                                                                                                                                                                                                 -\ell^{\,\theta}_{n-i+3}\big(\ell^{\,\theta}_i(g,A^{\otimes
                                                                                                                                                                                                 i-1}),f,A^{\otimes
                                                                                                                                                                                                 n-i+1}\big)\Big)
                                                                                                                                                                                                 \
                                                                                                                                                                                                 . \notag
\end{align}
We substitute the brackets from Proposition~\ref{prop:PoissonLinfty}
in (\ref{Jnfg1}), and after some careful simplification in local
coordinates its components become
\begin{align*}
\CJ^{\,\theta}_{n+2}\big(f,g,A^{\otimes n}\big)_i &=
                                         n!\,(-1)^{\frac{(n+1)\,(n+2)}{2}}
  \\
  & \quad \, \times \, \Big((n+1)\,\big(\gamma_i^{k|li_2\cdots
                                         i_{n+1}}
                                         -\gamma_i^{l|ki_2\cdots
                                         i_{n+1}}\big) \\
& \quad \, \hspace{1cm} +\sum_{m=1}^n\,(n-m+1)\,\big(\gamma_j^{l|i_1\cdots i_{m}}\,\gamma_i^{k|ji_{m+1}\cdots i_{n+1}} -
\gamma_j^{k|i_1\cdots i_{m}}\,\gamma_i^{l|ji_{m+1}\cdots i_{n+1}} \big) \\
& \quad \, \hspace{2cm} -\theta^{kj}\,\partial_j\gamma_i^{l|i_2\cdots i_{n+1}} + \theta^{lj}\,\partial_j\gamma_i^{k|i_2\cdots i_{n+1}} -
                                                                                   \gamma_i^{j|i_2\cdots i_{n+1}}\,\partial_j\theta^{lk}\Big) \\
  & \quad \, \hspace{8cm} \times \, 
                                                                                   \partial_kf\,\partial_lg\,A_{i_2}\cdots
                                                                                   A_{i_{n+1}}
                                                                                   \
                                                                                   .
\end{align*}
This vanishes as a consequence of (\ref{h5}), and thus the closure
condition for the gauge algebra implies that the brackets
$\ell_n^{\,\theta}$ indeed define an $L_\infty$-algebra. In other words, the homotopy Jacobi identities are equivalent to the Poisson integrability condition $[\omega^{-1},\omega^{-1}]=0$, as written in \eqref{eq:Omegamint} and \eqref{eq:Thetanalmost} (with $\underline{\theta}^{ij}(x,p)=\theta^{ij}(x)$); algebraically, the brackets follow from a standard higher derived bracket construction~\cite{Voronov2003}, as we shall see in Section~\ref{sec:Pinfty}.

This discussion illustrates the necessity of the
$L_\infty$-algebra formulation even in the case of Poisson gauge
transformations for non-constant bivectors $\theta$, despite the fact that they
are still defined by an underlying Lie algebra action.
\qen
\end{remark}

\begin{remark} \label{rem:uniqueness}
We can now give a homotopy algebraic answer to the question of
uniqueness of symplectic embeddings for a given Poisson manifold
$(M,\theta)$. Similarly to the uniqueness statement of
Proposition~\ref{prop:PoissonLinfty}, different local symplectic
embeddings $(T^*M,\omega)$ correspond to field redefinitions of the
brackets of the corresponding $L_\infty$-algebras, which are related
to one another through
$L_\infty$-quasi-isomorphisms~\cite{Kupriyanov:2019ezf,BBKT}. 
  \qen
\end{remark}

\subsection{$L_\infty$-structures on almost Poisson gauge algebroids}

The general case of almost Poisson gauge transformations from
Section~\ref{sec:quasiP} is markedly different from the Poisson case,
as now there is no underlying Lie algebra structure on $V$ and hence
no corresponding action algebroid. However, suppressing momentarily
the formal power series extension again as well as the jet coordinate
dependence, there is still an underlying Lie algebroid with vector bundle
\begin{align*}
C^\infty(U) \longrightarrow \CCE(U) \longrightarrow \Omega^1(U)
\end{align*}
whose fibre over a gauge field $A\in \Omega^1(U)$ is
$s_A^* C^\infty(T^*U) \simeq C^\infty(U)$. A field dependent gauge parameter 
can then be identified with a section of this bundle. The bracket and anchor map 
of this gauge algebroid are induced by the bracket
$[\![\,\cdot\,,\,\cdot\,]\!]_\theta$ and the gauge variation
$(f,A)\mapsto\delta^\theta_fA$ of Definition~\ref{def:quasiPoissongauge}. 

The corresponding $L_\infty$-structure making this Lie algebroid
picture precise is induced by the symplectic embedding construction of
Section~\ref{sec:quasiP} and may be presented in the following
geometric way. For this, we introduce tensor fields whose local components
are the coefficients of the formal power series expansions of
Sections~\ref{sec:quasiPoisson} and~\ref{sec:Lagrangian}, in order to
write the gauge variations and brackets of
Remark~\ref{rem:quasiPoissongauge} as formal power series analogously
to \eqref{eq:sAgammaformal}. Assuming without loss of generality that
the functions $\theta^{ij|i_1\cdots i_n}\in C^\infty(U)$ for $n\geq 2$
are symmetric
in their last $n$ indices, we introduce $n{+}2$-tensors
$\theta^{(n)}\in\mfX^2\big(U,\mfX(U)^{\odot n}\big)$ in local
coordinates by
\begin{align*}
\theta^{(n)} := n! \, \theta^{ij|i_1\cdots i_n}(x) \,
  (\partial_i\wedge\partial_j)\otimes(\partial_{i_1} \odot\cdots\odot
  \partial_{i_n}) \ .
\end{align*}
We set $\theta^{(0)}:=\theta$ and $\theta^{(1)}:=-\Pim$. Then
\begin{align}\label{eq:sAthetaformal}
s_{A}^*\underline{\theta}(\pi^*\,\cdot\,,\pi^*\,\cdot\,) = \sum_{n=0}^\infty \, \frac{t^n}{n!} \
  \theta^{(n)}\big(\,\cdot\,,\,\cdot\,,A^{\otimes n}\big)
\end{align}
as a formal power series in $\mfX^2(U)[[t]]$.

Similarly, taking $\Upm^{ijk|i_1\cdots i_{n-1}}\in C^\infty(U)$ for
$n\geq 2$ to be
symmetric in their last $n$ indices, we introduce $n{+}2$-tensors
$\Upm^{(n)}\in\mfX\big(U, \mfX(U)\otimes\mfX(U)^{\odot n}\big)$ in local
coordinates by
\begin{align*}
\Upm^{(n)} := n! \, \Upm^{ijk|i_1\cdots i_{n-1}}(x) \,
  \partial_i\otimes\partial_j \otimes
  (\partial_k\odot\partial_{i_1}\odot\cdots\odot \partial_{i_{n-1}}) \ .
\end{align*}
We set $\Upm^{(1)}:=\Pim$. By the construction of Section~\ref{sec:Lagrangian}, the vector fields
$\Upm^{(n)}(\dd f,A^{\otimes n})\in\mfX(U)$ completely determine the formal power series
expansions of the Lagrangian multipliers $L_f\in
\mfX_{\rm pol}^{\tt h}(J^1T^*U)[[t]]$ through
\begin{align}\label{eq:Lfformal}
L_f = -\sum_{n=1}^\infty\, \frac{t^{n+1}}{(n+1)!} \, L^{(n)}\big(\dd
  f,A^{\otimes n}\big) \ ,
\end{align}
where
\begin{align*}
& L^{(n)}\big(\dd
  f,A^{\otimes n}\big) := \Upm^{(n)}\big(\dd f,A^{\otimes
  n}\big) + \sum_{k=1}^{\lfloor \frac{n-1}2 \rfloor} \
  \sum_{\stackrel{\scriptstyle l_1,\dots,l_{k+1}\geq 1}{\scriptstyle
  l_1+\cdots+ l_{k+1}=n-k}} \, \frac{(n+1)!}{(l_1+1)! \cdots (l_{k+1}+1)!}
  \\ & \hspace{2cm}\ \times
  \Upm^{(l_1)}\Big(\Upm^{(l_2)\otimes}\big(\DD A,\cdots
  \Upm^{(l_k)\otimes}(\DD
  A,\Upm^{(l_{k+1})\otimes}(\DD A,\dd f,A^{\otimes
  l_{k+1}}),A^{\otimes l_k}), \cdots A^{\otimes l_2}\big), A^{\otimes
  l_1}\Big) 
\end{align*}
with the sum omitted for $n=1,2$.

\begin{proposition}\label{prop:quasiPoissonLinfty}
Let $(T^*M,\omega)$ be a local symplectic embedding of an almost Poisson
manifold $(M,\theta)$ with Jacobiator $\Pim$, and let $U\subseteq M$ be an
open subset. Then there is an $L_\infty$-structure
$(\ell_n^{\,\theta})_{n=1}^\infty$ on 
$C^\infty(U)\oplus\Omega^1(U)$, unique up to
$L_\infty$-quasi-isomorphism, which induces the almost Poisson gauge
algebroid constructed by Proposition~\ref{p1} with non-vanishing
maps on coincident degree~$1$ entries given as follows. The brackets
involving a single gauge parameter are given by
\begin{align*}
  \ell^{\,\theta}_1(f) & = \dd f \ , \\[4pt]
  \ell^{\,\theta}_2(f, A) & =  \{A,f\}_\theta +\gamma^{(1)}(\dd f,A) \ , \\[4pt]
  \ell_3^{\,\theta}\big(f, A^{\otimes 2}\big) & = -\gamma^{(2)}\big(\dd
                                          f,A^{\otimes 2}\big) +
                                                2\,\Pim^\otimes(\dd
                                                f,A,\DD A) +\dd
                                                A\big(\Pim(\dd
                                                f,A,\,\cdot\,),\,\cdot\,\big)
                                                \ , \\[4pt]
  \ell_{n+1}^{\,\theta}\big(f,A^{\otimes n}\big)
                       &= (-1)^{\frac{n\,(n-1)}2} \,
                         \bigg[\gamma^{(n)}\big(\dd f,A^{\otimes
                         n}\big) - n\,
                         \theta^{(n-1)\otimes}\big(\dd
                         f,\DD A,A^{\otimes n-1}\big) \\
  & \hspace{4cm} - \dd
                         A\big(L^{(n-1)}(\dd f,A^{\otimes n-1}),\,\cdot\,\big) \\
  & \quad \, + \sum_{k=2}^{n-1} \, \binom nk \bigg( \gamma^{(n-k)}\Big(\DD
    A\big(L^{(k-1) \otimes}(\dd f,A^{\otimes
    k-1})\big),A^{\otimes n-k}\Big) \\
  & \quad \hspace{2cm}
    -\gamma^{(n-k)\otimes}\big(\DD
    A,A^{\otimes n-k},L^{(k-1)}(\dd f,A^{\otimes k-1})\big) \\
  & \quad \hspace{2cm}
    -(n-k)\,\theta^{(n-k-1)\otimes}\Big(\DD
    A\big(L^{(k-1) \otimes}(\dd f,A^{\otimes
    k-1})\big),\DD A,A^{\otimes
    n-k-1}\Big)\bigg)\bigg] \ ,
\end{align*}
for $n\geq3$. The brackets involving two gauge parameters are given by
\begin{align*}
  \ell_2^{\,\theta}(f,g) &= -\{f,g\}_\theta \ , \\[4pt]
  \ell_3^{\,\theta}(f,g,A) &= \Pim(\dd f,\dd g,A) \ , \\[4pt]
  \ell_4^{\,\theta}\big(f,g,A^{\otimes 2}\big) &= \theta^{(2)}\big(\dd
                                                 f,\dd g,A^{\otimes
                                                 2}\big) \ , \\[4pt]
  \ell_5^{\,\theta}\big(f,g,A^{\otimes 3}\big) &= \theta^{(3)}\big(\dd
                                                 f,\dd g,A^{\otimes
                                                 3}\big)-\tfrac32\, \dd
                                                 A\big(\Pim(\dd
                                                 f,A,\,\cdot\,),\Pim(\dd
                                                 g,A,\,\cdot\,)\big) \
                                                 , \\[4pt]
  \ell_{n+2}^{\,\theta}\big(f,g,A^{\otimes n}\big) &=
                                                     -(-1)^{\frac{n\,(n-1)}2}
  \\
  & \hspace{-2.7cm} \times
                         \bigg[\theta^{(n)}\big(\dd f,\dd g,A^{\otimes
                                                     n}\big) -
                                                     \sum_{k=1}^{n-2}\,\frac1{k+1}\,\binom
                                                     nk \, \dd
                                                     A\big(L^{(k)}(\dd
                                                     f,A^{\otimes
                                                     k}),L^{(n-k-1)}(\dd
                                                     g,A^{\otimes
                                                     n-k-1})\big) \\
  & \hspace{-2.2cm} + \sum_{k=3}^{n-1}\,\frac1{k+1}\,\binom nk \
    \sum_{l=1}^{k-2}\, \binom{k+1}{l+1} \,
    \bigg(\gamma^{(n-k)}\Big(\DD A\big(L^{(l)\otimes}(\dd f,A^{\otimes
    l})\big),A^{\otimes n-k},L^{(k-l-1)}(\dd g,A^{\otimes k-l-1})\Big)
  \\
  & \quad \hspace{2.7cm} - \gamma^{(n-k)}\Big(\DD A\big(L^{(l)\otimes}(\dd g,A^{\otimes
    l})\big),A^{\otimes n-k},L^{(k-l-1)}(\dd f,A^{\otimes k-l-1})\Big)
  \\
  & \hspace{-1cm} -(n-k)\,\theta^{(n-k-1)}\Big(\DD
    A\big(L^{(l)\otimes}(\dd f,A^{\otimes l})\big),\DD
    A\big(L^{(k-l-1)\otimes}(\dd g,A^{\otimes k-l-1})\big),A^{\otimes
    n-k-1}\Big) \bigg) \bigg] \ ,
\end{align*}
for $n\geq 4$. In these expressions skew-symmetry between gauge
parameters and gauge fields is imposed by definition, for $f,g\in
C^\infty(U)$ and $A\in\Omega^1(U)$.
\end{proposition}
\begin{proof}
The proof is completely analogous to the proof of
Proposition~\ref{prop:PoissonLinfty}. The brackets $\ell_n^{\,\theta}$
now follow from a straightforward but laborious comparison of the gauge variation and bracket of
Remark~\ref{rem:quasiPoissongauge} with \eqref{eq:Linftygt} and
\eqref{eq:Linftyclosure}, respectively, order by order in $t$ using
\eqref{eq:sAgammaformal}, \eqref{eq:sAthetaformal} and~\eqref{eq:Lfformal}. The fact that Proposition~\ref{p1}  constructs an
almost Poisson gauge algebra, in the sense of
Definition~\ref{def:quasiPoissongauge}, then implies that the statement is a special instance
of~\cite[Theorem~2]{Fulp:2002kk} for the general case of field
dependent gauge transformations.
\end{proof}

\begin{remark}
  When $\theta$ is a Poisson structure, then the tensors $\Pim$,
  $\theta^{(n)}$ for $n\geq2$ and $L^{(k)}$ all vanish, and the
  $L_\infty$-structure of Proposition~\ref{prop:quasiPoissonLinfty}
  reduces to the $L_\infty$-structure of
  Proposition~\ref{prop:PoissonLinfty}. The brackets $\ell_3^{\,\theta}$ in
Proposition~\ref{prop:quasiPoissonLinfty} coincide exactly with the
brackets defined in~\cite{BBKL} (after taking into account that the
Jacobiator $\Pim$ in~\cite{BBKL} differs from ours by a factor $\frac13$). Moreover, the brackets
$\ell_{n+2}^{\,\theta}(f,g,A^{\otimes n})$ for $n=0,1,2$ coincide with
the brackets obtained in~\cite{Kupriyanov:2019ezf}. In particular, the
homotopy Jacobi identity $\CJ^{\,\theta}_3(f,g,h)=0$ for $f,g,h\in C^\infty(U)$
reads
\begin{align*}
{\sf Cyc}_{f,g,h}\,\{f,\{g,h\}_\theta\}_\theta = \tfrac12\,\Pim(\dd f,\dd
  g,\dd h) \ ,
\end{align*}
which is the familiar violation of the strict Jacobi identity for an
almost Poisson bracket with non-vanishing Jacobiator $\Pim$. However, for $n>2$
the brackets $\theta^{(n)}(\dd f,\dd g,A^{\otimes n})$ are corrected by
terms involving contributions from the Lagrangian multiplier vector
fields and do not on their own constitute the correct
$L_\infty$-structure required by the gauge closure condition
\eqref{eq:qPoissonalg} at higher orders, contrary to the conjecture
of~\cite{Kupriyanov:2018yaj}. We shall consider an explicit example
in Section~\ref{sec:Examples} below.
\qen
\end{remark}

\subsection{$P_\infty$-algebras of exterior differential forms}
\label{sec:Pinfty}

Let us now discuss some potential applications of our constructions to deformation quantization. A central problem in the formulation of noncommutative gauge theories is to find an extension of a star-product, which quantizes an almost Poisson manifold $(M,\theta)$, to the de~Rham complex $(\Omega^\bullet(M),\dd)$. Even at the purely kinematical level this problem is non-trivial, as the vector space action mentioned in Remark~\ref{rem:Abracket} does not generally make $\Omega^1(U)[[\hbar]]$ into a $C^\infty(U)[[\hbar]]$-bimodule. The problem can be traced back to the semi-classical limit: the differential graded algebra $\Omega^\bullet(M)$ is not generally a Poisson algebra, even when $\theta$ is a Poisson bivector field. In the case of symplectic manifolds, the problem of endowing $\Omega^\bullet(M)$ with the structure of a differential graded Poisson algebra is discussed in e.g.~\cite{Chu:1997ik,Hawkins:2002rf,Beggs:2003ne,McCurdy:2009xz}; the general construction depends on the choice of an almost symplectic connection or of a contravariant connection. Here we shall present an alternative treatment based on the constructions of this paper which is more general: it does not require the auxiliary data of a connection and works for any almost Poisson bivector field~$\theta$.

Our main observation is that a symplectic embedding,
which always exists locally, naturally induces a $P_\infty$-structure on the de~Rham complex of the underlying almost Poisson manifold. A \emph{$P_\infty$-algebra} is a graded commutative algebra $\CCA$
together with an $L_\infty$-structure $\{\ell_n\}_{n=1}^\infty$ such
that the differential $\ell_1:\CCA\to\CCA$ is a derivation of
degree~$1$ and, for each $n\geq2$ and fixed elements $a_1,\dots,a_{n-1}\in\CCA$, the
map $a\mapsto \ell_n(a_1,\dots,a_{n-1},a)$ is a derivation of
degree $2-n+\sum_{i=1}^{n-1}\,|a_i|$; this is a strong homotopy version of a
Poisson algebra. $P_\infty$-algebras were introduced by Cattaneo and
Felder in their approach to quantization of coisotropic submanifolds
of Poisson manifolds~\cite{Cattaneo:2005zz}. Just as a Poisson algebra 
can be regarded as the semi-classical limit of an associative algebra
in deformation quantization, $P_\infty$-algebras arise as
semi-classical limits of $A_\infty$-algebras.

Since a symplectic embedding $(T^*M,\omega)$ of a generic almost Poisson structure
$\theta$ similarly involves a Lagrangian submanifold of a symplectic
manifold, we can offer a homotopy algebraic explanation for the meaning of the
symplectic structure $\omega$ on $T^*M$ away from the zero section,
as well as a new perspective on the relation between our
symplectic embeddings and the semi-classical limit of a deformation
quantization of $\theta$. Let $C$ be any manifold. The general result
of~\cite[Proposition~2.1]{Cattaneo:2005zz} then constructs a
$P_\infty$-structure on the graded commutative algebra
$\Gamma\big(C,\midwedge^\bullet E\big)$ of sections of the exterior
algebra of any vector bundle $E\to C$ whose total space is a Poisson
manifold.

Applying this result to $E=T^*M$ with a local symplectic
embedding $\omega$ of the almost Poisson structure $\theta$,
we find that the almost Poisson bracket on $C^\infty(M)$ can be
viewed as part of a $P_\infty$-structure on the de~Rham complex of the
manifold $M$, induced by the Poisson structure $\omega^{-1}$ on
$T^*M$. 
\begin{proposition}\label{prop:PinftydeRham}
Let $(T^*M,\omega)$ be a local symplectic embedding of an almost Poisson
manifold $(M,\theta)$ with Jacobiator $\Pim$, and let $U\subseteq M$ be an
open subset. Then there is a $P_\infty$-structure
$\{\varrho_n^{\,\theta}\}_{n=1}^\infty$ on the exterior
algebra $\Omega^\bullet(U)$ of differential forms on $U$ defined as follows. On generators of $\Omega^\bullet(U)$, the
non-vanishing brackets are defined by
\begin{align}
  \varrho_1^{\,\theta}(f) &= \dd f \ ,  \qquad \varrho_1^{\,\theta}(\alpha) = \dd\alpha \ , \nonumber \\[4pt]
  \varrho_2^{\,\theta}(f,g) &= \{f,g\}_\theta \ , \qquad \varrho_2^{\,\theta}(\alpha,\beta) = \{\alpha,\beta\}_\theta -\gamma^{(1)\otimes}(\DD\alpha,\beta) - \gamma^{(1)\otimes}(\DD\beta,\alpha) \ , \nonumber \\[4pt]
  \varrho_2^{\,\theta}(f,\alpha) &= \{f,\alpha\}_\theta - \gamma^{(1)}(\dd f,\alpha) \ , \qquad \varrho_3^{\,\theta}(f,g,\alpha) = \Pim(\dd f,\dd g,\alpha) \ , \nonumber \\[4pt]
\varrho_n^{\,\theta}(\alpha_1,\dots,\alpha_n) &= \frac{(-1)^n}{(n-1)!} \, \Big( \sum_{i=1}^n\, \gamma^{(n-1)\otimes}\big(\DD\alpha_i,\alpha_1,\dots,\widehat{\alpha_i},\dots,\alpha_n\big) \nonumber \\
& \quad \, \hspace{1cm} - 2\,(n-1) \, \sum_{i<j} \, \theta^{(n-2)\otimes}\big(\DD\alpha_i,\DD\alpha_j,\alpha_1,\dots,\widehat{\alpha_i},\dots,\widehat{\alpha_j},\dots,\alpha_n\big)\Big) \ , \nonumber \\[4pt]
  \varrho_{n+1}^{\,\theta}(f,\alpha_1,\dots,\alpha_n) &= \frac{1}{n!} \,
                                             \gamma^{(n)}(\dd
                                             f,\alpha_1,\dots,\alpha_n) \notag \\
                                             & \quad \, + \frac1{(n-1)!} \, \sum_{i=1}^n \, \theta^{(n-1)\otimes}\big(\dd f,\DD\alpha_i,\alpha_1,\dots,\widehat{\alpha_i},\dots,\alpha_n\big)
                                             \ , \notag \\[4pt]
  \varrho_{n+2}^{\,\theta}(f,g,\alpha_1,\dots,\alpha_n) &= \frac{(-1)^n}{n!} \,
                                               \theta^{(n)}(\dd f,\dd
                                               g,\alpha_1,\dots,\alpha_n)
                                               \ , \label{eq:PinftydeRham}
\end{align}
for $n\geq2$, $f,g\in C^\infty(U)$, and
$\alpha,\beta,\alpha_1,\dots,\alpha_n\in\Omega^1(U)$, where a hat indicates omission of the corresponding entry. The brackets are then defined on higher degree forms by uniquely extending \eqref{eq:PinftydeRham}
to linear maps $\varrho_n^{\,\theta}:\Omega^\bullet(U)^{\otimes
  n}\to\Omega^\bullet(U)$ as polyderivations.
\end{proposition}
\begin{proof}
Since $\varrho_n^{\,\theta}$ is of degree $2-n$, it vanishes on elements of degree~$0$ or $1$ except in the cases given in \eqref{eq:PinftydeRham}. Its structural form is merely a translation of the statement
of~\cite[Theorem~2.2]{Cattaneo:2005zz} to this situation. That statement tells us that the components of the $P_\infty$-structure on $\Omega^\bullet(U)$ are the Taylor series expansion coefficents, in the transverse coordinates to the zero section $U\subset T^*U$, of the cosymplectic bivector field \eqref{PBq}. In local coordinates where $\alpha=\alpha_i\,\dd x^i\in\Omega^1(U)$, and vector fields $X=X^i\,\partial_i \in \mfX(U)$ are regarded as fibre-linear functions $X^i\,p_i$ on $T^*U$, these are given by
\begin{align*}
\varrho_{n+2}^{\,\theta}(f,g,\alpha_1,\dots,\alpha_n) &= (-1)^n \, \alpha_{1\,i_1}\cdots\alpha_{n\,i_n}\,\tilde\partial^{i_1}\cdots\tilde\partial^{i_n}\{\pi^*f,\pi^*g\}_{\omega^{-1}}\big|_{p=0} \ , \\[4pt]
\varrho_{n+1}^{\,\theta}(f,\alpha_1,\dots,\alpha_n)(X) &= \alpha_{1\,i_1}\cdots\alpha_{n\,i_n}\,\tilde\partial^{i_1}\cdots\tilde\partial^{i_n}\{\pi^*f,X^i\,p_i\}_{\omega^{-1}}\big|_{p=0} \ , \\[4pt]
\varrho_n^{\,\theta}(\alpha_1,\dots,\alpha_n)(X,Y) &= (-1)^n\,\alpha_{1\,i_1}\cdots\alpha_{n\,i_n}\,\tilde\partial^{i_1}\cdots\tilde\partial^{i_n}\{X^i\,p_i,Y^j\,p_j\}_{\omega^{-1}}\big|_{p=0} \ .
\end{align*}
Substituting the  series expansions \eqref{eq:gammaijxp} and \eqref{l12} at $t=1$, and using the Koszul formula for the de~Rham differential, after some calculation we arrive at the formulas \eqref{eq:PinftydeRham}. One can check directly that these brackets extend to a $P_\infty$-structure: the homotopy Jacobi identities are equivalent to the Poisson integrability condition $[\omega^{-1},\omega^{-1}]=0$, as written in \eqref{eq:Omegamint} and \eqref{eq:Thetanalmost}, and indeed the brackets follow algebraically from a standard higher derived bracket construction (see~\cite[Section~2.6]{Cattaneo:2005zz}).
\end{proof}

\begin{example}
Let $M=\real^d$ with a constant Poisson structure $\theta$. In this case all tensors $\Pim$, $\gamma^{(n)}$ and $\theta^{(n)}$ for $n\geq1$ vanish in Proposition~\ref{prop:PinftydeRham}, and the only non-zero brackets are given by
\begin{align*}
\varrho^{\,\theta}_1(\xi) = \dd \xi \qquad \mbox{and} \qquad \varrho_2^{\,\theta}(\xi,\zeta) = \{\xi,\zeta\}_\theta \ ,
\end{align*}
for all $\xi,\zeta\in\Omega^\bullet(\real^d)$.
In this case we recover the well-known realization of $\Omega^\bullet(\real^d)$ as a differential graded Poisson algebra. In general, however, the $P_\infty$-algebra of Proposition~\ref{prop:PinftydeRham} involves infinitely-many non-zero brackets on the exterior algebra of differential forms, even for Poisson bivectors.
\qen
\end{example}

\begin{remark}\label{rem:Linftyalgebroid}
An \emph{$L_\infty$-algebroid} is a Lie algebroid together with an
$L_\infty$-structure on its Chevalley-Eilenberg algebra whose
differential is the Lie algebroid differential; it is a
\emph{$P_\infty$-algebroid} if the brackets of the
$L_\infty$-structure are polyderivations. Hence an equivalent way of stating 
Proposition~\ref{prop:PinftydeRham} is that a symplectic embedding
makes the tangent Lie algebroid over an almost Poisson manifold into a
$P_\infty$-algebroid. In this language,
Proposition~\ref{prop:PinftydeRham} is essentially a special case
of~\cite[Theorem~9.4]{Oh:2003ay} when $(M,\theta)$ is a Poisson manifold. 
\qen\end{remark}

\begin{remark}\label{rem:Ainfty}
Proposition~\ref{prop:PinftydeRham} implies a new homotopy algebraic
construction of a deformation quantization of the exterior algebra $\Omega^\bullet(U)$ in the direction of a
generic almost Poisson bracket. Since $M$ is a Lagrangian submanifold
of the local symplectic embedding $(T^*M,\omega)$, and since for our local symplectic
embeddings the cohomological obstructions
of~\cite[Corollary~3.3]{Cattaneo:2005zz} are trivial, we may apply the
result of~\cite[Theorem~3.2]{Cattaneo:2005zz} to quantize the
$P_\infty$-algebra of Proposition~\ref{prop:PinftydeRham} to an
$A_\infty$-structure on $\Omega^\bullet(U)[[\hbar]]$ over
$\real[[\hbar]]$, which is a deformation of the exterior product on
$\Omega^\bullet(U)$ and whose semi-classical limit induces the
$P_\infty$-structure $\{\varrho_n^{\,\theta}\}_{n=1}^\infty$ in the
following sense. The alternatization of the structure maps of this
$A_\infty$-algebra, which are polydifferential operators, define an $L_\infty$-structure
$\{\varrho_n^{\,\star}\}_{n=1}^\infty$ on $\Omega^\bullet(U)[[\hbar]]$. Then
$\varrho_n^{\,\theta} = \frac1\hbar\,\varrho_n^{\,\star}\big|_{\hbar=0}$,
and in this sense we may regard the structure maps
$\{\varrho_n^{\,\theta}\}_{n=1}^\infty$ as a semi-classical limit of
$\{\varrho_n^{\,\star}\}_{n=1}^\infty$. These
considerations are based on a version of Kontsevich's formality
theorem for the case of Lagrangian submanifolds of a symplectic
manifold. The role of homotopy algebras in
nonassociative deformation quantization of twisted Poisson structures was anticipated
by~\cite{Mylonas2012}, and the present observation makes this precise
for arbitrary almost Poisson manifolds. The details are beyond the
scope of this paper and will be explored elsewhere.
\qen\end{remark}

\paragraph{Poisson gauge algebras.}

In the case of a Poisson structure $\theta$, the structure maps of Proposition~\ref{prop:PinftydeRham}, truncated to degrees~$0$ and~$1$ with $\theta^{(n)}=0$ for all $n\geq1$ and $\alpha=\beta=\alpha_1=\dots=\alpha_n=A$ for all $n\geq1$, essentially agree with those of
Proposition~\ref{prop:PoissonLinfty} up to numerical factors. This suggests that the
$L_\infty$-structure of Proposition~\ref{prop:PoissonLinfty} is
compatible with a graded commutative algebra structure on
$V=C^\infty(U)\oplus\Omega^1(U)$, such that the Poisson bracket on
$C^\infty(U)$ is part of a $P_\infty$-structure on the $2$-term
cochain complex $(V,\dd)$. This expectation turns out to be correct
and is the content of 
\begin{proposition}\label{prop:PinftyPoisson}
Let $(M,\theta)$ be a Poisson manifold and $U\subseteq M$ an open
subset. Let $\CCA=C^\infty(U)\oplus\Omega^1(U)$ be the graded
commutative algebra with multiplication defined by truncating the
product on the exterior algebra $\Omega^\bullet(U)$ at degree~$1$,
that is, 
\begin{align*}
f\cdot g = f\,g \ , \quad f\cdot A = f\, A \qquad \mbox{and} \qquad
  A\cdot B=0
\end{align*}
for all $f,g\in C^\infty(U)$ and $A,B\in\Omega^1(U)$. Then the
$L_\infty$-structure of Proposition~\ref{prop:PoissonLinfty} turns
$\CCA$ into a $P_\infty$-algebra.
\end{proposition}
\begin{proof}
This follows from the fact that the subalgebra of Proposition~\ref{prop:PinftydeRham} obtained by truncation to degrees~$0$ and~$1$, and subsequent replacement of $\theta$ with $-\theta$, is precisely the stated $L_\infty$-structure on $\CCA$. It is also an easy direct check using the Leibniz rules for the exterior
derivative $\dd$, the Poisson bracket on $C^\infty(U)$, and the
$\Omega^1(U)$-valued bracket of Remark~\ref{rem:Abracket}.
\end{proof}

The significance of this observation is that it offers a natural path
towards a homotopy algebraic construction of noncommutative gauge
transformations beyond the semi-classical level. Similarly to the
discussion of Remark~\ref{rem:Ainfty}, the 
$P_\infty$-algebra of Proposition~\ref{prop:PinftyPoisson} can be
quantized to an
$A_\infty$-structure on $\CCA[[\hbar]]$ over $\real[[\hbar]]$ which is
a deformation of the algebra structure on $\CCA$ and whose
semi-classical limit induces the $P_\infty$-structure
$\{\ell_n^{\,\theta}\}_{n=1}^\infty$ through the 
alternatization $\{\ell_n^{\,\star}\}_{n=1}^\infty$ of the structure
maps of this $A_\infty$-algebra: $\ell_n^{\,\theta} =
\frac1\hbar\,\ell_n^{\,\star}\big|_{\hbar=0}$. Then the
Poisson gauge algebra is a semi-classical limit of a noncommutative
gauge algebra which is organized by the ``quantum''
$L_\infty$-structure $\{\ell_n^{\,\star}\}_{n=1}^\infty$, making the discussion at the beginning of
Section~\ref{sec:Poissongauge} somewhat more precise. While this is an
interesting approach to the construction of noncommutative gauge
symmetries, it also lies beyond the scope of the present paper and we
leave it for future investigation.

\paragraph{Almost Poisson gauge algebras.}

The situation
is much more complicated in the case of a general almost Poisson
bivector $\theta$. Now, because of the presence of the non-vanishing brackets $\rho_{n+2}^{\,\theta}$ in Proposition~\ref{prop:PinftydeRham} involving two functions and forms of higher degree, the corresponding truncation does not determine a subalgebra, as this would violate the homotopy Jacobi identities. 
The difference between the $L_\infty$-structure maps of Propositions~\ref{prop:quasiPoissonLinfty} and~\ref{prop:PinftydeRham} (aside from numerical factors) lies entirely in the terms involving the Lagrangian multipliers of the almost Poisson gauge algebroid, whose inclusion restores the homotopy Jacobi identities. As we will now show, these terms violate the derivation properties of the original $P_\infty$-algebra from Proposition~\ref{prop:PinftydeRham}.

For this,
let us study in detail the derivation properties of the brackets
$\ell^{\,\theta}_{n+2}(f,g,A^{\otimes n})$ and $\ell^{\,\theta}_{n+1}(f,A^{\otimes n})$ from
Proposition~\ref{prop:quasiPoissonLinfty} with the same graded
commutative algebra $\CCA$ as in
Proposition~\ref{prop:PinftyPoisson}. First of all, for the
multiplication of gauge parameters it is clear that the derivation
properties
\begin{align*}
\ell^{\,\theta}_{n+2}(f \cdot h,g,A^{\otimes n})&=f\cdot
                                       \ell^{\,\theta}_{n+2}(h,g,A^{\otimes
                                                  n}) + \ell^{\,\theta}_{n+2}(f,g,A^{\otimes
                                       n})\cdot h \
                                       , \\[4pt]
\ell^{\,\theta}_{n+1}(f \cdot h,A^{\otimes n})&=f\cdot
                                                \ell^{\,\theta}_{n+2}(h,A^{\otimes
                                                n}) + \ell^{\,\theta}_{n+2}(f,A^{\otimes n})\cdot h
\end{align*}
hold.

For the multiplication of gauge fields by gauge parameters, consider
first the brackets involving two gauge parameters. Since brackets
with three gauge parameters vanish for degree reasons, the desired derivation
property is just $C^\infty(U)$-linearity
\begin{equation}\label{ph2}
\ell^{\,\theta}_{n+2}(f ,g, h\cdot A,A^{\otimes n-1})=h\cdot \ell^{\,\theta}_{n+2}(f
,g,A^{\otimes n}) \ .
\end{equation}
This essentially means that the field dependent gauge parameter
$[\![f,g]\!]_\theta$ from Proposition~\ref{p1}, when evaluated on the
argument $h\cdot A$, should not depend on the derivatives of $h$. The
potentially problematic terms are the ones involving the Lagrangian
multiplier vector fields. However, the functions $\Lambdam^{ij}(A)$
constructed in Corollary~\ref{cor:LambdamA} satisfy
\begin{equation}
\Lambdam^{ij}(A)\,A_i=0 \ , \label{tr1}
\end{equation}
and together with Proposition~\ref{p2} it now easily follows that
\begin{equation*}
L^{ij}\big(A,\partial(h\,A)\big)=L^{ij}(A,h\,\partial A) \ .
\end{equation*}
Moreover, from (\ref{m5}) the identity (\ref{tr1}) also implies
\begin{equation*}
L^{ij}(A,\partial A)\,A_i=0 \ , 
\end{equation*}
and as a consequence the terms with derivatives of $h$ disappear
from $s_{h\,A}^*\{\Phi_{h\,A}(L_f),\Phi_{h\,A}(L_g)\}_{\omega^{-1}}$. It follows that
all terms involving derivatives of $h$ also disappear from
$[\![f,g]\!]_\theta (h\,A)$ and this implies (\ref{ph2}). In other
words, the
map $A\mapsto\ell^{\,\theta}_{n+2}(f ,g, \alpha_1,\dots,\alpha_{n-1},A)$
is a derivation, for all $n\geq1$ and $\alpha_1,\dots,\alpha_{n-1}\in\Omega^1(U)$.

For the brackets with a single gauge parameter the desired derivation
property reads
\begin{equation}\label{tr3}
\ell^{\,\theta}_{n+1}(f,h\cdot A, A^{\otimes
  n-1})=h\cdot\ell^{\,\theta}_{n+1}(f,A^{\otimes n}) +(-1)^{n-1}
\ell^{\,\theta}_{n+1}(f,h,A^{\otimes n-1})\cdot A \ .
\end{equation}
This relation is more complicated since it involves brackets of
different nature and different graded symmetry. In particular, since
the bracket with two gauge parameters is skew-symmetric in $f$ and
$h$, the relation (\ref{tr3}) implies the consistency condition
\begin{equation}\label{tr4}
\ell^{\,\theta}_{n+1}(f,h\cdot A, A^{\otimes
  n-1})+\ell^{\,\theta}_{n+1}(h,f\cdot A, A^{\otimes
  n-1})=h\cdot\ell^{\,\theta}_{n+1}(f,A^{\otimes
  n})+f\cdot\ell^{\,\theta}_{n+1}(h,A^{\otimes n}) \ .
\end{equation}
For $n=1$ it is easy to
see that the bracket $\ell_2^{\,\theta}$ from
Proposition~\ref{prop:quasiPoissonLinfty} satisfies \eqref{tr3}, while
for $n=2$ the relevant brackets are
\begin{align*}
\ell^{\,\theta}_3(f,h,A)&=\Pim(\dd f,\dd h,A) \ , \\[4pt]
\ell^{\,\theta}_3(f,A,B)&=-\gamma^{(2)}(\dd f,A,B) + \Pim^\otimes(\dd
                          f,A,\DD B) + \Pim^\otimes(\dd f,B,\DD A) \\
  & \quad \, +
                          \tfrac12\,\dd A\big(\Pim(\dd
                          f,B,\,\cdot\,),\,\cdot\,\big) +
                          \tfrac12\,\dd B\big(\Pim(\dd
                          f,A,\,\cdot\,),\,\cdot\,\big) \ ,
\end{align*}
for $f,h\in C^\infty(U)$ and $A,B\in\Omega^1(U)$. 
One then calculates explicitly
\begin{align*}
\ell^{\,\theta}_3(f,h\cdot
  A,A) &= -h\cdot\gamma^{(2)}(\dd f,A,A) + h\cdot\Pim^\otimes(\dd
         f,A,\DD A) + \Pim^\otimes(\dd f,A, \dd h\otimes A + h\,\DD A) \\
  & \quad \, +\tfrac12\,(\dd h\wedge A)\big(\Pim(\dd
    f,A,\,\cdot\,),\,\cdot\,\big) + \tfrac12\, h\cdot\dd A\big(\Pim(\dd
    f,A,\,\cdot\,),\,\cdot\,\big) \\
  & \quad \, + \tfrac12\, h\cdot\dd A\big(\Pim(\dd
    f,A,\,\cdot\,),\,\cdot\,\big) \\[4pt]
  &= h\cdot\big(-\gamma^{(2)}(\dd f,A,A) + 2\,\Pim^\otimes(\dd f,A,\DD
    A)+\dd A(\Pim(\dd f,A,\,\cdot\,),\,\cdot\,) \big) \\
  & \quad \, +\Pim(\dd f,A,\dd h)\,\cdot A +\tfrac12\,\Pim(\dd f, A,\dd
    h)\cdot A - 
    \tfrac12\,\dd h\cdot\Pim(\dd f,A,A) \\[4pt]
  &=h\cdot\ell^{\,\theta}_3(f,A,A)-\tfrac32\,\ell^{\,\theta}_3(f,h,A)\cdot
A \ ,
\end{align*}
showing that already for $n=2$ the derivation property (\ref{tr3}) is
violated. Starting from the next order $n=3$, even the consistency
condition (\ref{tr4}) is violated. In the case of a Poisson structure
the higher brackets all satisfy \eqref{tr4} because of their
$\CCA$-linearity, but for generic almost Poisson structures the
derivation property imposes severe restrictions on the
$L_\infty$-algebra. We summarise the present discussion in
\begin{proposition}\label{prop:PinftyquasiPoisson}
Let $(M,\theta)$ be an almost Poisson manifold with non-zero Jacobiator
$\Pim$, and let $U\subseteq M$ be an open
subset. Then the $L_\infty$-structure of
Proposition~\ref{prop:quasiPoissonLinfty} and the graded commutative
algebra structure of Proposition~\ref{prop:PinftyPoisson} do not
combine into a compatible $P_\infty$-structure on
$C^\infty(U)\oplus\Omega^1(U)$.
\end{proposition}

\begin{remark}\label{rem:openPinfty}
We do not solve the problem of finding a $P_\infty$-structure on the
almost Poisson gauge algebroid in this paper, but let us briefly mention
some possible ways that one may proceed.

As we have shown above, the
derivation properties are violated by precisely the terms involving
the Lagrangian multipliers $L_f$, which were introduced to cancel
the second term in the commutator of two gauge transformations in
\eqref{t7}. Instead of including $L_f$, one could try to cure the
problem by identifying this term as the contribution from a non-zero
homogeneous subspace $V_2$ in degree~$2$. This would lead to a
$3$-term $L_\infty$-algebra, which is likely a $P_\infty$-algebra
under a suitable extension of the truncated exterior product on
$C^\infty(U)\oplus \Omega^1(U)$. However, the gauge theory
interpretation of the space $V_2$ is not clear, and this appears to be
a general feature of physical systems based on almost Poisson algebras:
to maintain all desirable features one inevitably needs to introduce
some auxiliary (unphysical) degrees of freedom, such that the
elimination of these auxiliary variables comes at the price of losing
some of the desired properties (see~\cite{KS18} for an example of
this). In the present situation, we lose the Leibniz rule for the
$L_\infty$-structure, but still retain a well-defined gauge algebra.

Another possibility would be to work instead with a \emph{curved}
$L_\infty$-algebra. A \emph{curving} of an $L_\infty$-structure
$\{\ell_n\}_{n=1}^\infty$ on a graded vector space $V$ is an
additional map $\ell_0$ of degree~$2$ from the ground field into $V$,
which intertwines with the higher brackets through the homotopy Jacobi
identities extended to $\{\ell_n\}_{n=0}^\infty$; then $\ell_1$ is no
longer a differential and one loses the underlying cochain complex. In
our situation, it
may be possible to redefine the brackets $\ell_n^{\,\theta}$ (via an $L_\infty$-quasi-isomorphism) to absorb
the violation of the Leibniz rule, which may then violate the standard
$L_\infty$-relations, but may instead form a curved
$L_\infty$-algebra. However, a curved $L_\infty$-algebra also necessarily
contains a non-zero homogeneous subspace $V_2$ in degree~$2$, whose
meaning at the purely kinematic level of gauge symmetries is not
clear.

Finally, one could work with a weaker notion of $P_\infty$-structure,
where the graded commutative product is also replaced by a sequence of
higher products such that the commutativity of the product and the
Leibniz rule hold only up to homotopy. It would be interesting to
explore all of these modifications of the $L_\infty$-structure
exhibited in Proposition~\ref{prop:quasiPoissonLinfty} in order
to fully elucidate the structure of the almost Poisson gauge algebroid,
and its extension to noncommutative and nonassociative 
gauge transformations. 
\qen
\end{remark}

\begin{remark}\label{rem:AinftyBBKL}
Our constructions of $P_\infty$-structures above differ in several ways from the treatments of~\cite[Section~4.3]{BBKL} and~\cite[Section~4]{Kupriyanov:2019cug}, where an $L_\infty$-structure on the de~Rham complex, truncated at degree~$2$, was constructed for generic almost Poisson structures. 

Firstly, the $L_\infty$-structure of Proposition~\ref{prop:PinftydeRham} differs for brackets involving forms of degree higher than one: for example, the $2$-bracket of $E\in\Omega^2(U)$ and $f\in C^\infty(U)$ is simply the almost Poisson  bracket in~\cite{BBKL,Kupriyanov:2019cug}, while in our case the $2$-bracket is
\begin{align*}
\varrho^{\,\theta}_2(f,E) = \{f,E\}_\theta - 2\,\gamma^{(1)\otimes}(\dd f,E) \ .
\end{align*}
On the one hand the brackets of~\cite{BBKL,Kupriyanov:2019cug} do not define a $P_\infty$-structure, while on the other hand the brackets of Proposition~\ref{prop:PinftydeRham} are not designed to close an arbitrary gauge algebra. For a Poisson bivector $\theta$, the $P_\infty$-algebra of Proposition~\ref{prop:PinftydeRham} contains the Poisson gauge algebra of Proposition~\ref{prop:PoissonLinfty}, and following the standard $L_\infty$-algebra formulation of field theories~\cite{Hohm:2017pnh,Jurco:2018sby}, it also contains the higher degree spaces needed to formulate the dynamics of a particular Poisson gauge theory. For an almost Poisson structure, one can use the $P_\infty$-algebra in this way to formulate a gauge algebra without the need of Lagrangian multipliers, at the price of obtaining a closure condition that involves the field equations in addition to the field dependent gauge transformations~\cite{Hohm:2017pnh,Jurco:2018sby}.
For example, in $d=3$ dimensions the brackets appear to define an almost Poisson Chern-Simons theory which is different from that of~\cite{BBKL,Kupriyanov:2019cug}; it would be interesting to further develop this field theory, which in our symplectic embedding approach can be written down concisely and explicitly to all orders using the brackets \eqref{eq:PinftydeRham}.

Secondly, an $A_\infty$-structure on $\CCA[[\hbar]]$  is
sketched in~\cite[Appendix~B]{BBKL}, where the first few
multiplication maps are deduced up to order $\hbar^2$. However, this
structure is different from what is proposed above, as in their case
the classical limit reduces the $A_\infty$-algebra to the \emph{differential}
graded commutative algebra $(\CCA,\dd)$, whereas here we propose that
the classical limit should simply be the algebra $\CCA$ without
further structure. In other words, even the differential of the
$A_\infty$-structure should be considered as part of the deformation
quantization of the graded commutative algebra $\CCA$.
\qen
\end{remark}

\subsection{Comparison with the $L_\infty$-bootstrap}

According to the prescription of the $L_\infty$-bootstrap approach to
constructing almost Poisson gauge algebroids~\cite{BBKL}, one
starts with the natural structure maps
\begin{align*}
\ell_1(f)=\dd f \qquad \mbox{and} \qquad \ell_2(f,g)=-\{f,g\}_\theta \ ,
\end{align*}
and attempts to construct the rest of the $L_\infty$-structure by
consistently solving the homotopy Jacobi identities order by order.
Let us briefly comment on the benefits of our approach to deformations of
gauge transformations, based on symplectic
embeddings, over the approach based on the $L_\infty$-bootstrap, which
was used in~\cite{Kupriyanov:2019ezf} to propose recursion relations
for the construction of the gauge $L_\infty$-algebras:

\begin{itemize}
\item In the $L_\infty$-bootstrap approach, one can define
  contributions to the gauge transformations order by order in the
  formal deformation parameter $t$. Symplectic embeddings are more
  appropriate for computing explicit all orders expressions, which are
  sometimes asymptotic expansions of analytic functions known in
  closed form. We will
  illustrate this in several examples in Section~\ref{sec:Examples} below.
  
\item Following the method proposed in \cite{Kupriyanov:2019ezf}, one
  can recursively construct the brackets of the form $\ell_{n+1}(f,\ell_1(g),
  \ell_1(h),\dots)$ from previously defined brackets at lower
  orders. In the case of Poisson deformations of gauge theories, there is no problem
  in restoring $\ell_{n+1}(f,A^{\otimes n})$ from the given brackets
  $\ell_{n+1}(f,\ell_1(g), \ell_1(h),\dots)$. However, in the case of
  almost Poisson deformations the situation is much more complicated,
  as the passage from the brackets $\ell_{n+1}(f,\ell_1(g), \ell_1(h),\dots)$ to
  $\ell_{n+1}(f,A^{\otimes n})$ is extremely non-trivial. This problem
  is circumvented when working with symplectic embeddings.
  
\item From purely technical and calculational standpoints, the approach of this paper
  based on symplectic embeddings is significantly simpler then the
  bootstrap approach proposed in \cite{Kupriyanov:2019ezf}.

  \item The $L_\infty$-bootstrap approach describes a linearization of
    the complete $A_\infty$-structure of noncommutative gauge
    variations, which loses part of the information required to
    completely determine the semi-classical limit of the full
    noncommutative gauge transformations. The missing information is a
    derivation property, which requires a $P_\infty$-algebra to
    describe the semi-classical limit, and this is naturally captured
    by our symplectic embedding construction.
\end{itemize}

\newsection{Examples}
\label{sec:Examples}

In previous sections we already looked at the two simplest examples of
Poisson bivector fields: the case of a manifold $M$ with the trivial
Poisson structure $\theta_0=0$, for which the symplectic embedding is
given by the associated symplectic groupoid $(T^*M,\omega_0)$ for $M$,
and $M=\real^d$ with a constant Poisson structure $\theta(x)=\theta$,
whose symplectic embedding is given by the strict deformation of the
cotangent bundle $T^*\real^d$ with symplectic form
$\omega=\omega_0+t\,\theta^*$; the corresponding Poisson
gauge transformations were described in Example~\ref{ex:Rd} and the
associated $L_\infty$-structure on the gauge algebroid in
Example~\ref{ex:Linftyconstant}. The purpose of this final section is to extend these
basic examples to some more complicated examples, and in particular to
consider an example of an almost Poisson structure. These cases
will generally involve symplectic embeddings of $(M,\theta)$ given by a
formal deformation of the cotangent bundle of $M$, while the
(almost) Poisson gauge algebroid is correspondingly described by an
$L_\infty$-algebra which in general is no longer simply a differential graded Lie
algebra and involves infinitely many brackets.

\subsection{Linear Poisson structures}

We consider first a large class of Poisson structures whereby the symplectic
embedding can be described as a \emph{strict} deformation of
$(T^*M,\omega_0)$. Let $\mfg$ be a Lie algebra of dimension $d$ with structure constants in a given basis
denoted by $f^{ij}_k$. On $M=\real^d$ we can define the linear Poisson
bivector field~\cite{Alekseev}
\begin{equation}\label{e1}
\theta^{ij}(x)=f^{ij}_k\,x^k \ ,
\end{equation}
which, by regarding $\real^d$ as the dual of the Lie algebra $\mfg$, is the
Kirillov-Kostant Poisson structure on $\mfg^*$. In this case the
function $\Sigma(p,p',x)$ in \eqref{eq:Spxdef} is a
generating function for the Dynkin series for the Baker-Campbell-Hausdorff formula for the
Lie algebra $\mfg=(\real^d)^*$; if $G$ is any Lie group whose Lie
algebra is $\mfg$, then an integrating symplectic groupoid is
$T^*G\simeq G \ltimes\mfg^*$, regarded as the action groupoid with
respect to the coadjoint action, see e.g.~\cite{Cattaneo:2000iw}. Using the
polydifferential representation constructed
in~\cite{Gutt,Meljanac,Kupriyanov:2015uxa} one finds, in the notation
of Section~\ref{sec:Poissonembedding}, that a local
symplectic embedding of the Poisson structure (\ref{e1}) can be
written as
 \begin{equation}\label{e3}
   \gamma^i_j(p)= \sum_{n=1}^\infty \, \frac{t^{n-1}\,B_n}{n!} \,
   {f}^{ij_1}_{k_1}\,{f}^{k_1j_2}_{k_2}\cdots{f}^{k_{n-1}j_n}_{j}\,
   p_{j_1}\cdots p_{j_n} \ ,
\end{equation}
where $B_n$ are the Bernoulli numbers
($B_n=-\frac12,\frac16,0,-\frac1{30},0,\dots$ for $n=1,2,3,4,5,\dots$). 

Suppose that the structure constants of $\mfg$ are chosen such that
the $d{\times}d$
matrix $\sf M$, with elements ${\sf M}^{i}_l(p):={f}^{ij_1}_k\, {f}^{kj_2}_l\,p_{j_1}\,p_{j_2}$,
is diagonalizable for all transverse coordinates $p$. Following~\cite{Kupriyanov:2015uxa}, one may then
construct an analytic function $\gamma^i_j(p)$ whose Taylor expansion
around $t=0$ coincides with the asymptotic series \eqref{e3}. For
this, we observe that (\ref{e3}) can be rewritten as
 \begin{equation}\label{e4}
   \gamma^i_j(p)=-\tfrac12\,f^{il}_j\,p_{l}+\tfrac1t\,\mathcal{X}\big(-t^2\,{\sf
     M}/{2}\big)^i_j \ ,
\end{equation}
where $\mathcal{X}({\sf M})^i_j$ is the matrix-valued function with
\begin{equation*}
\mathcal{X}(u)=\sqrt{\tfrac{u}{2}}\cot\sqrt{\tfrac{u}{2}}-1=\sum_{n=1}^\infty\,
\frac{(-2)^n\,B_{2n}\,u^{n}}{(2n)!} \ ,
\end{equation*}
and the power series converges for $u\in\complex$ with $|u|<\frac12$.
Since ${\sf M}$ is diagonalizable, there exists a non-degenerate
$d{\times}d$ matrix $\sf S$ such that
\begin{equation*}
{\sf M}={\sf S}\, {\sf D}\, {\sf S}^{-1} \ ,
\end{equation*}
where $\sf D$ is the diagonal matrix whose entries are the eigenvalues
$\lambda_1(p),\dots,\lambda_d(p)$ of $\sf M$ on the diagonal. Thus (\ref{e4}) becomes
 \begin{equation*}
   \gamma^i_j(p)=-\tfrac12\,f^{il}_j\,p_{l}+\tfrac1t\,\big[{\sf S}\,
   \mathcal{X}\big(-t^2\,{\sf D}/{2}\big)\, {\sf S}^{-1}\big]^{i}_j \ ,
\end{equation*}
where
\begin{equation*}
 \mathcal{X}({\sf D})=\begin{pmatrix}
 \mathcal{X}(\lambda_1) & & \\ & \ddots & \\ & &
 \mathcal{X}(\lambda_d) \end{pmatrix}
 \ .
\end{equation*}

Let us now consider two particular examples.

\begin{example}
  Let $\mfg=\mathfrak{su}(2)$ with the Lie-Poisson structure
  \begin{align*}
\theta^{ij}(x) = 2\,\varepsilon^{ij}{}_k\,x^k
  \end{align*}
on $\mathfrak{su}(2)^*=\real^3$, where
$\varepsilon^{ijk}$ is the Levi-Civita symbol in three dimensions, and
we used the standard Euclidean inner product on $\real^3$ to raise and lower indices: $\varepsilon^{ij}{}_k:=\varepsilon^{ijl}\,\delta_{lk}$;
the factor of~$2$ is just for convenience. In this case the generalized Bopp shift \eqref{eq:Boppshift} is given by~\cite{Kupriyanov:2015uxa}
\begin{align*}
\pi_\theta(x,p)^i = x^i - t\,\varepsilon^{ij}{}_k\,p_j\,x^k + t^2\,\chi\big(t^2\,|p|^2\big)\,\big(x^i\,|p|^2-p^i\,p_j\,x^j\big) \ ,
\end{align*}
from which one calculates 
\begin{equation*}
\gamma^i_j(p)= -
\varepsilon_j{}^{ik}\,p_k + t\, \chi\big(t^2\,|p|^2\big) \,\big( |p|^2\,
\delta^i_j- p_j\,p^i\big) \ ,
\end{equation*}
where
\begin{equation*}
\chi(u)=\tfrac1u\,\big(\sqrt{u}\cot\sqrt{u}-1\big) \qquad \mbox{with}
\quad \chi(0) = -\tfrac13 \ ,
\end{equation*}
and $|p|^2:=\delta^{ij}\,p_i\,p_j$.

According to (\ref{gtA}), the corresponding deformation of abelian
gauge transformations yields the Poisson gauge transformations~\cite{Kupriyanov:2019ezf}
\begin{equation*}
 \delta_{f}^\theta A=\dd f+t\, x\cdot (\nabla A\times\nabla f)+t\,A\times\nabla
 f+t^2\, \chi\big(t^2\,|A|^2\big)\,\big(|A|^2\,\dd f- (A\cdot\nabla
 f)\,A\big) \ ,
\end{equation*}
where $A\times\nabla f:=\ast(A\wedge\dd f)$, with $\ast$ the Hodge duality operator on
the Euclidean vector space $\real^3$, while
$|A|^2:=\delta^{ij}\,A_i\,A_j$ and $A\cdot\nabla
f:=\delta^{ij}\,A_i\,\partial_j f$; we also abbreviated
$\{A,f\}_\theta=: x\cdot (\nabla A\times\nabla f)$ at $x\in\real^3$. This may be verified directly to
close the Poisson gauge algebra (\ref{closure}) and to have the
correct deformation property \eqref{initial}. It encodes the gauge symmetry
of rotationally invariant Poisson gauge theories, which by
Proposition~\ref{prop:PoissonLinfty} is generated 
by the action of the $P_\infty$-algebra with infinitely many non-vanishing brackets in coincident
gauge field entries given by
\begin{align*}
  \ell^{\,\theta}_1(f) &= \dd f \ , \\[4pt]
  \ell^{\,\theta}_2(f,g) &= -x\cdot(\nabla f \times \nabla g) \ , \\[4pt]
  \ell^{\,\theta}_2(f,A) &= x\cdot (\nabla A\times\nabla f)+ A\times\nabla f \ ,
  \\[4pt]
  \ell^{\,\theta}_{2n+1}\big(f,A^{\otimes 2n}\big) &= B_{2n}\,\big(|A|^2\big)^{n-1}
                                          \, \big(|A|^2\,\dd f - (A\cdot\nabla f)\,
                  A\big) \ , 
\end{align*}
for $n\geq1$.
\qen\end{example}

\begin{example}
Let $\mfg$ be the $d$-dimensional $\kappa$-Minkowski algebra which
yields the Kirillov-Kostant structure
\begin{equation*}
\theta_a^{ij}(x)=2\,\big(a^i\,x^j-a^j\,x^i\big) \ ,
\end{equation*}
parameterized by a fixed constant vector $(a^i)\in\real^d$. For this Poisson structure one finds \cite{KKV}
\begin{equation*}
\gamma^i_j(p)= \big(\tfrac1t\,\sqrt{1+t^2\,\langle a, p\rangle^2}-\tfrac1t+\langle a,
p\rangle \big) \,
\delta^i_j -a^i\,p_j \ ,
\end{equation*}
where $\langle a, p\rangle :=a^i\,p_i$ is the usual pairing between
$\real^d=\mfg^*$ and $\mfg$. 
The corresponding deformation of abelian gauge transformations
becomes the Poisson gauge transformations
\begin{equation*}
 \delta_{f}^{\theta_a} A=\big(\sqrt{1+t^2\,(\iota_aA)^2}+t\,\iota_a A\big) \, \dd
 f + t\,\{A,f\}_{\theta_a} - t\,(\iota_a\dd f)\, A \ ,
\end{equation*}
where $\iota_a$ denotes interior multiplication with the constant
vector field $a^i\,\partial_i$ on $\real^d$. By expanding the square root
in its Taylor series around $t=0$ using the binomial series
\begin{align*}
\sqrt{1+u^2} = 1 - \sum_{n=1}^\infty \, \frac2n \, \binom{2n-2}{n-1}
  \, \Big(-\frac{u^2}4\Big)^n \ ,
\end{align*}
we can calculate the non-zero coincident gauge field
brackets of the corresponding $P_\infty$-algebra from
Proposition~\ref{prop:PoissonLinfty} to get
\begin{align*}
  \ell^{\,\theta_a}_1(f) &= \dd f \ , \\[4pt]
  \ell^{\,\theta_a}_2(f,g) &= -\{f,g\}_{\theta_a} \ , \\[4pt]
  \ell^{\,\theta_a}_2(f,A) &= \{A,f\}_{\theta_a} + \iota_a(A\wedge\dd
                             f) \ , \\[4pt]
  \ell^{\,\theta_a}_{2n+1}\big(f,A^{\otimes 2n}\big) &= -\frac1{2^{2n-1}} \,
                                                \frac{(2n-2)!\,(2n)!}{(n-1)!\,n!}
                                                \ (\iota_aA)^{2n} \,
                                                \dd f \ ,
\end{align*}
for $n\geq1$.
\qen\end{example}

\subsection{Poisson families}

We will now show that a rather large class of
non-linear Poisson structures
on $M=\real^2$ admit symplectic embeddings that yield Poisson gauge
algebras whose homotopy algebraic structures are given by differential graded Poisson algebras,
though in a different form than what we have encountered thus far in
this paper. The basic idea is to start from the outset with smooth
families of bivectors $\theta_t\in\mfX^2(\real^2)$,
parameterized smoothly by $t\in[0,1]$. Any such bivector defines a family of
Poisson structures on $\real^2$ by dimensional
reasons and can be written as
\begin{align}\label{e24}
\theta_t^{ij}(x) = \vartheta_t(x)\,\varepsilon^{ij} \ ,
\end{align}
where $\varepsilon^{ij}$ is the Levi-Civita symbol in two
dimensions. We assume that the smooth functions $\vartheta_t(x)$
satisfy two properties: 
\begin{itemize}
\item[(a)] $\vartheta_t(x)\neq0$ for all $x\in\real^2$ and $t\in[0,1]$; and
\item[(b)] $\vartheta_0(x)=1$ for all $x\in\real^2$.
\end{itemize}

From this we can construct a symplectic embedding of
$(\real^2,\theta_t)$ by finding a one-form $b=b_i(x)\,\dd x^i$ on
$\real^2$ whose components solve the first order differential equation
\begin{equation}\label{e25}
1+t\,\varepsilon^{ij}\,\partial_ib_j=\frac{1}{\vartheta_t(x)}
\ ,
\end{equation}
together with the Jacobian equation
\begin{equation}\label{e26}
 \det(\partial_{i}b_{j}) = 0 \ .
\end{equation}
For given $\vartheta_t(x)$ the solution of (\ref{e25}) and (\ref{e26})
was constructed in \cite{Gomes:2009tk}, where a first order classical mechanics on
the cotangent bundle of $\real^2$ was proposed whose canonical
quantization gives a quantization of Poisson algebras
with non-constant bivectors. The action functional of this one-dimensional
topological sigma-model is given by
\begin{align*}
{\mathcal S}(X,P) = \int_\real \, \langle P,\dd X\rangle + \frac
  t2\,(P+X^*b)\wedge\dd(P+X^*b) \ ,
\end{align*}
where $(X,P):\real\to T^*\real^2=\real^2\times(\real^2)^*$ are
(suitably supported) smooth maps, with $\langle\,\cdot\,,\,\cdot\,\rangle$ the 
pairing between $(\real^2)^*$ and $\real^2$. Hamiltonian reduction of the
phase space of this sigma-model by its rank~$2$ second class constraints
defines Dirac brackets which, in our language, yield a symplectic embedding of
the family \eqref{e24} with cosymplectic structure
\begin{align}\label{e29}
\omega_t^{-1} =
  \tfrac t2\,\vartheta_t(x)\, \varepsilon^{ij}\, \partial_i\wedge\partial_j +
  \tfrac12\,\gamma_t{}^i_j(x)\,
  \big(\partial_i\wedge\tilde\partial{}^j+\tilde\partial{}^j\wedge\partial_i\big)
  \ ,
\end{align}
where
\begin{align}\label{eq:gammaijx}
\gamma_t{}^i_j(x) =
  \vartheta_t(x)\,\big(\delta_j^i-t\,\varepsilon^{ik}\,
  \partial_kb_j(x)\big) \ ,
\end{align}
and the Jacobian condition \eqref{e26} ensures the Lagrangian section
condition. 
By condition (a) above this is indeed non-degenerate for all
$t\in[0,1]$, and by condition (b) it coincides with the canonical
cosymplectic structure on $T^*\real^2$ at $t=0$. In particular, it
defines a \emph{strict} deformation of the cotangent symplectic
groupoid $(T^*\real^2,\omega_0)$ for $\real^2$.

The important new feature here, compared to our previous formulation
in \eqref{PB1}, is that the matrix \eqref{eq:gammaijx} does not depend
on the transverse coordinates $p$. Regarding it as a $(1,1)$-tensor field on
$\real^2$, according to \eqref{gtA0} the deformation of abelian gauge
transformations by the Poisson family \eqref{e24} is given by the Poisson
gauge transformations
\begin{align}\label{e30}
\delta_f^{\theta_t}A = \gamma_t(\dd f,\,\cdot\,) +
  t\,\{A,f\}_{\theta_t} \ .
\end{align}
The $L_\infty$-algebra formulation of these gauge symmetries has the
remarkable property given by

\begin{proposition}\label{e31}
The Poisson gauge transformations \eqref{e30} are generated by the
action of the
family of differential graded Poisson algebras on
$C^\infty(\real^2)\oplus\Omega^1(\real^2)$ with non-vanishing brackets 
\begin{align*}
\ell_1^{\,\theta_t}(f) = \gamma_t(\dd f,\,\cdot\,) \ , \quad
  \ell_2^{\,\theta_t}(f,g) = -\{f,g\}_{\theta_t} \qquad \mbox{and}
  \qquad \ell_2^{\,\theta_t}(f,A) = \{A,f\}_{\theta_t} \ ,
\end{align*}
for $f,g\in C^\infty(\real^2)$ and $A\in\Omega^1(\real^2)$.
\end{proposition}

\begin{proof}
This is a simple consequence of the fact that the components of the
Poisson bivector \eqref{e29} do not depend on $p$, the Jacobi
identities for the Poisson brackets, and the Lagrangian zero section condition
$\{p_i,p_j\}_{\omega_t^{-1}}=0$. In coordinates
\begin{align*}
\ell^{\,\theta_t}_1(f) = \{f,p_i\}_{\omega_t^{-1}}\,\dd x^i \ ,
\end{align*}
and the differential condition $\big(\ell_1^{\,\theta_t}\big)^2=0$ follows
from the Jacobi identity for the Poisson bracket
$\{\,\cdot\,,\,\cdot\,\}_{\omega_t^{-1}}$ along with
$\{p_i,p_j\}_{\omega_t^{-1}}=0$. Similarly, the derivation property
\begin{align*}
  \ell_1^{\,\theta_t}\{f,g\}_{\theta_t} =
  \big\{\ell_1^{\,\theta_t}(f),g\big\}_{\theta_t} +
  \big\{f,\ell^{\,\theta_t}_1(g)\big\}_{\theta_t}
\end{align*}
holds as a consequence of the Jacobi identity for the bracket
$\{\,\cdot\,,\,\cdot\,\}_{\omega_t^{-1}}$. The homotopy Jacobi
identity $\CJ_3=0$ is simply the graded Jacobi identity for the Lie bracket
$\{\,\cdot\,,\,\cdot\,\}_{\theta_t}$, and the derivation properties of the brackets are clear.
\end{proof}

\begin{remark}
Proposition~\ref{e31} implies that there is no need to introduce higher
brackets in the $L_\infty$-structure on the gauge algebra in this
case, provided that one ``twists'' the differential, which in our
previous treatments always coincided with the exterior derivative
$\dd$, in such a way that the new differential is a derivation of the
Poisson algebra corresponding to \eqref{e24}. This generalizes 
Example~\ref{ex:Linftyconstant} for constant Poisson structures, recovered here in the case that $\vartheta_t(x)=1$ for all
$x\in\real^2$ and $t\in[0,1]$, for which $b=0$. The ``twisting'' here is
captured automatically by the symplectic embedding approach to the
deformation of gauge transformations that we developed in the present
paper. 
\qen\end{remark}

\begin{example}
Consider the family of rotationally symmetric Poisson structures
\eqref{e24} with 
\begin{align*}
\vartheta_t(x) = \frac1{1+t\,|x|^2} \ ,
\end{align*}
where $|\cdot|$ is the standard Euclidean norm. In this case, one
finds
\begin{align*}
b_i(x) = -\tfrac14\,|x|^2\,\varepsilon_{ij}\,x^j
\end{align*}
as a solution to \eqref{e25} and \eqref{e26}. The twisted differential
from Proposition~\ref{e31} then has components
\begin{align*}
\ell_1^{\,\theta_t}(f)_i = \frac1{1+t\,|x|^2} \, \Big( \big(1+\tfrac t4\,|x|^2\big)\,\partial_if +
  \tfrac
  t2\,\varepsilon_{ik}\,x^k\,x^l\,\varepsilon_l{}^j\,\partial_jf \Big)
  \ ,
\end{align*}
for $i=1,2$.
\qen\end{example}

\subsection{Magnetic Poisson structures}
\label{sec:magneticPoisson}

Let $M=T^*Q$ be the cotangent bundle of a $d$-dimensional manifold
$Q$, with bundle projection $\varpi:M\to Q$. The manifold $M$ is a symplectic manifold with
canonical symplectic two-form which we denote by $\sigma_0$. We write local
coordinates on $M$ as $(x^i)=(q^a,q^*_a)$, with $a=1,\dots,d$ and
$i=1,\dots,2d$, where $(q^a)$ are local coordinates on $Q$ and
$(q^*_a)$ are canonically conjugate coordinates in the normal
directions to the zero section $Q\subset M$. We denote the
corresponding derivatives by $(\partial_i)=(\partial_a,\partial_*^a)$,
where $\partial_a=\partial/\partial q^a$ and
$\partial_*^a=\partial/\partial q^*_a$.

Let $B\in\Omega^2(Q)$ be an arbitrary two-form on the base manifold
  $Q$. Its pullback to $M$ deforms the symplectic structure $\sigma_0$ to an almost
  symplectic form
  \begin{align*}
\sigma_B = \sigma_0-\varpi^*B
  \end{align*}
which is closed if and only if $B$ is a closed two-form on $Q$. The
inverse $\theta_B=\sigma_B^{-1}$ is an $H$-twisted Poisson structure
on $M$, with twisting three-form $H\in\Omega^3(M)$ given by
\begin{align*}
H = \varpi^*\dd B \ .
\end{align*}
The bivector $\theta_B$ defines a \emph{magnetic Poisson structure} on
$M$; the terminology comes from monopole physics where the
two-form $B$ plays the role of a magnetic field.
In local coordinates, where $B=\frac12\,B_{ab}(q) \, \dd q^a\wedge\dd
q^b$ and $H=\frac1{3!}\,H_{abc}(q)\,\dd q^a\wedge\dd q^b\wedge\dd
q^c$, the bivector $\theta_B$ reads 
\begin{align*}
\theta_B = \tfrac12\,(\partial_a\wedge
  \partial_*^a+\partial_*^a\wedge\partial_a) + \tfrac12\,B_{ab}(q)\,
  \partial_*^a\wedge\partial_*^b \ ,
\end{align*}
and the Jacobiator
\begin{align*}
  \Pim_B=[\theta_B,\theta_B]=\tfrac1{3!}\,
  H_{abc}(q)\,\partial_*^a\wedge\partial_*^b\wedge\partial_*^c 
\end{align*}
has non-vanishing components
only in the normal directions to the zero section $Q\subset M$.

Deformation quantization of such twisted Poisson manifolds was originally considered
in~\cite{Mylonas2012} using Kontsevich's formalism. Their higher
geometric quantization was developed in~\cite{Bunk:2018qvk} by regarding the
three-form $H$ as the curvature of a trivial gerbe on $M$, and the
extension to non-trivial gerbes is discussed in~\cite{Bunk:2020rju};
see~\cite{Szabo:2019gvu} for a review of the different perspectives to
quantization of magnetic Poisson structures. In the language of
Remark~\ref{rem:twistedB}, the almost cosymplectic structure
$\sigma_B^{-1}$ is equivalent to the canonical cosymplectic structure
$\sigma_0^{-1}$ by means of a $B$-transformation, which leads to many simplifications in its symplectic
embedding construction: many components of the tensors $\gamma^{(n)}$ and
$\theta^{(n)}$ of the symplectic embedding vanish as various
combinations of derivatives acting on the components of $\theta_B$ and
$\Pim_B$, which are functions on the base manifold $Q$, are identically zero; in particular
$\Pim_B^{ijk}\,\partial_k\theta_B^{lm}=0$, see e.g.~\cite{KS17}.

\paragraph{Linear magnetic Poisson structures.}

We specialise now to
$Q=\real^d$ with the \emph{linear} magnetic Poisson structure defined by the two-form
\begin{align*}
B_{ab}(q) = \tfrac12\,H_{abc}\,q^a \ ,
\end{align*}
where $H$ is a constant three-form on $\real^d$. This is
the twisted Poisson analog of a constant Poisson structure: In this
case the components of the Jacobiator $\Pim_B$ are constant and the generalized Bopp shift \eqref{eq:Boppshift} reads
\begin{align*}
\pi_{\theta_B}(x,p)^i = x^i-\tfrac t2\,\theta_B^{ij}(x)\,p_j \ ,
\end{align*}
so $\gamma^{(n)}=0$ and $\theta^{(n)}=0$ for all
$n\geq2$~\cite{KS18}. We denote the fibre coordinates of $T^*M$ by
$(p_i)=(\xi_a,\xi_*^a)$. In the notation of Sections~\ref{sec:Poissonembedding}
and~\ref{sec:quasiPoisson}, the non-zero components of the symplectic
embedding are then given by 
\begin{align*}
\gamma_{ab}(\xi_*) = -\tfrac12\,H_{abc}\,\xi_*^c \ , \quad
  \underline{\theta}\,^a_b=-\underline{\theta}\,_{a}^b = \delta^a_b
  \qquad \mbox{and} \qquad \underline{\theta}_{\,ab}(q,\xi_*) =
  H_{abc}\,(q^c-t\,\xi_*^c) \ .
\end{align*}

For the corresponding deformation of abelian gauge transformations, we
note that the tensors $\Upm^{(n)}$ determined from Corollary~\ref{cor:LambdamA} also vanish for all
$n\geq2$, so that in this case $\Lambdam^{ij}(A) =
-\tfrac12\,\Pim_B^{ijk}\,A_k$. This depends only on the transverse
components of the gauge fields $A\in\Omega^1(M)$ to the zero section
$Q\subset M$, so we decompose gauge fields as
\begin{align}\label{eq:Aqdecomp}
  A = A_i(x)\, \dd x^i = \alpha_a(q,q^*)\, \dd q^a + \alpha^a_*(q,q^*)\, \dd q_a^*
\end{align}
and obtain explicitly
\begin{align*}
\Lambdam_{ab}(\alpha_*) = -\tfrac12\,H_{abc}\,\alpha_*^c
\end{align*}
as the only non-zero components of $\Lambdam^{ij}(A)$.
The Lagrangian multipliers $L_f$ are constructed from
Proposition~\ref{p2} as formal power series
which are now given explicitly to all orders by the constant Jacobiator
$\Pim_B$ on $M$. We can write them in terms of an analytic
function on the jet space $J^1T^*M$ whose Taylor series around
$t=0$ coincides with the asymptotic series \eqref{m5}. For this, we
introduce the $2d{\times}2d$ matrix ${\sf M}$ with elements
\begin{align*}
  {\sf M}^{a}_b(\alpha_*,\partial_* \alpha_*)=-H_{bce}\,\alpha_*^c\,\partial_*^a\alpha_*^e
\end{align*}
and observe that the non-vanishing components in \eqref{m5} can be rewritten as
\begin{align}\label{eq:Labsum}
L_{ab}(\alpha_*,\partial_* \alpha_*) = -\frac12\,H_{ace}\,\alpha_*^e \ 
  \sum_{n=0}^\infty \, \Big(\frac{t^2}2\Big)^n \, \big({\sf
  M}^n\big)^{c}_{b} = -\frac12\,H_{ace}\,\alpha_*^e\,\big[\big(\one -
  \tfrac{t^2}2 \,{\sf M}\big)^{-1}\big]_{b}^{c} \ .
\end{align}
The Lagrangian multipliers 
\begin{align*}
L_f={L_f}_a\,\partial_*^a \qquad \mbox{with} \quad {L_f}_a = t^2\, L_{ab}(\alpha_*,\partial_*\alpha_*) \, \partial_*^bf
\end{align*}
are then transverse to $Q\subset M$, depend only on the transverse
components of the gauge fields and their normal derivatives, and are
determined entirely by the normal derivatives of gauge parameters
$f\in C^\infty(M)$. 

From Remark~\ref{rem:quasiPoissongauge}, the corresponding deformation
of abelian gauge transformations is given by the almost Poisson gauge
transformations
\begin{align*}
\delta_f^{\theta_B} A = \delta_f^{\theta_B}\alpha_a \, \dd q^a +
  \delta_f^{\theta_B}\alpha^a_* \, \dd q^*_a \ ,
\end{align*}
where
\begin{align}
\delta_f^{\theta_B}\alpha_a &=
                              \partial_af+t\,\{\alpha_a,f\}_{\theta_B}
                              + \tfrac t2\,
                              H_{abc}\,\alpha_*^c\,\partial_*^bf +
                              t^2\,H_{bce}\,\alpha_*^c\,\partial_*^ef\,\partial_*^b\alpha_a
  \nonumber \\
  & \quad \, \hspace{1cm} +\tfrac{t^2}2\, H_{bln}\,\alpha_*^n\,\partial_*^kf\, \big[\big(\one -
  \tfrac{t^2}2 \,{\sf M}\big)^{-1}\big]_{k}^{l} \,
    \big(\partial_a\alpha_*^b-\partial_*^b\alpha_a + \tfrac t2\,
    H_{ace}\,\alpha_*^e\,\partial_*^c\alpha_*^b
    +t\,\{\alpha_a,\alpha_*^b\}_{\theta_B} \nonumber \\
  & \quad \, \hspace{7cm} + t^2\,
    H_{cem}\,\alpha_*^m\,\partial_*^c\alpha_*^b\,\partial_*^e\alpha_a
    \big) \ , \nonumber \\[4pt]
\delta_f^{\theta_B}\alpha_*^a &= \partial_*^af +
                                t\,\{\alpha_*^a,f\}_{\theta_B} +
                                t^2\,H_{bce}\,\alpha_*^c\,\partial_*^ef\,\partial_*^b\alpha_*^a
  \nonumber \\
  & \quad \, \hspace{1cm}+\tfrac{t^2}2\, H_{bln}\,\alpha_*^n\,\partial_*^kf\, \big[\big(\one -
  \tfrac{t^2}2 \,{\sf M}\big)^{-1}\big]_{k}^{l} \,
    \big(\partial_*^a\alpha_*^b-\partial_*^b\alpha_*^a
    +t\,\{\alpha_*^a,\alpha_*^b\}_{\theta_B} \nonumber\\
  & \quad \, \hspace{7cm} +
    t^2\,H_{cem}\,\alpha_*^m\,\partial_*^c\alpha_*^b\,\partial_*^e\alpha_*^a
    \big) \ . \label{eq:deltaalpha}
\end{align}
These satisfy the almost Poisson gauge algebra \eqref{eq:qPoissonalg}
with the field dependent gauge parameter
\begin{align}
[\![f,g]\!]_{\theta_B}(\alpha_*) &= \{f,g\}_{\theta_B} -
                            t\,H_{abc}\,\alpha_*^c\,\partial_*^af\,\partial_*^bg
  \nonumber \\
  & \quad \, -\tfrac{t^3}4\,H_{bsn}\,H_{crm}\,\alpha_*^n\,\alpha_*^m\,
    \partial_*^kf\,\partial_*^pg\, \big[\big(\one -
  \tfrac{t^2}2 \,{\sf M}\big)^{-1}\big]_{k}^{s} \, \big[\big(\one -
    \tfrac{t^2}2 \,{\sf M}\big)^{-1}\big]_{p}^{r} \label{eq:fgthetaBA}
  \\
  & \quad \, \hspace{4cm} \times 
    \big(\partial_*^b\alpha_*^c-\partial_*^c\alpha_*^b + t\,
    \{\alpha_*^b,\alpha_*^c\}_{\theta_B} 
    -t^2\,H_{ael}\,\alpha_*^l\,\partial_*^a\alpha_*^b\,\partial_*^e\alpha_*^c
    \big) \ . \nonumber
\end{align}
Notice that, apart from the terms involving the almost Poisson brackets
$\{\,\cdot\,,\,\cdot\,\}_{\theta_B}$, the gauge variations of the
transverse components $\alpha_*^a$ in \eqref{eq:deltaalpha} as well as
the brackets \eqref{eq:fgthetaBA} are
determined completely by normal components to the embedding $Q\subset
M$. 

The $L_\infty$-algebra of these almost Poisson gauge symmetries is
given by Proposition~\ref{prop:quasiPoissonLinfty}, with infinitely
many non-vanishing brackets which in
this example greatly simplify due to the vanishing of most tensors in
the symplectic embedding and in the Lagrangian multipliers.  
They can be read
off directly from \eqref{eq:deltaalpha} and \eqref{eq:fgthetaBA} by reinstating
the formal power series expansion \eqref{eq:Labsum}. For the coincident
gauge field brackets involving a single gauge parameter, we use the
decomposition \eqref{eq:Aqdecomp} to similarly write
$\ell_{n+1}^{\,\theta_B}\big(f,A^{\otimes n}\big)\in\Omega^1(M)$ in
component form
\begin{align*}
\ell_{n+1}^{\,\theta_B}\big(f,A^{\otimes n}\big) =
  \lambda^{\theta_B}_{n+1}\big(f,A^{\otimes n}\big)_a \, \dd q^a +
  \lambda^{\theta_B}_{n+1}\big(f,\alpha_*^{\otimes n}\big)^a_* \, \dd q^*_a \ ,
\end{align*}
where
\begin{align*}
  \lambda_1^{\theta_B}(f)_a &= \partial_af \ , \\[4pt]
  \lambda_2^{\theta_B}(f,A)_a &= \{\alpha_a,f\}_{\theta_B} +
                                \tfrac12\,H_{abc}\,\alpha_*^c\,\partial_*^bf
                                \ , \\[4pt]
  \lambda_3^{\theta_B}\big(f,A^{\otimes2}\big)_a &=
                                                   H_{blk}\,\alpha_*^l\,\partial^k_*f
                                                   \,
                                                   \big(\partial_a\alpha_*^b-3\,\partial_*^b\alpha_a\big)
                                                   \ , \\[4pt]
  \lambda_4^{\theta_B}\big(f,A^{\otimes3}\big)_a &= 3\,
                                                   H_{blk}\,\alpha_*^l\,\partial_*^kf
                                                   \,
                                                   \big(\{\alpha_a,\alpha_*^b\}_{\theta_B}
                                                   +
                                                   \tfrac12\,H_{ace}\,\alpha_*^e\,\partial_*^c\alpha_*^b\big)
                                                   \ , \\[4pt]
  \lambda_{2n-1}^{\theta_B}\big(f,A^{\otimes2n-2}\big)_a &=
                                                         -\frac{(2n-2)!}{2^{n-1}}
                                                         \,
                                                         H_{blm}\, H_{kb_1c_1}\,H_{a_1b_2c_2}\cdots
                                                         H_{a_{n-3}b_{n-2}c_{n-2}}\,\alpha_*^m\, \alpha_*^{b_1}\cdots\alpha_*^{b_{n-2}}\,\partial_*^kf                                            
  \\
  & \quad \, \hspace{1cm} \times \,
    \partial_*^{a_1}\alpha_*^{c_1}\cdots
    \partial_*^{a_{n-3}}\alpha_*^{c_{n-3}} \,
    \partial_*^l\alpha_*^{c_{n-2}}\,\big(\partial_a\alpha_*^b-3\,\partial_*^b\alpha_a\big)
  \ , \\[4pt]
  \lambda_{2n}^{\theta_B}\big(f,A^{\otimes2n-1}\big)_a &=
                                                         -\frac{(2n-1)!}{2^{n-1}}
                                                         \,
                                                         H_{blm}\, H_{kb_1c_1}\,H_{a_1b_2c_2}\cdots
                                                         H_{a_{n-3}b_{n-2}c_{n-2}}\,\alpha_*^m\, \alpha_*^{b_1}\cdots\alpha_*^{b_{n-2}}\,\partial_*^kf\,
  \\
  & \quad \, \hspace{1cm} \times \,
    \partial_*^{a_1}\alpha_*^{c_1}\cdots\partial_*^{a_{n-3}}\alpha_*^{c_{n-3}}\,\partial_*^l\alpha_*^{c_{n-2}}
    \,
    \big(\{\alpha_a,\alpha_*^b\}_{\theta_B} + \tfrac12\,H_{ace}\,\alpha_*^e\,\partial_*^c\alpha_*^b
    \big) \ , 
\end{align*}
and
\begin{align}
\lambda_1^{\theta_B}(f)_*^a &= \partial_*^af \ , \nonumber \\[4pt]
\lambda_2^{\theta_B}(f,\alpha_*)_*^a &= \{\alpha_*^a,f\}_{\theta_B} \ ,
  \nonumber \\[4pt]
  \lambda_3^{\theta_B}\big(f,\alpha_*^{\otimes2}\big)_*^a &=
                                                            H_{blk}\,\alpha_*^l\,\partial_*^kf
                                                            \,
                                                            \big(\partial_*^a\alpha_*^b-3\,\partial_*^b\alpha_*^a\big)
                                                            \ ,
  \nonumber \\[4pt]
  \lambda_4^{\theta_B}\big(f,\alpha_*^{\otimes3}\big)_*^a &=
                                                            3\,H_{blk}\,\alpha_*^l\,\partial_*^kf
                                                            \,
                                                            \{\alpha_*^a,\alpha_*^b\}_{\theta_B}
                                                            \ , \label{eq:lambdamagnetic}\\[4pt]
\lambda_{2n-1}^{\theta_B}\big(f,\alpha_*^{\otimes2n-2}\big)_*^a &=
                                                         -\frac{(2n-2)!}{2^{n-1}}
                                                         \,
                                                         H_{blm}\, H_{kb_1c_1}\,H_{a_1b_2c_2}\cdots
                                                         H_{a_{n-3}b_{n-2}c_{n-2}}\,\alpha_*^m\, \alpha_*^{b_1}\cdots\alpha_*^{b_{n-2}}\,\partial_*^kf                                            
  \nonumber \\
  & \quad \, \hspace{1cm} \times \,
    \partial_*^{a_1}\alpha_*^{c_1}\cdots
    \partial_*^{a_{n-3}}\alpha_*^{c_{n-3}} \,
    \partial_*^l\alpha_*^{c_{n-2}}\,\big(\partial_*^a\alpha_*^b-3\,\partial_*^b\alpha_*^a\big)
  \ , \nonumber \\[4pt]
  \lambda_{2n}^{\theta_B}\big(f,\alpha_*^{\otimes2n-1}\big)_*^a &=
                                                         -\frac{(2n-1)!}{2^{n-1}}
                                                         \,
                                                         H_{blm}\, H_{kb_1c_1}\,H_{a_1b_2c_2}\cdots
                                                         H_{a_{n-3}b_{n-2}c_{n-2}}\,\alpha_*^m\, \alpha_*^{b_1}\cdots\alpha_*^{b_{n-2}}\,\partial_*^kf\,
  \nonumber \\
  & \quad \, \hspace{1cm} \times \,
    \partial_*^{a_1}\alpha_*^{c_1}\cdots\partial_*^{a_{n-3}}\alpha_*^{c_{n-3}}\,\partial_*^l\alpha_*^{c_{n-2}}
    \,
    \{\alpha_*^a,\alpha_*^b\}_{\theta_B} \ , \nonumber
\end{align}
for $n\geq3$. For the coincident gauge field brackets involving two
gauge parameters, after some further tedious calculation and simplification we obtain for the first few brackets
\begin{align}
  \ell_2^{\,\theta_B}(f,g) &= -\{f,g\}_{\theta_B} \ , \nonumber \\[4pt]
  \ell_3^{\,\theta_B}(f,g,\alpha_*) &=
                             H_{abc}\,\alpha_*^c\,\partial_*^af\,\partial_*^bg
                             \ , \nonumber \\[4pt]
  \ell_4^{\,\theta_B}\big(f,g,\alpha_*^{\otimes2}\big) &= 0 \ ,
                                                       \nonumber \\[4pt]
  \ell_5^{\,\theta_B}\big(f,g,\alpha_*^{\otimes 3}\big) &= -\tfrac32\, H_{bke}\,H_{cpm}\,\alpha_*^e\,\alpha_*^m\,
    \partial_*^kf\,\partial_*^pg\, 
    \big(\partial_*^b\alpha_*^c-\partial_*^c\alpha_*^b\big) \ ,
  \nonumber \\[4pt]
  \ell_6^{\,\theta_B}\big(f,g,\alpha_*^{\otimes 4}\big) &= -6\,H_{bke}\,H_{cpm}\,\alpha_*^e\,\alpha_*^m\,
    \partial_*^kf\,\partial_*^pg\,
                                                 \{\alpha_*^b,\alpha_*^c\}_{\theta_B}
                                                        \ , \nonumber \\[4pt]
  \ell_7^{\,\theta_B}\big(f,g,\alpha_*^{\otimes5}\big) &=
                                                       -30\,H_{bke}\,H_{cpm}\,H_{arl}\,\alpha_*^e\,\alpha_*^m\,\partial_*^kf\,\partial_*^pg
  \nonumber \\
  & \quad \, \qquad \times \,
    \big(\partial_*^a\alpha_*^b\,\partial_*^r\alpha_*^c + \tfrac12\,
    \partial_*^b\alpha_*^a\,\partial_*^r\alpha_*^c + \tfrac12\,
    \partial_*^a\alpha_*^b\,\partial_*^c\alpha_*^r\big) \ , \label{eq:ellfgmagnetic1}
\end{align}
together with the higher order brackets
\begin{align}
  \ell_{2n}^{\,\theta_B}\big(f,g,\alpha_*^{\otimes2n-2}\big) &=
                                                             \frac{(2n-2)!}{2^{n-1}}
                                                             \,
                                                             H_{bse}\,H_{crm}\,\alpha_*^e\,\alpha_*^m\,\alpha_*^{b_1}
                                                             \cdots
                                                             \alpha_*^{b_{n-3}}\,\partial_*^kf\,\partial_*^pg\,\{\alpha_*^b,\alpha_*^c\}_{\theta_B}
  \nonumber \\
  & \quad \, \times \,\Big(H_{a_1b_2c_2}\cdots
    H_{a_{n-4}b_{n-3}c_{n-3}}\,\partial_*^{a_1}\alpha_*^{c_1}\cdots\partial_*^{a_{n-4}}\alpha_*^{c_{n-4}}
  \nonumber \\
  & \quad \, \hspace{2cm} \times \, \big(\delta_k^s
    \, H_{pb_1c_1}\,\partial_*^r\alpha_*^{c_{n-3}} + \delta_p^r\,
    H_{kb_1c_1}\,\partial_*^s\alpha_*^{c_{n-3}}\big) \nonumber \\
  & \quad \, \qquad + \sum_{l=1}^{n-4} \,
    H_{kb_1c_2}\,H_{a_1b_2c_2}\cdots
    H_{a_{l-1}b_lc_l} \nonumber \\
  & \quad \, \qquad \hspace{1.5cm} \times \, H_{pb_{l+1}c_{l+1}}\,H_{a_{l+1}b_{l+2}c_{l+2}}\cdots
    H_{a_{n-4}b_{n-3}c_{n-3}} \nonumber \\
  & \quad \, \qquad \hspace{2cm} \times \,
    \partial_*^{a_1}\alpha_*^{c_1}\cdots\partial_*^{a_{l-1}}\alpha_*^{c_{l-1}}\,\partial_*^s\alpha_{*}^{c_l}
  \nonumber \\
  & \quad \, \qquad \hspace{2.5cm} \times \, \partial_*^{a_{l+1}}\alpha_*^{c_{l+1}}\cdots\partial_*^{a_{n-4}}\alpha_*^{c_{n-4}}\,\partial_*^r\alpha_*^{c_{n-3}}
    \Big) \ , \label{eq:ellfgmagnetic2}
\end{align}
and
\begin{align}
  \ell_{2n+1}^{\,\theta_B}\big(f,g,\alpha_*^{\otimes2n-1}\big) &=
                                                               -\frac{(2n-1)!}{2^{n-1}}
                                                               \,
                                                               H_{bse}\,H_{crm}\,\alpha_*^e\,\alpha_*^m\,\alpha_*^{b_1}\cdots\alpha_*^{b_{n-2}}\,\partial_*^kf\,\partial_*^lg
  \nonumber \\
  & \quad \, \times \, \bigg(H_{a_1b_2c_2}\cdots
    H_{a_{n-4}b_{n-3}c_{n-3}}\,\partial_*^{a_1}\alpha_*^{c_1}\cdots\partial_*^{a_{n-4}}\alpha_*^{c_{n-4}}
  \nonumber \\
  & \quad \, \qquad \times \, \Big(\frac12\,
    H_{a_{n-3}b_{n-2}c_{n-2}}\,\partial_*^{a_{n-3}}\alpha_*^{c_{n-3}}\,\big(\partial_*^b\alpha_*^c
    - \partial_*^c\alpha_*^b\big) \nonumber \\
  & \quad \qquad \hspace{2cm} \times \, \big(\delta_k^s\,
    H_{pb_1c_1}\,\partial_*^r\alpha_*^{c_{n-2}} + \delta_p^r\,
    H_{kb_1c_1}\,\partial_*^s\alpha_*^{c_{n-2}}\big) \nonumber \\
  & \quad \, \qquad \qquad +
    H_{alb_{n-2}}\,\partial_*^a\alpha_*^b\,\partial_*^l\alpha_*^c\,\big(\delta_k^s\,
    H_{pb_1c_1}\,\partial_*^r\alpha_*^{c_{n-3}} + \delta_p^r\,
    H_{kb_1c_1}\,\partial_*^s\alpha_*^{c_{n-3}}\big) \Big) \nonumber \\
  & \quad \, \qquad +
    \Big(\frac12\,H_{a_{n-3}b_{n-2}c_{n-2}}\,\partial_*^{a_{n-3}}\alpha_*^{c_{n-3}}\,\partial_*^r\alpha_*^{c_{n-2}}\,\big(\partial_*^b\alpha_*^c
    - \partial_*^c\alpha_*^b\big) \nonumber \\
  & \quad \, \qquad \hspace{2cm} +
    H_{alb_{n-2}}\,\partial_*^a\alpha_*^b\,\partial_*^l\alpha_*^c\,\partial_*^r\alpha_*^{c_{n-3}}\Big)
  \nonumber \\
  & \quad \, \qquad \qquad \times \, \sum_{l=1}^{n-4} \,
    H_{kb_1c_1}\,H_{a_1b_2c_2}\cdots
    H_{a_{l-1}b_lc_l} \nonumber \\
  & \quad \qquad \qquad \hspace{1.5cm} \times \, H_{pb_{l+1}c_{l+1}}\,H_{a_{l+1}b_{l+2}c_{l+2}}\cdots
    H_{a_{n-4}b_{n-3}c_{n-3}} \nonumber \\
  & \quad \qquad \qquad \hspace{1.5cm} \times
    \partial_*^{a_1}\alpha_*^{c_1}\cdots\partial_*^{a_{l-1}}\alpha_*^{c_{l-1}}\,\partial_*^s\alpha_*^{c_l}\,\partial_*^{a_{l+1}}\alpha_*^{c_{l+1}}\cdots\partial_*^{a_{n-4}}\alpha_*^{c_{n-4}}
  \nonumber \\
  & \quad \, \qquad + \frac12\,H_{kb_1c_1}\,H_{a_1b_2c_2}\cdots
    H_{a_{n-4}b_{n-3}c_{n-3}}\,H_{pb_{n-2}c_{n-2}} \nonumber \\
  & \quad \, \qquad \hspace{2cm} \times \,
    \partial_*^{a_1}\alpha_*^{c_1}\cdots\partial_*^{a_{n-4}}\alpha_*^{a_{n-4}c_{n-4}}\,\partial_*^s\alpha_*^{c_{n-3}}\,\partial_*^r\alpha_*^{c_{n-2}}\bigg)
    \ , \label{eq:ellfgmagnetic3}
\end{align}
for $n\geq4$.

\begin{remark}\label{rem:triproducts}
The forms of the transverse gauge transformations in \eqref{eq:deltaalpha} and
the brackets \eqref{eq:fgthetaBA}, as well as their corresponding
$L_\infty$-structures \eqref{eq:lambdamagnetic} and
\eqref{eq:ellfgmagnetic1}--\eqref{eq:ellfgmagnetic3}, suggest a natural truncation of
almost Poisson gauge transformations in this case to gauge parameters and gauge
fields along the normal directions to the zero section $Q\subset
M$. Let $Q^*$ be
the submanifold defined by the equations $q_a=0$ in $M$. Given $f,g\in
C^\infty(Q^*)$ and $\alpha_*\in\Omega^1(Q^*)$, the pullbacks
$\delta_f^{\theta_B}\alpha_*\big|_{Q^*}$ and
\smash{$[\![f,g]\!]_{\theta_B}(\alpha_*)\big|_{Q^*}$} eliminate only the terms
involving almost Poisson brackets $\{\,\cdot\,,\,\cdot\,\}_{\theta_B}$,
leaving a non-trivial
dependence on the three-form $H$ which deforms the standard abelian
gauge transformations on the manifold
$Q^*$. In~\cite{Aschieri:2015roa} it was shown that this pullback
operation turns the nonassociative star-product on
$C^\infty(M)[[\hbar]]$, which quantizes the
twisted Poisson structure $\theta_B$, into a sequence of `triproducts' on
$C^\infty(Q^*)[[\hbar]]$, which quantizes the
trivector $\Pim_B\in\mfX^3(Q^*)$ that 
for $d=3$ defines a Nambu-Poisson structure (of degree~$3$) on
$Q^*=\real^3$. This pullback operation can thus be thought of as
defining a `Nambu-Poisson gauge symmetry' which ought to be related to
a 3-Lie algebra (or equivalently a related Lie 2-algebra) action on
$\Omega^1(Q^*)$. Indeed, the pullback of the $L_\infty$-structure to
$Q^*$ does not involve any $2$-brackets, as
$\ell_2^{\,\theta_B}(f,\alpha_*)\big|_{Q^*}=0$ and
\smash{$\ell_2^{\,\theta_B}(f,g)\big|_{Q^*}=0$}, and should be regarded as the
homotopy algebra action underlying this higher Lie algebra gauge
symmetry. It would be interesting to work out the details and
understand the underlying structures better.
\qen
\end{remark}

\appendix

\newsection{Cosymplectic brackets on the constraint locus}
\label{app:useful}

Let $(T^*M,\omega)$ be a local symplectic embedding of an almost Poisson
manifold $(M,\theta)$, let $U\subseteq M$ be an open subset, and let
$\lambda_0=p_i\,\dd x^i$ be the Liouville one-form on $T^*U$. Let
$f,g,h\in C^\infty(U)$ and $A\in\Omega^1(U)$. In this appendix we
derive some useful identities for the cosymplectic brackets on the
constraint locus ${\sf im}(s_A)\subset T^*U$, where
$\Phi_A=(p_i-A_i(x))\,\dd x^i=0$, which we use in the main text.

We will need the relation
\begin{align}
& s_A^*\{\pi^*f,\{\pi^*g,\lambda_0\}_{\omega^{-1}}\}_{\omega^{-1}} -
  s_A^*\{\pi^*f,\pi^*s_A^*\{\pi^*g,\lambda_0\}_{\omega^{-1}}\}_{\omega^{-1}}
  \nonumber \\[4pt]
& \hspace{5cm} = t\,s_A^*\{\pi^*f,\gamma(\pi^*\dd
            g,\,\cdot\,)\}_{\omega^{-1}} - t\,
            s_A^*\{\pi^*f,\pi^*s_A^*\gamma(\dd
            g,\,\cdot\,)\}_{\omega^{-1}} \nonumber \\[4pt]
& \hspace{5cm} =
                                                   s_A^*\big(\tilde\partial^i\{\pi^*g,\lambda_0\}_{\omega^{-1}}\big)\,\big(s_A^*\{\pi^*f,p_i\}_{\omega^{-1}}
                                                   -
                                                   s_A^*\{\pi^*f,\pi^*A_i\}_{\omega^{-1}}\big)
  \nonumber \\[4pt]
& \hspace{5cm} =
            s_A^*\big(\tilde\partial^i\{\pi^*g,\lambda_0\}_{\omega^{-1}}\big)
            \, s_A^*\{\pi^*f,(\Phi_A)_i\}_{\omega^{-1}} \ .
                      \label{r1}\end{align}
We also need
\begin{align}
& s_A^*\{\pi^*f,\{\pi^*g,\pi^*h\}_{\omega^{-1}}\}_{\omega^{-1}} -
                s_A^*\{\pi^*f,\pi^*s_A^*\{\pi^*g,\pi^*h\}_{\omega^{-1}}\}_{\omega^{-1}}
                \notag \\[4pt]
& \hspace{5cm} =
                                 t\,s_A^*\{\pi^*f,\{\pi^*g,\pi^*h\}_{\underline{\theta}}\}_{\omega^{-1}}
                                 - t\,
                                 s_A^*\{\pi^*f,\pi^*s_A^*\{\pi^*g,\pi^*h\}_{\underline{\theta}}\}_{\omega^{-1}}
                                 \notag \\[4pt]
& \hspace{5cm} =
                                                  s_A^*\big(\tilde\partial^i\{\pi^*g,\pi^*h\}_{\omega^{-1}}\big)
                                                  \, \big(s_A^*\{\pi^*f,p_i\}_{\omega^{-1}}
                                                   -
                                                   s_A^*\{\pi^*f,\pi^*A_i\}_{\omega^{-1}}\big)
  \nonumber \\[4pt]
& \hspace{5cm} =
            s_A^*\big(\tilde\partial^i\{\pi^*g,\pi^*h\}_{\omega^{-1}}\big)
            \, s_A^*\{\pi^*f,(\Phi_A)_i\}_{\omega^{-1}} \ .
\label{r2}\end{align}
Now the combination of (\ref{r1}) and (\ref{r2}) implies
\begin{align}
& s_A^*\{\pi^*f,\{\pi^*g,\Phi_A\}_{\omega^{-1}}\}_{\omega^{-1}} -
                s_A^*\{\pi^*f,\pi^*s_A^*\{\pi^*g,\Phi_A\}_{\omega^{-1}}\}_{\omega^{-1}}
                \notag \\[4pt]
 & \hspace{5cm}               =
  s_A^*\big(\tilde\partial^i\{\pi^*g,\Phi_A\}_{\omega^{-1}}\big) \,
  s_A^*\{\pi^*f,(\Phi_A)_i\}_{\omega^{-1}} \ .
\label{r3}\end{align}

Next we calculate
\begin{align}
& s_A^*\{\{\pi^*g,\pi^*f\}_{\omega^{-1}},\Phi_A\}_{\omega^{-1}} -
                s_A^*\{\pi^*s_A^*\{\pi^*g,\pi^*f\}_{\omega^{-1}},\Phi_A\}_{\omega^{-1}}
                \notag \\[4pt]
& \hspace{1.5cm} =
                                 s_A^*\big(\tilde\partial^i\{\pi^*g,\pi^*f\}_{\omega^{-1}}\big)
                                 \,
                                 \big(-s_A^*\{\pi^*A_i,\lambda_0\}_{\omega^{-1}}
                                 - s_A^*\{p_i,\pi^*A\}_{\omega^{-1}} +
                                 s_A^*\{\pi^*A_i,\pi^*A\}_{\omega^{-1}}\big)
                                 \notag \\[4pt]
& \hspace{1.5cm} =
                                                  s_A^*\big(\tilde\partial^i\{\pi^*g,\pi^*f\}_{\omega^{-1}}\big)
                                                  \,
                                                  s_A^*\{(\Phi_A)_i,\Phi_A\}_{\omega^{-1}}
                                                  \ .
\label{r4}\end{align}
One may also check
\begin{align}
& s_A^*\{\{\pi^*f,(\Phi_A)_j\}_{\omega^{-1}},\Phi_A\}_{\omega^{-1}} -
  s_A^*\{\pi^*s_A^*\{\pi^*f,(\Phi_A)_j\}_{\omega^{-1}},\Phi_A\}_{\omega^{-1}}
  \notag \\[4pt]
& \hspace{5cm} =
                   s_A^*\big(\tilde\partial^i\{\pi^*f,(\Phi_A)_j\}_{\omega^{-1}}\big)
                   \, s_A^*\{(\Phi_A)_i,\Phi_A\}_{\omega^{-1}} \ ,
  \label{r5}\end{align}
and
\begin{align}
&
                s_A^*\{\{(\Phi_A)_k,(\Phi_A)_j\}_{\omega^{-1}},\Phi_A\}_{\omega^{-1}}
                -
                s_A^*\{\pi^*s_A^*\{(\Phi_A)_k,(\Phi_A)_j\}_{\omega^{-1}},\Phi_A\}_{\omega^{-1}}
                \notag \\[4pt]
& \hspace{5cm} =
                                 s_A^*\big(\tilde\partial^i\{(\Phi_A)_k,(\Phi_A)_j\}_{\omega^{-1}}\big)
                                 \,
                                 s_A^*\{(\Phi_A)_i,\Phi_A\}_{\omega^{-1}}
                                 \ .
  \label{r6}\end{align}

\newsection{Closure of almost Poisson gauge transformations}
\label{app:p1proof}

In this appendix we prove that the gauge transformations defined in
Proposition~\ref{p1} close the almost Poisson gauge algebra
\eqref{c2}. For this, we use \eqref{c20} to calculate the left-hand
side of \eqref{c2}, and after rearranging the terms we find
\begin{align}
&
                \delta^\theta_f\big(s_A^*\{\pi^*g,\Phi_A\}_{\omega^{-1}}+L_g^j\,s_A^*\{(\Phi_A)_j,\Phi_A\}_{\omega^{-1}}\big)-\delta^\theta_g\big(s_A^*\{\pi^*f,\Phi_A\}_{\omega^{-1}}+L_f^i\,s_A^*\{(\Phi_A)_i,\Phi_A\}_{\omega^{-1}}\big)
                \notag \\[4pt]
& = 
                -s_A^*\{\pi^*g+L_g^j\,(\Phi_A)_j,\pi^*s_A^*\{\pi^*f,\Phi_A\}_{\omega^{-1}}+L_f^i\,\pi^*s_A^*\{(\Phi_A)_i,\Phi_A\}_{\omega^{-1}}\}_{\omega^{-1}}\notag
  \\
& \hspace{0.5cm}
                                                                                                                                                        -s_A^*\big(\tilde\partial{}^k\{\pi^*f,\Phi_A\}_{\omega^{-1}}\big)\,\big(s_A^*\{\pi^*g,(\Phi_A)_k\}_{\omega^{-1}}+L_g^j\,s_A^*\{(\Phi_A)_j,(\Phi_A)_k\}_{\omega^{-1}}\big)
                                                                                                                                                        \notag
  \\
& \hspace{1cm}
       -L_f^i\,s_A^*\big(\tilde\partial{}^k\{(\Phi_A)_i,\Phi_A\}_{\omega^{-1}}\big)\,\big(s_A^*\{\pi^*g,(\Phi_A)_k\}_{\omega^{-1}}+L_g^j\,s_A^*\{(\Phi_A)_j,(\Phi_A)_k\}_{\omega^{-1}}\big)
       \notag \\
& \hspace{1.5cm}
                   +s_A^*\{\pi^*f+L_f^i\,(\Phi_A)_i,\pi^*s_A^*\{\pi^*g,\Phi_A\}_{\omega^{-1}}+L_g^j\,\pi^*s_A^*\{(\Phi_A)_j,\Phi_A\}_{\omega^{-1}}\}_{\omega^{-1}}\notag\\
& \hspace{2cm}
                                                                                                                                                                             +s_A^*\big(\tilde\partial{}^k\{\pi^*g,\Phi_A\}_{\omega^{-1}}\big)\,\big(s_A^*\{\pi^*f,(\Phi_A)_k\}_{\omega^{-1}}+L_f^i\,s_A^*\{(\Phi_A)_i,(\Phi_A)_k\}_{\omega^{-1}}\big)
                                                                                                                                                                             \notag
  \\
& \hspace{2.5cm}
       +L_g^j\,s_A^*\big(\tilde\partial{}^k\{(\Phi_A)_j,\Phi_A\}_{\omega^{-1}}\big)\,\big(s_A^*\{\pi^*f,(\Phi_A)_k\}_{\omega^{-1}}+L_f^i\,s_A^*\{(\Phi_A)_i,(\Phi_A)_k\}_{\omega^{-1}}\big)  \notag\\
& \hspace{3cm} -L_g^j\,s_A^*\{\pi^*s_A^*\{\pi^*f,(\Phi_A)_j\}_{\omega^{-1}} +
                                                                                                                                                                                                        L_f^i\,\pi^*s_A^*\{(\Phi_A)_i,(\Phi_A)_j\}_{\omega^{-1}},\Phi_A\}_{\omega^{-1}}
                                                                                                                                                                                                        \notag
  \\
& \hspace{3.5cm} +L_f^i\,s_A^*\{\pi^*s_A^*\{\pi^*g,(\Phi_A)_i\}_{\omega^{-1}} +
       L_g^j\,\pi^*s_A^*\{(\Phi_A)_j,(\Phi_A)_i\}_{\omega^{-1}},\Phi_A\}_{\omega^{-1}}
       \notag \\
& \hspace{4cm} +\big(\delta_f^\theta L_g^i - \delta_g^\theta
                   L_f^i\big)\,s_A^*\{(\Phi_A)_i,\Phi_A\}_{\omega^{-1}} \ .
\label{c6}\end{align}
By explicit calculation and using the formulas \eqref{r3}, \eqref{r5}
and \eqref{r6} from Appendix~\ref{app:useful} one may check
\begin{align*}
& s_A^*\{\pi^*f + L_f^i\,(\Phi_A)_i,\{\pi^*g +
                 L_g^j\,(\Phi_A)_j,\Phi_A\}_{\omega^{-1}}\}_{\omega^{-1}} \\
& \hspace{6cm} - s_A^*\{\pi^*f + L_f^i\,(\Phi_A)_i,\pi^*s_A^*\{\pi^*g +
                 L_g^j\,(\Phi_A)_j,\Phi_A\}_{\omega^{-1}}\}_{\omega^{-1}} \\[4pt]
& \hspace{1cm} =
                                                                                    \big(s_A^*(\tilde\partial{}^k\{\pi^*g,\Phi_A\}_{\omega^{-1}})
                                                                                    +
                                                                                    L_g^j\,s_A^*\big(\tilde\partial{}^k\{(\Phi_A)_j,\Phi_A\}_{\omega^{-1}}
                                                                                    +
                                                                                    s_A^*\{L_g^k,\Phi_A\}_{\omega^{-1}}
                                                                                    \big)
  \\
  & \hspace{5cm} \times \
    \big(s_A^*\{\pi^*f,(\Phi_A)_k\}_{\omega^{-1}}+L_f^i\,s_A^*\{(\Phi_A)_i,(\Phi_A)_k\}_{\omega^{-1}}\big)
    \ .
\end{align*}
Using this expression we rewrite (\ref{c6}) as
\begin{align*}
&
                 s_A^*\{\pi^*f+L_f^i\,(\Phi_A)_i,\{\pi^*g+L_g^j\,(\Phi_A)_j,\Phi_A\}_{\omega^{-1}}\}_{\omega^{-1}}
  \\
& \hspace{6cm} 
  -
  s_A^*\{\pi^*g+L_g^i\,(\Phi_A)_i,\{\pi^*f+L_f^j\,(\Phi_A)_j,\Phi_A\}_{\omega^{-1}}\}_{\omega^{-1}}
  \\
& \hspace{1cm} + s_A^*\{L_f^i,\Phi_A\}_{\omega^{-1}} \,
       s_A^*\{\pi^*g,(\Phi_A)_i\}_{\omega^{-1}} -
       s_A^*\{L_g^i,\Phi_A\}_{\omega^{-1}} \,
       s_A^*\{\pi^*f,(\Phi_A)_i\}_{\omega^{-1}} \\
& \hspace{1cm} +
       L_f^i\,s_A^*\{\pi^*s_A^*\{\pi^*g,(\Phi_A)_i\}_{\omega^{-1}},\Phi_A\}_{\omega^{-1}}
       -
       L_g^i\,s_A^*\{\pi^*s_A^*\{\pi^*f,(\Phi_A)_i\}_{\omega^{-1}},\Phi_A\}_{\omega^{-1}}
  \\
& \hspace{1cm} + 2\, \big( L_f^i\,s_A^*\{L_g^j,\Phi_A\}_{\omega^{-1}}
       - L_g^i\,s_A^*\{L_f^j,\Phi_A\}_{\omega^{-1}} \big) \,
       s_A^*\{(\Phi_A)_j,(\Phi_A)_i\}_{\omega^{-1}} \\
& \hspace{1cm} + 2\,
                                                         L_f^i\,L_g^j\,s_A^*\{\pi^*s_A^*\{(\Phi_A)_j,(\Phi_A)_i\}_{\omega^{-1}},\Phi_A\}_{\omega^{-1}}
                                                         +\big(\delta_f^\theta L_g^i - \delta_g^\theta
                   L_f^i\big)\,s_A^*\{(\Phi_A)_i,\Phi_A\}_{\omega^{-1}}
       \ .
\end{align*}
Using the Jacobi identity in the first two lines and simplifying the
remaining lines, we rewrite this expression as
\begin{align*}
& s_A^*\{\{\pi^*f+L_f^i\,(\Phi_A)_i,\pi^*g+L_g^j\,(\Phi_A)_j\}_{\omega^{-1}},\Phi_A\}_{\omega^{-1}} +\big(\delta_f^\theta L_g^i - \delta_g^\theta
                   L_f^i\big)\,s_A^*\{(\Phi_A)_i,\Phi_A\}_{\omega^{-1}}
  \\
& \hspace{2cm} - s_A^*\{L_g^j\,\pi^*s_A^*\{\pi^*f,(\Phi_A)_j\}_{\omega^{-1}} +
       L_f^i\,\pi^*s_A^*\{(\Phi_A)_i,\pi^*g\}_{\omega^{-1}},\Phi_A\}_{\omega^{-1}} \\
& \hspace{4cm} +
       2\,s_A^*\{L_f^i\,L_g^j\,\pi^*s_A^*\{(\Phi_A)_i,(\Phi_A)_j\}_{\omega^{-1}},\Phi_A\}_{\omega^{-1}}
       \ .
\end{align*}
Using $\Phi_A=0$ on the constraint locus ${\sf im}(s_A)$, we rewrite the first line so
that this expression becomes
\begin{align*}
& s_A^*\{\{\pi^*f,\pi^*g\}_{\omega^{-1}} {+}
  L_g^j\,\{\pi^*f,(\Phi_A)_j\}_{\omega^{-1}} {+}
                 L_f^i\,\{(\Phi_A)_i,\pi^*g\}_{\omega^{-1}} {+}
  L_f^i\,L_g^j\,\{(\Phi_A)_i,(\Phi_A)_j\}_{\omega^{-1}},\Phi_A\}_{\omega^{-1}}
  \\
& \quad +\big(s_A^*\{L_f^k,\pi^*g\}_{\omega^{-1}} +
       s_A^*\{\pi^*f,L_g^k\}_{\omega^{-1}} +
       L_f^i\,s_A^*\{(\Phi_A)_i,L_g^k\}_{\omega^{-1}} \\
& \hspace{2cm} +
       L_g^j\,s_A^*\{L_f^k,(\Phi_A)_j\}_{\omega^{-1}}\big) \,
       s_A^*\{(\Phi_A)_k,\Phi_A\}_{\omega^{-1}} +\big(\delta_f^\theta L_g^i - \delta_g^\theta
                   L_f^i\big)\,s_A^*\{(\Phi_A)_i,\Phi_A\}_{\omega^{-1}}
  \\
& \hspace{4cm} - s_A^*\{L_g^j\,\pi^*s_A^*\{\pi^*f,(\Phi_A)_j\}_{\omega^{-1}} +
       L_f^i\,\pi^*s_A^*\{(\Phi_A)_i,\pi^*g\}_{\omega^{-1}},\Phi_A\}_{\omega^{-1}} \\
& \hspace{6cm} +
       2\,s_A^*\{L_f^i\,L_g^j\,\pi^*s_A^*\{(\Phi_A)_i,(\Phi_A)_j\}_{\omega^{-1}},\Phi_A\}_{\omega^{-1}}
       \ .
\end{align*}
Applying again the formulas \eqref{r4}--\eqref{r6} from
Appendix~\ref{app:useful} in the first line of this
expression, after simplification we end up with
\begin{align*}
& s_A^*\{\pi^*s_A^*\{\pi^*f,\pi^*g\}_{\omega^{-1}} -
  L_f^i\,L_g^j\,\pi^*s_A^*\{(\Phi_A)_i,(\Phi_A)_j\}_{\omega^{-1}},\Phi_A\}_{\omega^{-1}}
  \\
& \quad + \Big(  \delta_f^\theta L_g^k - \delta_g^\theta
                   L_f^k +
       s_A^*\big(\tilde\partial^k\{\pi^*f,\pi^*g\}_{\omega^{-1}}\big) +
       L_g^j\,s_A^*\big(\tilde\partial^k\{\pi^*f,(\Phi_A)_j\}_{\omega^{-1}}\big)
  \\ 
& \hspace{1cm} +
       L_f^i\,s_A^*\big(\tilde\partial^k\{(\Phi_A)_i,\pi^*g\}_{\omega^{-1}}\big)
        +
       L_f^i\,L_g^j\,s_A^*\big(\tilde\partial^k\{(\Phi_A)_i,(\Phi_A)_j\}_{\omega^{-1}}\big)
       + s_A^*\{L_f^k,\pi^*g\}_{\omega^{-1}} \\
& \hspace{1.5cm} +
       s_A^*\{\pi^*f,L_g^k\}_{\omega^{-1}} + L_f^i\,s_A^*\{(\Phi_A)_i,L_g^k\}_{\omega^{-1}} +
                                           L_g^j\,s_A^*\{L_f^k,(\Phi_A)_j\}_{\omega^{-1}}
                                           \Big) \,
                                           s_A^*\{(\Phi_A)_k,\Phi_A\}_{\omega^{-1}}
                                           \ .
\end{align*}
By \eqref{l3} and \eqref{c4} this is equal to
$s_A^*\{t\,[\![f,g]\!]_\theta(A) + L^k_{t\,
  [\![f,g]\!]_\theta(A)}\,(\Phi_A)_k,\Phi_A\}_{\omega^{-1}}$, 
which proves the closure formula (\ref{c2}).

\end{document}